\documentclass[11pt]{article}
\usepackage{library}

\newcommand{\rev}{\mathrm{rev}}

\newcommand{\Res}{\mathrm{Res}}
\newcommand{\Disc}{\mathrm{Disc}}
\newcommand{\Esym}{\mathsf{Esym}}
\newcommand{\Trunc}{\hspace{4pt}\mathsf{Trunc}\hspace{4pt}}
\newcommand{\size}{\mathrm{size}}
\newcommand{\depth}{\mathrm{depth}}
\newcommand{\Syl}{\mathrm{Syl}}
\newcommand{\Char}{\mathrm{char}}

\newcommand{\N}{\mathbb N}

\newcommand{\C}{\mathbb C}
\newcommand{\K}{\mathbb K}
\newcommand{\F}{\mathbb F}
\newcommand{\Q}{\mathbb Q}

   \newcommand{\cB}{{\mathcal{B}}}
\newcommand{\cC}{{\mathcal{C}}}   \newcommand{\cD}{{\mathcal{D}}}
   
\newcommand{\cG}{{\mathcal{G}}}   
   
   \newcommand{\cL}{{\mathcal{L}}}
   
\newcommand{\cO}{{\mathcal{O}}}   
   
   \newcommand{\cT}{{\mathcal{T}}}

\newcommand{\bZ}{\mathbf{Z}}

\newcommand{\hf}{{\hat{f}}}
\newcommand{\hg}{{\hat{g}}}

\newcommand{\tf}{{\tilde{f}}}
\newcommand{\tg}{{\tilde{g}}}

\newcommand{\abs}[1]{\left| {#1} \right|}

\usepackage{amsthm}
\usepackage{aliascnt}
\usepackage{hyperref}
\usepackage{cleveref}

\newtheorem{theorem}{Theorem}[section]

\newaliascnt{lemma}{theorem}
\newtheorem{lemma}[lemma]{Lemma}
\aliascntresetthe{lemma}

\newaliascnt{proposition}{theorem}

\aliascntresetthe{proposition}

\newaliascnt{corollary}{theorem}
\newtheorem{corollary}[corollary]{Corollary}
\aliascntresetthe{corollary}

\theoremstyle{definition}

\newaliascnt{definition}{theorem}
\newtheorem{definition}[definition]{Definition}
\aliascntresetthe{definition}

\newaliascnt{remark}{theorem}
\newtheorem{remark}[remark]{Remark}
\aliascntresetthe{remark}

\newaliascnt{example}{theorem}

\aliascntresetthe{example}

\newaliascnt{claim}{theorem}
\newtheorem{claim}[claim]{Claim}
\aliascntresetthe{claim}

\crefname{theorem}{Theorem}{Theorems}
\crefname{lemma}{Lemma}{Lemmas}
\crefname{proposition}{Proposition}{Propositions}
\crefname{corollary}{Corollary}{Corollaries}
\crefname{definition}{Definition}{Definitions}
\crefname{remark}{Remark}{Remarks}
\crefname{example}{Example}{Examples}
\crefname{claim}{Claim}{Claims}

\numberwithin{equation}{section}

\DeclareMathOperator{\Diag}{Diag}

\newcommand{\X}{\mathbf{X}}
\def\poly{\operatorname{poly}}
\newcommand{\0}{\bm 0}
\newcommand{\balpha}{\bm \alpha}

\RestyleAlgo{ruled}

\title{Deterministic Algorithms for Low Individual Degree Factors of Sparse Polynomials }
\author{Somnath Bhattacharjee\thanks{University of Toronto, Canada. Email: \texttt{somnath.bhattacharjee@mail.utoronto.ca}. Research partially supported
by an NSERC Discovery Grant.}\and Rishabh Kothary\thanks{University of Toronto, Canada. Email: \texttt{rishabh.kothary@mail.utoronto.ca}.  Research partially supported
by an NSERC Discovery Grant.}\and Shanthanu S. Rai\thanks{Tata Institute of Fundamental Research, Mumbai, India. Email: \texttt{shanthanu.rai@tifr.res.in}. Research supported by the Department of Atomic Energy, Government of India, under project number RTI4014, TCS Research Fellowship and in part by research grants from Google Research, Premji Invest and SERB.}\and Shubhangi Saraf\thanks{University of Toronto, Canada. Email: \texttt{shubhangi.saraf@utoronto.ca}.   Research partially supported
by an NSERC Discovery Grant and a McLean Award.}}
\date{}

\begin{document}
\maketitle
\begin{abstract}
    We study factoring algorithms for general sparse polynomials and sparse polynomials of bounded individual degree and prove the following results.

\begin{enumerate}
\item
We give a deterministic polynomial-time algorithm which takes as input an $n$-variate $s$-sparse polynomial
$f$ of bounded individual degree $d$ and outputs a list $\cL$ of circuits which
contains all factors of $f$, although there might be additional spurious circuits in the list. The algorithm runs in time $\poly(n, s^d)$. Additionally, we can ensure that every circuit in the list has constant depth. Our algorithm works over all fields of characteristic 0 or sufficiently large characteristic.

Our result generalizes a recent result of Chuyoon and Shpilka~\cite{chuyoon2026factorization} that gives a $\poly(n, s^d)$-time algorithm for recovering all \emph{sparse} factors of $f$ (without spurious factors). As a corollary, we can also recover \emph{all} factors of $f$ (without spurious factors) in time $\poly(n, s^{d^2 \log n})$, and recover the algorithmic result of Bhargava, Saraf and Volkovich~\cite{BSV20} and its improvement by Chuyoon and Shpilka~\cite{chuyoon2026factorization}. Both the above consequences follow from known interpolation and divisibility testing techniques.

\item We give a deterministic quasipolynomial-time algorithm which takes as input a general $n$-variate $s$-sparse polynomial
$f$ of (unbounded) individual degree $D$ and outputs a list $\cL$ of polynomials which
contains all factors of $f$ that have bounded individual degree $d$. The algorithm runs in time $\poly(D^{d \log s}, s^{d^2 \log n})$ and works over arbitrary fields. The list may again contain spurious elements. Our result strengthens results of Dutta, Sinhababu and Thierauf~\cite{DST24} and Kumar, Ramanathan and Saptharishi~\cite{KRS24} which give algorithms to recover all factors of $f$ of bounded total degree.

A consequence of our algorithm is a new upper bound on the total number of bounded individual degree factors of a sparse polynomial.

\end{enumerate}
\end{abstract}
\newpage
\tableofcontents
\newpage

\section{Introduction}

Polynomial factorization is a central problem in computational algebra and algebraic complexity, with connections to coding theory, cryptography, pseudorandomness and derandomization. We refer the reader to the recent survey of Bhargav, Dwivedi, and Saxena~\cite{BDSSurvey25} for a detailed account.

Foundational work by Kaltofen~\cite{Kaltofen89}, and also Kaltofen and Trager~\cite{KT90}, showed that multivariate polynomial factorization can be performed efficiently by randomized algorithms when the input is given by an algebraic circuit. These works also established the remarkable structural fact that factors of low-degree polynomials computed by small algebraic circuits themselves have small algebraic circuits. Even prior to these works, in the sparse setting, von zur Gathen and Kaltofen gave randomized algorithms for factoring sparse multivariate polynomials, with running time depending polynomially on the sparsity of the input and of the output factors~\cite{vzGK85}.
Some very natural questions that arise are the following: can such algorithms be derandomized, and can one prove strong upper bounds on the complexity of factors of sparse polynomials?

By the work of Kopparty, Saraf and Shpilka~\cite{KSS15}, we know that for general algebraic circuits, derandomizing multivariate factorization is essentially \emph{equivalent} to derandomizing polynomial identity testing. Since this work, there has been significant progress on understanding if this equivalence extends to other subclasses of algebraic circuits, and this work will focus on the class of sparse polynomials. 

In a recent work, Bhattacharjee, Kumar, Rai, Ramanathan, Saptharishi and Saraf proved closure under factorization for formulas and constant-depth circuits, and using this closure, extended the connection of derandomizing PIT and derandomizing factoring to these restricted but fundamental algebraic models \cite{BKRSS25, BKRRSS25a}. Thus, for several natural circuit classes, deterministic factorization and PIT are now known to be tightly linked. Despite this progress, we are still quite far from understanding whether this equivalence is true for sparse polynomials.
On the one hand, sparse polynomials are among the simplest algebraic models, and efficient deterministic PIT and reconstruction algorithms for them have been known since the work of Klivans and Spielman~\cite{KS01} and related earlier works~\cite{GK87, BOT88, GKS90, CDGK91}. On the other hand, we are still far from obtaining a deterministic polynomial-time algorithm for factoring general sparse polynomials.
There has nevertheless been significant progress on important special cases.

For constant-depth circuits, and in particular for sparse polynomials, Kumar, Ramanathan, Saptharishi and Volk gave a deterministic subexponential-time algorithm that outputs a list of candidate circuits containing all relevant irreducible factors, but the list may also contain spurious circuits \cite{KRSV26}.
This issue of having spurious circuits in the list was overcome in the constant-depth (and hence also sparse) setting by \cite{BKRSS25}, who gave a deterministic subexponential-time algorithm that outputs all irreducible factors of a constant-depth circuit, together with their multiplicities. Their algorithm yields, in particular, the first subexponential-time deterministic algorithm for factoring sparse polynomials, since sparse polynomials are depth-two circuits. The subsequent closure-under-factorization result of \cite{BKRRSS25a} gave a cleaner structural explanation and strengthened the connection between PIT and factorization for constant-depth circuits and formulas. However, these general constant-depth algorithms are still subexponential, and do not exploit the full special structure of sparse polynomials.

Kumar, Ramanathan and Saptharishi, and independently Dutta, Sinhababu and Thierauf, gave quasipolynomial-time deterministic algorithms for recovering all constant-degree factors of sparse polynomials \cite{KRS24,DST24}\footnote{These works also do more. \cite{KRS24} gave a subexponential-time algorithm for outputting all the constant-degree factors of constant-depth circuits. \cite{DST24} showed that computing all the constant-degree factors of a polynomial in a class $\cC$ reduces to PIT for the class $\cC$ and divisibility tests of the polynomial by constant-degree polynomials.}.
However, to recover higher degree factors, very little is known. Indeed, in the setting of sparse polynomials, other than subexponential algorithms of~\cite{BKRSS25, BKRRSS25a}, we know nothing better even for the task of finding multilinear factors.
The following is an open question from~\cite{DST24}: \emph{Can we find bounded individual degree factors of a sparse polynomial (that may not have bounded individual degree) in deterministic quasi-polynomial time?}.

A more tractable and very natural setting is that of factoring sparse polynomials of bounded individual degree. Here the input polynomial is sparse, and in addition the exponent of each variable is at most some bounded parameter $d$. The cases $d=1$ and $d=2$ were handled by Shpilka and Volkovich and by Volkovich, respectively, who gave efficient deterministic factorization algorithms for sparse multilinear and sparse multiquadratic polynomials \cite{SV10,Volkovich17}. For general constant $d$, Bhargava, Saraf and Volkovich gave the first quasipolynomial-time deterministic algorithm for factoring sparse polynomials of bounded individual degree \cite{BSV20}.

A key ingredient in their work was a structural theorem: if an $s$-sparse $n$-variate polynomial has individual degree at most $d$, then every factor has sparsity at most $s^{O(d^2 \log n)}$. The algorithmic framework by \cite{BSV20} has the useful feature that an improved sparsity bound would immediately yield an improved factoring algorithm; in particular, a polynomial sparsity bound would imply a polynomial-time deterministic factorization algorithm in this setting.

Further evidence for this conjectural polynomial sparsity bound was given by Bisht and Volkovich, who obtained polynomial-time algorithms for several problems that would follow from such a bound, including identity testing for certain depth-four circuits and testing whether a sparse polynomial is an exact power \cite{BV25}.

Very recently, Chuyoon and Shpilka made further progress and gave deterministic polynomial-time algorithms for finding all sparse divisors of bounded-individual-degree sparse polynomials, and more generally developed algorithms for recovering sparse factors in several related settings \cite{chuyoon2026factorization}. This shows that, at least if one is only interested in factors that are themselves sparse, polynomial-time deterministic algorithms are possible in the bounded-individual-degree regime. However, the full factorization problem remains elusive, since the currently known bounds do not rule out factors of sparse polynomials (with bounded individual degree) having super-polynomially many monomials.

 \subsection{Our Results}
 In this work, we make progress toward deterministic polynomial-time factorization of sparse polynomials of bounded individual degree, as well as the problem of finding bounded individual degree factors of general sparse polynomials.

 We first state our results informally.

\begin{enumerate}

 \item Our first main result gives a deterministic polynomial-time algorithm which, on input an $s$-sparse polynomial of individual degree at most $d$, outputs a list of candidate factors that contains all factors of the input polynomial that are not divisible by a monomial. The list may also contain spurious candidates. Thus our result is in the spirit of the list-producing algorithm of Kumar, Ramanathan, Saptharishi and Volk~\cite{KRSV26} for constant-depth circuits, but in the bounded-individual-degree sparse setting we obtain a polynomial-time algorithm rather than a subexponential-time one.

 Our algorithm also recovers and strengthens some of the earlier mentioned results. In the bounded-individual-degree setting, efficient reconstruction and divisibility testing for sparse polynomials allow us to recover the polynomial-time sparse-divisor algorithms of Chuyoon and Shpilka~\cite{chuyoon2026factorization}, as well as the quasipolynomial-time factorization algorithm of Bhargava, Saraf and Volkovich~\cite{BSV20}.

\item In the setting of general sparse polynomials, we give a quasipolynomial-time deterministic algorithm that outputs a list containing all factors of individual degree at most $d$, again allowing the possibility of spurious candidates.

Thus we make progress towards the aforementioned open question of~\cite{DST24}. Our results essentially show that factoring in this setting reduces to divisibility testing. If we could efficiently test whether a sparse polynomial of bounded individual degree divides another sparse polynomial, then we would be able to prune out the list and output only the true factors.

When specialized to factors of constant total degree, known reconstruction and divisibility-testing tools~\cite{Forbes15} allow the list to be pruned, thus recovering the quasipolynomial-time algorithms of Kumar, Ramanathan and Saptharishi~\cite{KRS24} and of Dutta, Sinhababu and Thierauf~\cite{DST24} for constant-degree factors of sparse polynomials.

\end{enumerate}
We hope that our results are a useful step toward the long-standing goal of deterministic polynomial-time factorization of sparse polynomials.
We now state our results more formally.

\paragraph{Factoring sparse polynomials of bounded individual degree:}
As our first result, we give a deterministic polynomial-time algorithm which takes as input an $n$-variate $s$-sparse polynomial
$f$ of bounded individual degree $d$ and outputs a list $\cL$ of circuits which
contains all factors of $f$ that are not divisible by a monomial. The algorithm runs in time $\poly(n, s^d)$.

The list may contain spurious circuits which do not correspond to factors of $f$. Additionally, we can ensure that every circuit in the list has constant depth.

\begin{theorem}[Main Theorem 1]\label{thm:main1}
   Let $\F$ be any field such that $\Char(\F)$ is $0$ or greater than $d$. Let $f(X_1,\dots, X_n)\in \F[X_1, \dots, X_n]$ be an $s$-sparse polynomial with individual degree at most $d$. Then there is an algorithm that runs in time $\poly(s^d, d!, n)$ and outputs an $s^d$-size list of circuits $\cL$ such that every factor of $f$ that is not divisible by a monomial is computed by some circuit in $\cL$. The algorithm also outputs the maximal monomial divisor of $f$.

   Moreover, each circuit in $\cL$ has size $\poly(s^d, n)$ and depth $\cO(1)$, and if  $|\F| > \poly(n, s^d)$ the circuits are over $\F$.
   \end{theorem}

\begin{remark}
The above runtime suppresses polynomial dependence on the running time of factoring degree $d$ univariate polynomials over $\F$. Over $\Q$, this is just polynomial in $d$ and the input bit complexity. However, for fields of characteristic $p$, there is also a polynomial dependence on $p$. See \cref{thm:compute-all-factors-of-bounded-degree} and \cref{rem:dependence-on-field} for a more detailed statement.
\end{remark}

\begin{remark}
If $|\F|< \poly(n, s^d)$, we just work over a suitable field extension that makes the size sufficiently large for our computations. In this case, the output circuits will be over the field extension. 
\end{remark}
\begin{remark}
\label{rem:factor-not-divisible-by-monomial}
The algorithm outputs only factors not divisible by a monomial. This is unavoidable, since if the product of \(X_1^d\cdots X_n^d\) and \(g\) divides a sparse polynomial \(f\), then \(X_1^{i_1}\cdots X_n^{i_n}g\) also divides \(f\) for all \(0\leq i_1,\ldots,i_n\leq d\), yielding \((d+1)^n\) factors. The omitted monomial part is easy to recover: for \(f=\sum_{j=1}^s c_j X_1^{{\balpha}_{j, 1}} \cdots X_n^{{\balpha}_{j, n}}\), the maximal monomial divisor is \(\prod_{i=1}^n X_i^{\min_{j\in[s]}{\balpha}_{j,i}}\).
\end{remark}

\begin{remark}
Our algorithm implicitly gives an upper bound on the total number of factors of an $s$-sparse polynomial of bounded individual degree that are not divisible by a monomial. Such an upper bound was also obtained by Chuyoon and Shpilka~\cite{chuyoon2026factorization} using different techniques. 
\end{remark}

Given the above theorem, as a special case one can recover the recent result of Chuyoon and Shpilka~\cite{chuyoon2026factorization}\footnote{The~\cite{chuyoon2026factorization} result holds over all fields.}, which gives a $\poly(s^d, d!, n)$-time algorithm to return all $s$-sparse factors of an $s$-sparse polynomial with individual degree at most $d$.
Indeed, given the list of circuits produced by our algorithm, we can run a sparse polynomial interpolation algorithm (say by~\cite{KS01}) on each circuit to recover the sparse polynomials computed by them. We then test each interpolated polynomial for divisibility against the input sparse polynomial using the divisibility testing algorithm designed by \cite{Volkovich17}. This retains precisely the $s$-sparse divisors in the list and rejects all remaining candidates.

Another immediate corollary is a quasipolynomial-time deterministic algorithm to recover all factors of any $s$-sparse polynomial $f$ with bounded individual degree $d$. \cite{BSV20} showed that every factor of $f$ has quasipolynomial sparsity $\poly(n, s^{d^2 \log{n}})$. Using this, they obtained an algorithm that ran in time $\poly(n, s^{d^7 \log{n}})$. This was improved to $\poly(n, s^{d^2 \log{n}})$ in \cite{chuyoon2026factorization}.

We can recover this latter bound on running time as a corollary of \cref{thm:main1}. Indeed, by the quasipolynomial sparsity bound of~\cite{BSV20}, every factor of $f$ appears %
in the output list and has only  $\poly(n, s^{d^2 \log{n}})$ monomials if expressed in the monomial representation. As in the sparse-factor case discussed above, we run (quasipolynomial) sparse polynomial interpolation on each circuit in the list, and then use divisibility testing to retain precisely those interpolated polynomials that divide the input polynomial $f$. This yields all factors of $f$ within the same $\poly(n,s^{d^2\log n})$ running time.

We formally prove the above corollaries in \cref{cor:cs-and-bsv-generalization} (also see \cref{rem:cs-and-bsv-generalization}).

\paragraph{Factoring general sparse polynomials}
Note that when the input polynomial $f$ has bounded individual degree $d$, all factors of $f$ automatically have bounded individual degree $d$. We also provide efficient factoring algorithms for the case of \emph{general} sparse polynomials with no bound on individual degree. However, we only show how to recover bounded individual degree factors of $f$. As a first step, we show that number of factors of bounded individual degree of a general $s$-sparse polynomial cannot be too large.

 In the theorem below, it will be useful to think of $D$ as a large parameter, and $d$ as a small fixed constant or slowly growing parameter. In the statement below, we use $\binom{D}{\leq d}$ to denote the expression $\sum_{i=0}^d \binom{D}{i}$. Note that this quantity is upper bounded by $D^{O(d)}$.

\begin{theorem}[Divisor bound for general sparse polynomials]\label{thm:divbound}
Let $f$ be any $n$-variate $s$-sparse polynomial of individual degree at most $D$. Then $f$ can have at most $\binom{D}{\leq d}^{\log s}$ factors of individual degree at most $d$ that are not divisible by any monomial. Moreover, this bound is tight.
\end{theorem}

We prove this theorem formally in \cref{sec:recovering-bounded-individual-degree-factors}.
Prior to this work, no such nontrivial bound was known even for the number of multilinear factors.

We also give a deterministic quasipolynomial-time algorithm which takes as input a general multivariate $s$-sparse polynomial $f$ of individual degree at most $D$ and outputs a list $\cL$ of polynomials which contains all factors of $f$ that have individual degree at most $d$. The list may contain additional spurious elements.

Again, in the theorem below, it will be useful to think of $D$ as a large parameter, and $d$ as a small fixed constant or slowly growing parameter.

\begin{theorem}[Main Theorem 2]\label{thm:main2}
   Let $\F$ be an arbitrary field. Let $f(X_1,\dots, X_n)\in \F[X_1, \dots, X_n]$ be an $s$-sparse polynomial with individual degree at most $D$. Then there is an algorithm that runs in time $\poly(s^{d^2 \log n}, \binom{D}{\leq d}^{\log s}, n)$ and outputs a $\binom{D}{\leq d}^{\log s}$-size list of polynomials $\cL$ such that every factor of $f$ that is not divisible by a monomial and that has individual degree at most $d$ lies in $\cL$.
\end{theorem}

\begin{remark}
The above runtime suppresses polynomial dependence on the running time of factoring degree $D$ univariate polynomials over $\F$. Over $\Q$, this is just polynomial in $D$ and the input bit complexity. However, for fields of characteristic $p$, there is also a polynomial dependence on $p$. See \cref{thm:general-sparse-bounded_ind_deg_factor} and \cref{rem:dependence-on-field-all-factors} for a more detailed statement.
\end{remark}
\begin{remark}
It follows from our algorithm that additionally each polynomial in $\cL$ has at most $\poly(s^{d^2 \log n})$ monomials.
\end{remark}

\begin{remark}
If we only wanted to recover a list containing all $s'$-sparse factors of individual degree at most $d$, then our algorithm runtime can be improved to $\poly(s', \binom{D}{\leq d}^{\log s}, n)$. By the factor sparsity bound of~\cite{BSV20}, every factor of $f$ of individual degree at most $d$ has sparsity at most $s^{\cO(d^2 \log n)}$. Thus, setting $s'=s^{\cO(d^2 \log n)}$ recovers the parameters in Main Theorem 2.
\end{remark}

\cref{thm:main2} makes progress towards an open problem by Dutta, Sinhababu and Thierauf~\cite{DST24}, who asked if all \emph{bounded individual degree} (sparse) factors of a general sparse polynomial can be found in deterministic quasipolynomial-time. We make progress on this problem by outputting a list of all bounded individual degree factors, but with the caveat that the list may contain spurious elements. Our result in particular also makes progress towards answering a question of Volkovich~\cite{vol15} about finding multilinear factors of a general sparse polynomial. %

\cref{thm:main2} reduces the problem of computing all bounded-individual-degree factors of a sparse polynomial to divisibility testing. Indeed, if one can deterministically test whether an input sparse polynomial of bounded individual degree divides a general sparse polynomial, then one can run this test on every polynomial in the output list and prune the list down to exactly the true factors.

An interesting special case of the above theorem is the problem of recovering factors of constant total degree. In this setting, the quasipolynomial-time divisibility tests of Forbes~\cite{Forbes15} allow us to prune the spurious candidates from our list and output exactly the constant-total-degree factors. This recovers the results of Dutta, Sinhababu and Thierauf~\cite{DST24} and Kumar, Ramanathan and Saptharishi~\cite{KRS24} in the setting of factoring sparse polynomials. Indeed, if the divisibility test can be done in polynomial-time, then a simple modification of our algorithm will give a polynomial-time factoring algorithm.

\subsection{Proof Overview}
\label{sec:proof-overview}

Recall that \cite{BSV20} gives a deterministic quasipolynomial-time algorithm for factoring \(n\)-variate \(s\)-sparse polynomials of bounded individual degree \(d\), with running time \(\poly(n, s^{d^7\log n})\). More recently, this was improved in \cite{chuyoon2026factorization} to \(\poly(n, s^{d^2\log n})\). Moreover, and perhaps more significantly, the latter work gives a deterministic \emph{polynomial-time} algorithm to output all \(s\)-sparse factors of such polynomials in time \(\poly(s^d, d!, n)\).

Since our algorithms are inspired by and build upon the algorithms in \cite{BSV20} and \cite{chuyoon2026factorization}, it will be instructive to see high-level sketches of these algorithms, and to understand where they fail for the more general task we are attempting. Thus through the proof overview we will discuss our techniques in conjunction with those from \cite{BSV20} and \cite{chuyoon2026factorization}.

The algorithms in \cite{BSV20}, \cite{chuyoon2026factorization}, and the algorithms from our main theorems can all be viewed as a composition of the following steps: (0) a structural theorem, (1) normalization to obtain a monic polynomial, (2) factoring the normalized polynomial, (3) a recursive call to a simpler polynomial, and (4) recovering factors of the original polynomial.

Before we delve into a discussion of these steps, we first set up some notation.

We use $\X_n$ to denote the tuple of variables $(X_1,\dots,X_n)$ and often just abbreviate it to $\X$.

Let $f(\X,Y)\in\F[\X,Y]$ be an $s$-sparse polynomial of individual degree $d$ and suppose that
\[
    f(\X,Y)=f_kY^k+f_{k-1}Y^{k-1}+\cdots+f_1Y+f_0,
\]
where for each $i\in\{0,\dots,k\}$, we have $f_i\in\F[\X]$ and $f_k\neq 0$.

\subsubsection{Sparse Polynomials of Bounded Individual Degree}

We first give a proof sketch of our first main result (\cref{thm:main1}).

\paragraph{Step 0: Structural theorem.}

The starting point of~\cite{BSV20} is a structural theorem.
\cite{BSV20} shows, using tools from discrete geometry and the theory of Newton polytopes, that if an $n$-variate polynomial $f$ of individual degree $d$ has $s$ monomials,
then any factor of $f$ can have at most $s^{O(d^2 \log n)}$ monomials.

In \cite{chuyoon2026factorization}, again using the theory of Newton polytopes, it is shown that if an $n$-variate polynomial $f$ of individual degree $d$ has $s$ monomials,
then the total number of factors (and in particular sparse factors) not divisible by a monomial is at most $s^d$.
In order to efficiently output all sparse factors, it is crucial for the analysis in \cite{chuyoon2026factorization} to show that the number of sparse factors is polynomially bounded.

In our proof of \cref{thm:main1}, we will give a new and implicit proof of the \cite{chuyoon2026factorization} factor bound. The bound on the output list size and number of factors will follow naturally from the analysis of the algorithm.

As in the formal theorem statement, all factor bounds in this discussion are for factors that are not divisible by a monomial. This restriction is unavoidable if one wants a polynomial-size output list: monomial multiples of a fixed factor can already give exponentially many divisors. Algorithmically, we handle this by outputting the maximal monomial divisor of the input separately; the recursive factor-finding procedure only has to recover the non-monomial part.

Our analysis of \cref{thm:main2} will also implicitly imply the divisor bound for general sparse polynomials, i.e., \cref{thm:divbound}.

\paragraph{Step 1: Normalization}

A standard preprocessing step in the factorization literature is a normalization procedure that converts the input polynomial to one that is monic in some chosen variable, say $Y$, i.e., whose leading coefficient is 1 (or reverse-monic, i.e., whose constant term is 1).

Let the \emph{normalization} of $f(\X,Y)$ be the monic polynomial

  \[
        \hat f(\X,Y)
        =
        f\!\left(\X,\frac{Y}{f_k}\right)f_k^{\,k-1}.
    \]

Assuming $f_0\neq 0$, the \emph{reverse normalization}\footnote{This operation is called normalization in~\cite{chuyoon2026factorization}. We call it reverse normalization to distinguish it from the monic normalization above: it scales the constant coefficient in $Y$ to $1$, rather than the leading coefficient.} of $f(\X,Y)$ is the reversed-monic polynomial
    \[
        \tilde f(\X, Y)
        \;=\;
        \frac{f(\X,Yf_0)}{f_0}.
    \]

\cite{BSV20} deals with normalization and \cite{chuyoon2026factorization} deals with reverse normalization, but either step has essentially the same effect: the resulting polynomial has no factor independent of $Y$. Since the degree in $Y$ is at most $d$ and $\hat f$ is monic, $\hat f$ has at most $d$ irreducible factors (counted with multiplicity). So, all its factors can be obtained by multiplying subcollections of these irreducible factors. Though the entire analysis in \cite{chuyoon2026factorization} is presented with respect to reverse normalization, it could just as easily and equivalently have been presented with respect to normalization, and in the rest of the overview, for simplicity, we will assume
\cite{chuyoon2026factorization} uses normalization.

A crucial aspect of our algorithms, both for \cref{thm:main1} and for \cref{thm:main2}, is that we will make use of \emph{both} normalization and reverse normalization.

\paragraph{Step 2: Factoring the normalized polynomial.}

The algorithm in \cite{BSV20} shows how to completely factor $\hat f(\X,Y)$ in quasipolynomial time. The algorithm in \cite{chuyoon2026factorization} aims only to recover those factors that correspond to \emph{sparse} factors of the original polynomial $f$, and shows how to do this in polynomial time. There are some subtleties in making this formal, but we will not go into the details here.

Our algorithm (for \cref{thm:main1}) shows how to find a $2^d$-size list of circuits that contains \emph{all} factors of $\hat f(\X,Y)$ in polynomial time. One caveat is that the list may contain some circuits that do not represent factors, and unfortunately we do not know how to detect these spurious circuits and prune them out.

For now, just to get a bird's-eye view of the entire algorithm, let us not go into details of the monic factorization, but assume that it can be done. We first describe how this can be used to recover all factors of $f(\X,Y)$ efficiently.

Note that for our analysis to work, we will also need to have an efficient way of finding a $2^d$-size list of circuits that contains \emph{all} factors of $\tilde f(\X,Y)$ (the reverse-normalized polynomial) in polynomial time, and let us for now assume that this can be done as well.

\paragraph{Step 3: Recursively factoring $f_k(\X)$ or $f_0(\X)$.}
Both~\cite{BSV20} and \cite{chuyoon2026factorization} recursively factor $f_k(\X)$ and then show how these factors can be combined with the factors of $\hat f(\X,Y)$
to obtain the factors of $f(\X,Y)$.

We will need to be able to factor \emph{either} $f_k(\X)$ \emph{or} $f_0(\X)$, and show how
these factors can be combined with the factors of $\hat f(\X,Y)$ or $\tilde f(\X,Y)$ to recover factors of $f(\X,Y)$. The reason for keeping both options is that one of $f_k$ and $f_0$ has sparsity at most half the sparsity of $f$, and recursing on the sparser one will keep the recursion depth to at most $\log s$.

\paragraph{Step 4: From factors of $\hat f(\X,Y)$ (or $\tilde f(\X,Y)$) and $f_k(\X)$ (or $f_0(\X)$) to factors of $f(\X,Y)$.}

We will use the fact that factors of $f(\X,Y)$ that depend on the variable $Y$ are closely related to factors of $\hat f(\X,Y)$.
Indeed, for every factor $g(\X, Y)$ of $Y$-degree $r$ of $f(\X, Y)$, if
\[
    g(\X,Y)=g_rY^r+g_{r-1}Y^{r-1}+\cdots+g_1Y+g_0,
\]

then the following polynomial

\[
    \hat g(\X, Y) = g\left(\X, \frac{Y}{f_k(\X)}\right) \cdot \frac{f_k(\X)^r}{g_r(\X)}
\]

is monic in $Y$ and it can be shown by Gauss' Lemma that it is a factor of $\hat f(\X,Y)$. Moreover $g_r(\X)$ is a factor of $f_k(\X)$, since the leading coefficient (in $Y$) of $g(\X, Y)$ divides the leading coefficient (in $Y$) of $f(\X, Y)$.

Thus to recover all factors of $f(\X,Y)$, it suffices to learn factors of $\hat f(\X, Y)$ and factors of $f_k(\X)$ and suitably \emph{combine} them.

Indeed, \[
    g(\X, Y) = \hat g(\X, Y \cdot f_k(\X)) \cdot \frac{g_r(\X)}{f_k(\X)^r}
\]

The factors of $f_k(\X)$ will contain all factors of $f(\X,Y)$ that do not depend on $Y$, as well as factors that can be combined with factors of
$\hat f(\X,Y)$ to recover factors of $f(\X,Y)$ that do depend on $Y$. In some sense this is what both \cite{BSV20} and \cite{chuyoon2026factorization} also do.

However, one key step that both \cite{BSV20} and \cite{chuyoon2026factorization} can perform is \emph{testing} whether the purported factor returned by the combining step is indeed a true factor. Thus their algorithms can prune out all the spurious circuits and keep only the true factors. Indeed this notion of keeping only true factors can be attained at each step of the recursive procedure, and hence the final list size output by the algorithm is upper bounded by the total number of irreducible factors (in \cite{BSV20}) and the total number of factors\footnote{More precisely, only those not divisible by a monomial.} (in \cite{chuyoon2026factorization}).

In our case, since the testing or pruning step is something we do not know how to do, \emph{even if we knew an upper bound on the total number of factors}, following the above algorithmic scheme would not let us control the size of the output list. This is because there could potentially be $n$ recursive calls, and at each recursive step the list size grows by a multiplicative factor of $2^d$. This would make the final list size output by the algorithm exponential in $n$.

This is where our ability to choose between normalization and reverse-normalization becomes useful: it lets us decide whether to recurse on $f_k$ or on $f_0$.

In past algorithms, when the algorithm was recursively factoring $f_k(\X)$, one reason $f_k(\X)$ was a simpler polynomial was that it depended on one fewer variable. One could have also carried out a very similar scheme and recursed on $f_0(\X)$, again a polynomial on one fewer variable. However, notice that if $f(\X,Y)$ has $s$ monomials, then since $f_k(\X)$ and $f_0(\X)$ contribute to distinct monomials in $f$, one of them has sparsity at most $s/2$. Thus, if at each step the algorithm chooses to recurse on the polynomial with lower sparsity, then the recursion depth is at most $\log s$.

This also gives the informal divide-and-conquer bound on the list size. At each recursive level, the factoring of the normalized or reverse-normalized polynomial contributes a factor of at most $2^d$ to the list size, while the sparsity of the polynomial on which we recurse drops by a factor of two. Thus, at the level of this overview, one should think of the recurrence as
\[
    L(s)\leq 2^d\cdot L(s/2),
\]
and hence $L(s)\leq (2^d)^{\log s}=s^d$. This is the basic reason our algorithm produces an $s^d$-size list, and also why it implicitly recovers the corresponding bound on the number of non-monomial factors.

For completeness, we spell out the steps needed if the algorithm chooses to recurse on $f_0(\X)$.

For this, note that factors of $f(\X,Y)$ that depend on the variable $Y$ are also closely related to factors of the reverse-monic polynomial $\tilde f(\X,Y)$.
Indeed, for every factor $g(\X, Y)$ of $Y$-degree $r$ of $f(\X, Y)$, if
\[
    g(\X,Y)=g_rY^r+g_{r-1}Y^{r-1}+\cdots+g_1Y+g_0,
\]

then the following polynomial

\[
    \tilde g(\X, Y) = g\left(\X, Y \cdot f_0(\X)\right) \cdot \frac{1}{g_0(\X)}
\]

is reverse-monic in $Y$ and it can be shown that it is a factor of $\tilde f(\X,Y)$. Moreover $g_0(\X)$ is a factor of $f_0(\X)$.

Thus to recover all factors of $f(\X,Y)$, it suffices to learn factors of $\tilde f(\X, Y)$ and factors of $f_0(\X)$ and suitably \emph{combine} them.

Indeed, \[
    g(\X, Y) = \tilde g\left(\X, \frac{Y}{f_0(\X)}\right) \cdot g_0(\X)
\]

The factors of $f_0(\X)$ will contain all factors of $f(\X,Y)$ that do not depend on $Y$, as well as factors that can be combined with factors of
$\tilde f(\X,Y)$ to recover factors of $f(\X,Y)$ that do depend on $Y$.

Strictly speaking, the formal algorithm first produces monomial-equivalent candidates: for every desired factor $g$, the recursive list contains some monomial multiple of $g$. A final post-processing step removes the largest monomial divisor from each candidate. This does not change the above list-size recurrence, but it is the reason the formal proof is phrased in terms of monomial-equivalent factors before the final output list is produced.

\paragraph{Details of Step 2: Factoring the normalized polynomial}

The algorithm in \cite{BSV20} shows how to completely factor $\hat f(\X,Y)$. Recall that the algorithm from~\cite{BSV20} runs in quasipolynomial time and returns
monomial representations of all factors. The proof in~\cite{BSV20} is in the style of the Kaltofen-Trager~\cite{KT90, Kaltofen89} factoring algorithm and we omit the details.
A crucial step is to find a projection to a univariate polynomial just in $Y$ that keeps coprime factors of the original polynomial coprime even after the projection. For this, \cite{BSV20} needs to set the $\X$ variables to a value $\balpha$ which keeps the resultant (in $Y$) of every pair of irreducible factors nonzero. They show that this can be done in quasipolynomial time since the irreducible factors have sparsity at most $s^{O(d^2 \log n)}$ and have $Y$-degree at most $d$, and the low degree of the $Y$ variable and the sparsity of the factors imply that these resultants are also quasipolynomially sparse.

We would like to run in \emph{polynomial time} and recover \emph{all} factors, though possibly with some spurious factors. We show how, in polynomial time, to output a $2^d$-size list of constant-depth circuits such that each factor of $\hat f(\X,Y)$ is a member of this list.

At a high level, our factoring algorithm first projects $\hat f(\X,Y)$ to a univariate polynomial in $Y$, finds the roots of the univariate projection, lifts these roots to power series roots of the shifted polynomial, and then combines these roots to recover the various potential factors of $\hat f(\X,Y)$.

A well-known lemma, namely \textit{Hensel's lemma}, says that if $\hat{f}$ is monic in $Y$ and satisfies some suitable conditions, then $\hf$ splits completely into power series roots (see \cref{Hensel lemma}).

Hensel's lemma (\cref{Hensel lemma}) essentially says that any polynomial $P(\X,Y)$ that is monic in $Y$ admits a factorization into power series roots if it satisfies the following two conditions:
\begin{enumerate}
    \item For every irreducible factor $g(\X,Y)$ of $P(\X,Y)$, $g(\mathbf{0},Y)$ is square-free.
    \item For every pair of distinct irreducible factors $g_1(\X,Y)$ and $g_2(\X,Y)$ of $P$, $g_1(\mathbf{0},Y)$ and $g_2(\mathbf{0},Y)$ are coprime.
\end{enumerate}

Our algorithm for factoring will first find an $\balpha$ such that $\hf_{\balpha}(\X,Y):=\hf(\X+\balpha,Y)$ satisfies the above conditions. For this, we essentially need an assignment $\balpha\in\F^n$ such that, after shifting by $\balpha$ and then setting $\X=\mathbf{0}$, all irreducible factors of $\hf$ remain square-free and pairwise coprime.

In the special case where $\hat f(\X,Y)$ is square-free, it is well known that any $\balpha$ which hits the discriminant (i.e., keeps it nonzero) of $\hf$ (see \cref{def: Discriminant}) will work here. The discriminant in this case is $(2d)!s^{d^2}$-sparse, and hence such an $\balpha$ can be found using PIT algorithms for sparse polynomials. In the general case, it turns out that it suffices to find an $\balpha$ that \emph{maintains the rank of the Sylvester matrix} of $\hf$ and $\partial_Y\hf$. Naively this again leads to PIT for $(2d)!s^{d^2}$-sparse polynomials, but we use the relation between the Sylvester matrices of $f$ and $\hat f$ to reduce the relevant hitting-set computation to polynomials of sparsity $(2d)!s^d$. This is what keeps the running time polynomial in $s^d$ (instead of $s^{d^2}$). In a slightly different context (for divisibility testing of products of sparse polynomials), the underlying rank-preservation criterion was already observed in~\cite{chuyoon2026factorization} and earlier in~\cite{BV25}; moreover, these works observed that finding such a projection can be reduced to finding a hitting set for a certain sparse polynomial.

Once we find this $\balpha$, we consider $\hf_{\balpha}(\X,Y):=\hf(\X+\balpha,Y)$.
We will compute all factors of $\hf_{\balpha}(\X,Y)$ and then recover all factors of $\hf(\X,Y)$ by undoing the shift.

Note that this shifting will destroy the sparsity of $\hf$, but to find the power series roots, we actually do not need the sparsity assumption any more; it was only needed to find the $\balpha$.

Once we have the shifted polynomial $\hf_{\balpha}$, there are standard ways, such as Newton iteration, to compute these power series roots up to any desired degree precision. For our purposes, however, we use Furstenberg's theorem~\cite{Furstenberg67}, which was recently used by \cite{BKRRSS25a} in the context of factorization of constant-depth circuits. The advantage is that Furstenberg's theorem gives additional structure on the power series roots of $\hat f(\X,Y)$, which will be useful to conclude that our final output circuits have constant depth\footnote{We can obtain constant-depth circuits for each homogeneous component of the power series roots. One can obtain a general circuit by using standard Newton iteration.}.

We will now see why we get a list of $2^d$ candidate factors.

First we note that $\hf_{\balpha}(\X,Y)$ has $d$ ($=\deg_Y(\hf)$) power series roots, counted with multiplicity. In other words, there are power series $\varphi_1(\X),\dots,\varphi_d(\X)\in\overline{\F}[[\X]]$ such that we can write
\[
    \hf_{\balpha}(\X,Y)=\prod_{i=1}^d(Y-\varphi_i(\X)).
\]
Moreover, any factor $\hg_{\balpha}(\X,Y)$ of $\hf_{\balpha}(\X,Y)$ is of the form $\displaystyle\prod_{j\in S}(Y-\varphi_j(\X))$ for some subset $S\subseteq[d]$, where repeated roots are listed with their multiplicities. Hence to find any factor $\hg_{\balpha}$, it is enough to find the set of power series $\{\varphi_i\}$ up to sufficiently high precision, combine the appropriate subsets of these roots, and then truncate to the required degree. Since there are at most $2^d$ such subsets, this gives the desired $2^d$-size list.

\subsubsection{Bounded Individual Degree Factors of General Sparse Polynomials}
\label{sec:proof-overview-bounded-individual-degree-factors}

We now describe the ideas that go into the algorithm for \cref{thm:main2}.

Recall that in this setting the input is a general $n$-variate $s$-sparse polynomial $f(\X, Y)$ where each variable has individual degree $D$. We do not assume $D$ is bounded or constant. However, the goal is to output all factors of $f(\X, Y)$ that have bounded individual degree $d$.

\paragraph{Step 0: Structural theorem.}
In order to have an efficient algorithm for outputting all bounded-individual-degree factors (not divisible by a monomial) of a general sparse polynomial, one needs to show that the number of such factors is small. We show that any such polynomial $f(\X, Y)$ can have at most $\binom{D}{\leq d}^{\log s}$ such factors, and indeed this bound is tight.
Instead of separately proving this bound, the bound will naturally follow from the structure of the algorithm.

\paragraph{Step 1: Normalization}
Just as before, we will have a normalization or reverse normalization step.

\paragraph{Step 2: Factoring the normalized (or reverse-normalized) polynomial}
The main issue here is that the monic polynomial $\hat f(\X,Y)$ now might have large or unbounded $Y$-degree. In the previous algorithm, the boundedness of the $Y$-degree was crucial. It allowed us to argue about the sparsity of the discriminant of $\hat f(\X,Y)$, and this allowed us to efficiently find a projection of the $\X$ variables that preserved the coprimality of factors. This is what made factoring via power-series methods, such as Newton iteration or Furstenberg's theorem, effective. Another issue is that $\hat f(\X,Y)$ may not even be sparse any more. These challenges are significant, and indeed we do not try to recover all factors of $\hat f(\X,Y)$. We aim for a more relaxed goal: recovering only those factors that are somewhat sparse (at most $\poly(s^{d^2 \log n})$ monomials) and have $Y$-degree at most $d$.

The reason we only need to focus on somewhat sparse polynomials is that by the sparsity bound of~\cite{BSV20}, any factor of $f(\X,Y)$ that has individual degree at most $d$ can have at most $s^{O(d^2 \log n)}$ monomials. Moreover, it is not hard to see that any monic (in $Y$) polynomial of degree $D$ can have at most $\binom{D}{\leq d}$ factors of $Y$-degree at most $d$.

To achieve this relaxed goal, we use a key idea from the algorithm of~\cite{chuyoon2026factorization}: applying the Klivans-Spielman generator~\cite{KS01} for sparse polynomials. We will use this generator in the parameter settings needed for PIT and reconstruction for polynomials of sparsity $s^{O(d^2 \log n)}$.

In essence, this is a variable-reduction procedure that has the net effect of
``preserving'' all the information in the sparse factors. We map $\hat f(\X,Y)$ to $\hat f(\cG(\bZ),Y)$, a polynomial with a constant number of variables.
We then factor the constant-variate polynomial $\hat f(\cG(\bZ),Y)$ and ``lift" the factors by interpolation to recover all potential \emph{sparse} factors of $\hat f(\X,Y)$.

Note that since only the sparse factors are being recovered, we do not need any sophisticated factoring algorithm for large-variate polynomials. It suffices to factor constant-variate polynomials
and have a way of interpolating and recovering any sparse factor $\hat g(\X,Y)$ of $\hat f(\X,Y)$ from the corresponding factor $\hat g(\cG(\bZ),Y)$ of $\hat f(\cG(\bZ),Y)$.

There is a subtle issue with the above plan. To make the whole factoring scheme work, it will not be enough to recover sparse factors of $\hat f(\X,Y)$; we will also need to do more. We elaborate on this in Step 4 of the algorithm.

\paragraph{Step 3: Recursively factor $f_k(\X)$ or $f_0(\X)$.}
Just as before, we will either recursively factor $f_k(\X)$ or $f_0(\X)$ depending on their sparsity.

\paragraph{Step 4: From factors of $\hat f(\X,Y)$ (or $\tilde f(\X,Y)$) and $f_k(\X)$ (or $f_0(\X)$) to factors of $f(\X,Y)$.}

Let us assume for now that in Step 3, we recursed on $f_k(\X)$ and thus we need to find factors of $\hat f(\X,Y)$.

This step has most of the same components as in the previous algorithm, but with some crucial differences.

One significant difference is that in Step 2, we recover a larger list of $\binom{D}{\leq d}$ potential factors of $\hat f(\X, Y)$, and this list is not guaranteed to contain all factors of $\hat f(\X, Y)$, but only those that are somewhat sparse (at most $\poly(s^{d^2 \log n})$ monomials) and have $Y$-degree at most $d$. When we recursively factor $f_k(\X)$, we are guaranteed to recover all factors (not divisible by a monomial) that have individual degree at most $d$. We need to show that this suffices for recovering all factors of $f(\X,Y)$ of individual degree at most $d$.

The issue is that a bounded-individual-degree factor $g$ of $f$ does not necessarily give rise to a sparse factor $\hat g$ of the normalized polynomial, even when $g$ itself is sparse.

Let $g(\X, Y)$ be a factor of $f(\X,Y)$ that has individual degree at most $d$. As a first step, by the sparsity bound of~\cite{BSV20}, $g(\X,Y)$ must be somewhat sparse: it can have at most $s^{O(d^2 \log n)}$ monomials.
If $g(\X,Y)$ does not depend on $Y$, then it will be output by the recursive call on $f_k(\X)$. Let us therefore assume that $g(\X,Y)$ depends on $Y$ and has $Y$-degree $r$, where $r \leq d$. If
\[
    g(\X,Y)=g_rY^r+g_{r-1}Y^{r-1}+\cdots+g_1Y+g_0,
\]

then

\[
    \hat g(\X, Y) = g\left(\X, \frac{Y}{f_k(\X)}\right) \cdot \frac{f_k(\X)^r}{g_r(\X)}
\]

is monic in $Y$, and, as we argued before, is a factor of $\hat f(\X,Y)$.

The sparsity of $\hat g(\X,Y)$ is controlled by the sparsity of $g(\X,Y)$ except for the division by $g_r(\X)$. Because of this division, we do not have a useful upper bound on the sparsity of $\hat g(\X,Y)$. Thus the algorithm cannot write out $\hat g(\X,Y)$, even though the use of generators gives us access to $\hat g(\cG(\bZ), Y)$.

Instead of interpolating and recovering $\hat g(\X,Y)$, we apply the combining step to $\hat g(\cG(\bZ), Y)$ and $g_r(\cG(\bZ))$ (by iterating over all factors of $f_k(\X)$) to recover $g(\cG(\bZ), Y)$. We then use interpolation to recover $g(\X, Y)$, which is guaranteed to be sparse.

Thus, instead of first recovering $\hat g(\X,Y)$ and then recovering $g(\X, Y)$ by the following combining step with $g_r(\X)$,

\[
    g(\X, Y) = \hat g(\X, Y \cdot f_k(\X)) \cdot \frac{g_r(\X)}{f_k(\X)^r}
\]

we combine $\hat g(\cG(\bZ),Y)$ with $g_r(\cG(\bZ))$ to obtain
\[
    g(\cG(\bZ), Y) = \hat g(\cG(\bZ), Y \cdot f_k(\cG(\bZ))) \cdot \frac{g_r(\cG(\bZ))}{f_k(\cG(\bZ))^r}
\]

Once we have $g(\cG(\bZ), Y)$, since $g(\X, Y)$ has sparsity at most $s^{O(d^2 \log n)}$, we can interpolate to recover $g(\X, Y)$. This idea of doing the combining step in conjunction with the use of generators is inspired by the algorithms in~\cite{chuyoon2026factorization}, although in a different setting of parameters.

Note that for the wrong choice of $g_r(\X)$, we will recover a sparse polynomial which does not correspond to a true factor, and these will contribute to the spurious elements in our list.

There is one final issue caused by monomial factors. In the formal algorithm, we first output a list of monomial-equivalent candidates: for every desired factor $g$, the list contains some monomial multiple of $g$. A final post-processing step then removes the largest monomial divisor from each candidate. Since at each recursive step we recurse on the smaller of $f_k$ and $f_0$, the recursion has depth at most $\log s$. Together with the $\binom{D}{\leq d}$ possible bounded-$Y$-degree factors considered at each level, this gives the claimed list-size bound, and hence also the divisor bound.

\subsection{Organization}

We begin in \cref{sec:preliminaries} by summarizing some well-known results on
algebraic circuits, polynomial factorization, and sparse polynomials. Readers familiar with
these notions may skip directly to
\cref{sec:factor-normalized-and-reverse-normalized-polynomial} and return to
\cref{sec:preliminaries} as needed. In
\cref{sec:factor-normalized-and-reverse-normalized-polynomial}, we show how to
factor normalized and reverse-normalized polynomials. In
\cref{sec:recovering-all-factors}, we present our first
main algorithm, which recovers all factors of sparse polynomials of bounded
individual degree. Finally, in \cref{sec:recovering-bounded-individual-degree-factors}, we present our
second main algorithm, which recovers all bounded-individual-degree factors of
general sparse polynomials.

\section{Preliminaries}
\label{sec:preliminaries}

\paragraph{Bit-complexity convention:}
We will summarize several standard algorithms that take polynomials and circuits as input.
Over fields of characteristic 0, the running time of these algorithms also depends polynomially on the bit
complexity $B$ of the coefficients and constants appearing in the input polynomials and
circuits. Over finite fields, our algorithms have a polylogarithmic dependence on the field size. %
For brevity and ease of exposition, we suppress this dependence throughout. %

\paragraph{Convention on field size:} For the rest of the paper we will assume that   the field we are working over is sufficiently large. If not, we work over a suitable extension.

\subsection{Notation}
Throughout this paper, $[n]$ denotes the set $\{1, 2, \dots, n\}$. We use $\X_n$ to denote the tuple of variables $(X_1,\dots,X_n)$. When the number of variables is clear from context, we abbreviate $\X_n$ to $\X$. We say a polynomial $f(\X_n)$ has individual degree $d$ if for each $i \in [n]$, the degree in $X_i$ is at most $d$.

For parameters $a_1,\dots,a_k$, we write $\poly(a_1,\dots,a_k)$ to denote a quantity bounded by $(a_1+\cdots+a_k)^c$ for some absolute constant $c$.

For field elements or rational functions $a_1,\dots,a_m$, we write $\Diag(a_1,\dots,a_m)$ for the $m\times m$ diagonal matrix with diagonal entries $a_1,\dots,a_m$.

When a polynomial is given explicitly as a sum of monomials, its \emph{bit complexity} is the total number of bits needed to write down its coefficients together with the exponents appearing in its sparse representation.

We will often work with a list of algebraic circuits (defined in \cref{sec:algebraic-circuit}), say $\cL$. For a polynomial $g$, we sometimes abuse notation and write $g \in \cL$ to mean that there exists a circuit in $\cL$ that computes $g$.

We will use $\Esym_r(S)$ for a finite subset $S$ to represent the elementary symmetric polynomial of degree $r$ evaluated on the elements of $S$.

For any polynomial (or formal power series) $P(T)$, $P(T) \Trunc T^{D+1}$ will denote the polynomial obtained by discarding all terms of degree at least $D+1$ in $T$.

\begin{definition}
    \label{def:monomial-equivalent-factor}
    Let $g,h \in \F[\X]$ and $r \in \N$. We write $g \prec_r h$ if there exists a monomial $M$ whose individual degree is at most $r$ such that $h = gM$.
    In this case, we say that $h$ is \emph{$r$-monomial equivalent} to $g$. When the parameter $r$ is clear from context, we simply say that $h$ is monomial equivalent to $g$.
\end{definition}

Let $f(\X,Y)\in\F[\X,Y]$ and suppose
\[
    f(\X,Y)=f_dY^d+f_{d-1}Y^{d-1}+\cdots+f_1Y+f_0,
\]
where for each $i\in\{0,\dots,d\}$, we have $f_i\in\F[\X]$ and $f_d\neq 0$.

\begin{itemize}
    \item $f(\X,Y)$ is said to be \emph{monic in $Y$} if $f_d=1$.

    \item $f(\X,Y)$ is said to be \emph{reversed-monic in $Y$} if $f_0=1$.

    \item The \emph{normalization} of $f(\X,Y)$ is the monic polynomial
    \[
        \hat f(\X,Y)
        =
        f\!\left(\X,\frac{Y}{f_d}\right)f_d^{\,d-1}.
    \]

    \item Assuming $f_0\neq 0$, the \emph{reverse normalization} of $f(\X,Y)$ is the reversed-monic polynomial
    \[
        \tilde f(\X, Y)
        \;=\;
        \frac{f(\X,Yf_0)}{f_0}.
    \]

    \item The \emph{reversal} of $f$ in $Y$, denoted by $\rev_Y(f)$, is defined as
    \[
        \rev_Y(f)
        =
        Y^d f(\X,1/Y)
        =
        f_0Y^d+f_1Y^{d-1}+\cdots+f_d.
    \]
\end{itemize}

We note below that the reversal of a product equals the product of the reversals. Under the mild assumption that the constant term in $Y$ is non-zero, this lets us transfer irreducibility and factorization statements between a polynomial and its reversal.

\begin{lemma}\label{lem: reversefactor}
    Let $f(\X,Y), g(\X,Y), h(\X,Y) \in \F[\X,Y]$. If $f = gh$, then
    \[
        \rev_Y(f) = \rev_Y(g)\rev_Y(h).
    \]
\end{lemma}

\begin{corollary}
    Let $f(\X,Y) \in \F[\X,Y]$, viewed as a polynomial in $\F(\X)[Y]$, and assume $f_0 \neq 0$. Then $f$ is irreducible in $\F(\X)[Y]$ if and only if $\rev_Y(f)$ is irreducible in $\F(\X)[Y]$.
\end{corollary}

\begin{corollary}
    Let $f = g_1^{e_1} g_2^{e_2} \cdots g_k^{e_k}$ be the irreducible factorization of $f(\X,Y) \in \F[\X,Y]$ in $\F(\X)[Y]$. If each factor $g_i$ has non-zero constant term in $Y$, then
    \[
        \rev_Y(f) = \rev_Y(g_1)^{e_1} \rev_Y(g_2)^{e_2} \cdots \rev_Y(g_k)^{e_k}
    \]
    is the irreducible factorization of $\rev_Y(f)$.
\end{corollary}

\subsection{Algebraic circuits}
\label{sec:algebraic-circuit}

Algebraic circuits are a standard way to compactly represent multivariate polynomials in algebraic complexity theory.
\begin{definition}
    An \emph{algebraic circuit} $C$ over a field $\F$ is a directed acyclic graph whose internal gates are labeled by $+$ or $\times$, and whose leaves are labeled by variables or elements of $\F$. We allow gates to have arbitrary fan-in. The circuit computes the polynomial obtained in the natural way by evaluating gates from the leaves to the output gate.

    The \emph{size} of $C$, denoted $\size(C)$, is the number of wires in $C$. The \emph{depth} of $C$, denoted $\depth(C)$, is the length of the longest directed path from an input leaf to the output gate.

    A \emph{circuit with division gates} is defined similarly, except that internal gates may also be labeled by $\div$. Such a circuit computes a rational function over $\F$; in some applications below, this rational function will be promised to equal a polynomial.
\end{definition}

We will also use interpolation to recover coefficients with respect to a single variable while preserving efficient circuit representations. The following standard lemma is due to Michael Ben-Or.

\begin{lemma}[Interpolation]
\label{lem:circuit-coefficient-interpolation}
    Let $C$ be an algebraic circuit of size $m$ and depth $\Delta$ computing a polynomial
    \[
        f(\X,Y)=\sum_{i=0}^{d} f_i(\X)Y^i,
    \]
    where each $f_i(\X)\in\F[\X]$. Suppose $\abs{\F}>d$. Then, given $C$, one can construct circuits
    $C_0,\dots,C_d$ computing $f_0,\dots,f_d$, respectively, each of size
    $\poly(m,d)$ and depth at most $\Delta+2$. Moreover, these circuits can be
    constructed in time $\poly(m,d)$.

\end{lemma}

We will routinely truncate a polynomial computed by a circuit by discarding its
higher-degree terms in one of the variables. We formalize this notation below
and record that such truncations can be computed efficiently.

\begin{definition}[Truncation]
For a polynomial $f(\X,Z)\in\F[\X,Z]$ and an integer $D\geq 0$, we denote by
$f \Trunc Z^D$ the polynomial obtained from $f$ by discarding all monomials whose
degree in $Z$ is at least $D$. Equivalently, if
\[
    f(\X,Z)=\sum_{i\geq 0} f_i(\X)Z^i,
\]
then
\[
    f \Trunc Z^D = \sum_{0\leq i<D} f_i(\X)Z^i.
\]
\end{definition}

The following lemma shows that, given a circuit computing $f(\X,Z)$, one can
efficiently compute a circuit for the truncation of $f$ without significantly
increasing the size or depth.

\begin{lemma}
\label{lem:circuit-truncation}
Let $C$ be an algebraic circuit of size $m$ and depth $\Delta$ over $\F$
computing a polynomial $f(\X,Z)\in\F[\X,Z]$ with $\deg_Z(f)=d$. Suppose
$\abs{\F}>d$. Then, for every $D\in\N$, one can construct, in
time $\poly(m,d)$, a circuit of size $\poly(m,d)$ and depth at most $\Delta+4$
computing $f \Trunc Z^D$.
\end{lemma}

\begin{proof}
Write
\[
    f(\X,Z)=\sum_{i=0}^{d} f_i(\X)Z^i .
\]
By \cref{lem:circuit-coefficient-interpolation}, given $C$, we can construct
circuits $C_0,\dots,C_d$ computing $f_0,\dots,f_d$, respectively, each of size
$\poly(m,d)$ and depth at most $\Delta+2$, in time $\poly(m,d)$. Taking the sum
\[
    \sum_{0\leq i<D} C_i(\X) Z^i
\]
therefore gives a circuit computing $f \Trunc Z^D$. This adds one multiplication
layer to form the terms $C_i(\X)Z^i$ and one addition layer to sum them. Hence
the resulting circuit has size $\poly(m,d)$ and depth at most $\Delta+4$.
\end{proof}

We will later need to recover factors of the original polynomial from factors
obtained from the corresponding normalized and reverse-normalized polynomials. At the level of
circuits, this leads to substitutions of the form $Y\mapsto Y\cdot f$ and
$Y\mapsto Y/f$. A direct implementation by composing the two input circuits
would stack their depths, which is too large for our applications. The following
lemma shows that, using interpolation (\cref{lem:circuit-coefficient-interpolation}) with respect to $Y$, these substitutions
can instead be carried out with only a constant additive overhead over the
larger of the two input depths.

\begin{lemma}
\label{lem:circuit-depth-on-compose}
Let $g(\X,Y)\in \F[\X,Y]$, let $f(\mathbf{Z})\in \F[\mathbf{Z}]$, and let
$d=\deg_Y(g)$. Suppose $g$ is computable by a circuit $C_g$ of size $s_g$ and
depth $\Delta_g$, and $f$ is computable by a circuit $C_f$ of size $s_f$ and
depth $\Delta_f$. Assume $\abs{\F}>d$. Then one can construct circuits with
division gates computing
\[
    g(\X,Y\cdot f)
    \qquad\text{and}\qquad
    g(\X,Y/f),
\]
each of size $\poly(s_g,s_f,d)$ and depth at most
$\max\{\Delta_g,\Delta_f\}+6$. Moreover, each output circuit may be arranged to
have a multiplication gate at the top.
\end{lemma}

\begin{proof}
Write
\[
    g(\X,Y)=\sum_{i=0}^{d} g_i(\X)Y^i .
\]
By \cref{lem:circuit-coefficient-interpolation}, applied to $C_g$ with respect
to $Y$, we can construct circuits $C_0,\dots,C_d$ computing
$g_0,\dots,g_d$, respectively. Each $C_i$ has size $\poly(s_g,d)$ and depth at
most $\Delta_g+2$.

To compute $g(\X,Y\cdot f)$, for each $0\leq i\leq d$ construct a circuit for
\[
    g_i(\X)Y^i f^i
\]
by multiplying the output of $C_i$, $i$ copies of the input $Y$, and $i$ copies
of $C_f$. Summing these circuits over all $i$ gives
\[
    \sum_{i=0}^{d} g_i(\X)Y^i f^i
    =
    g(\X,Y\cdot f).
\]
The size is $\poly(s_g,s_f,d)$, and the depth is at most
$\max\{\Delta_g,\Delta_f\}+5$ after adding, if necessary, a multiplication-by-$1$
gate at the output.

The construction for $g(\X,Y/f)$ is analogous. For each $0\leq i\leq d$, compute
\[
    \frac{g_i(\X)Y^i}{f^i}
\]
using a division gate, where the numerator is the product of the output of
$C_i$ with $i$ copies of $Y$, and the denominator is the product of $i$ copies
of $C_f$. Summing these rational functions gives
\[
    \sum_{i=0}^{d} \frac{g_i(\X)Y^i}{f^i}
    =
    g(\X,Y/f).
\]
Again the size is $\poly(s_g,s_f,d)$, and the depth is at most
$\max\{\Delta_g,\Delta_f\}+6$ after adding, if necessary, a
multiplication-by-$1$ gate at the output.
\end{proof}

Finally, we will use the following lemma on symmetric polynomials to ensure that
the algorithm remains over the ground field. Later, when recovering factors from
the normalized and reverse-normalized polynomials, it will be straightforward to
describe the relevant factors after passing to an extension field containing the
roots of a suitable univariate polynomial. However, our algorithm needs to output
circuits over the original field $\F$. The lemma below allows us to rewrite the
required symmetric expressions over the base field: elementary symmetric polynomials in rational functions evaluated at these roots admit small circuits over the base field.

Andrews and Wigderson~\cite{AW24} proved such a statement over fields of
characteristic zero or sufficiently large characteristic, and
\cite{BKRRSS25b} extended it to all sufficiently large fields, with no
restriction on the characteristic. The proof is constructive and yields the
stated polynomial-time construction.

\begin{lemma}[\cite{AW24,BKRRSS25b}]
\label{lem:esym-of-g/h-on-roots-of-f}
Let $\F$ be a sufficiently large field, say $|\F|\geq (mdn)^c$ for some absolute
constant $c>0$. Let $q(Y)\in \F[Y]$ and
$g(\X_n,Y),h(\X_n,Y)\in \F[\X_n,Y]$ be monic in $Y$, of degrees
$d_q,d_g,d_h$, respectively. Suppose that $q,g,h$ are
given by circuits of size at most $m$ and depth at most $\Delta$. Assume moreover that
$h(\X_n,\sigma_i)$ is nonzero as a polynomial over $\overline{\F}$ for every
root $\sigma_i$ of $q$. Let
$d=\max\{d_q,d_g,d_h\}$. If $\sigma_1,\dots,\sigma_{d_q}$ are the roots of $q$
over $\overline{\F}$, counted with multiplicity, then for every
$1\leq r\leq d_q$, one can construct, in time $\poly(m,n,d)$, circuits
$A(\X)$ and $B(\X)$ over $\F$, each of size $\poly(m,n,d)$ and depth
$\Delta+\cO(1)$, such that
\[
    \Esym_r\left(
        \left\{
            \frac{g(\X,\sigma_i)}{h(\X,\sigma_i)}
            : i\in[d_q]
        \right\}
    \right)
    =
    \frac{A(\X)}{B(\X)} .
\]
\end{lemma}

\subsection{Deterministic Univariate Factorization}\label{UnivariateFactorization}

We briefly summarize some standard results on deterministic algorithms for univariate polynomial factorization.

\begin{definition}
    \label{def:time-taken-to-factor-univariates}
    Let $f \in \F[X]$ be an univariate polynomial of degree $d$. We write $\cT(\F,d)$ for the time taken to factor $f$ over $\F$.
\end{definition}

Univariate factorization over the rational numbers can be performed in deterministic polynomial time using the celebrated LLL algorithm of Lenstra, Lenstra, and Lov\'asz~\cite{LLL82}.

\begin{theorem}[Univariate factorization over rationals {\cite{LLL82}}]
\label{thm:univariate-factorization-rationals}
There is a deterministic algorithm that, given a polynomial $f\in \Q[x]$ of degree $d$ and coefficient bit complexity at most $B$, computes the irreducible factorization of $f$ over $\Q$ in time $\poly(d,B)$. Thus, over $\Q$, $\cT(\Q,d)$ is polynomial in $d$ and in the coefficient bit complexity of the input.
\end{theorem}

Polynomial factorization over finite fields has been extensively studied; see von zur Gathen and Gerhard~\cite[Chapter~14]{von_zur_Gathen_Gerhard_2013} for more details.
The following standard deterministic bound is sufficient for our purposes.

\begin{theorem}[{\cite{Berlekamp70, Shoup90}}]
\label{thm:univariate-factorization-finite-fields}
There is a deterministic algorithm that, given a polynomial $f\in \F_{p^m}[x]$ of degree $d$, computes the irreducible factorization of $f$ over $\F_{p^m}$ in time $\poly(d,m,p)$. In other words, $\cT(\F_{p^m},d)=\poly(d,m,p)$.
\end{theorem}

Since specifying the input polynomial $f$ requires only $\cO(dm\log p)$ bits, the running time in \cref{thm:univariate-factorization-finite-fields} can be super-polynomial in the bit complexity of the input when the characteristic $p$ is large. Finally, we state the time complexity of the following folklore result on factorization of constant-variate polynomials.

\begin{lemma}
    \label{lem: Constantvariate factorization}
    Let $f(\X,Y)$ be a $(c+1)$-variate polynomial of degree $D$ over $\F$. Suppose $f$ is monic or reversed-monic in $Y$. Then all factors of $f$ of $Y$-degree at most $d$ can be computed in time $\poly(D^{cd},\cT(\F,D))$.
\end{lemma}

\subsection{Sylvester Matrix, Resultant, and Discriminant}

\label{subsec:resultant_discriminant}

We recall the definitions of the Sylvester matrix, the resultant, and the discriminant of univariate polynomials. A more thorough treatment of this material can be found in \cite{von_zur_Gathen_Gerhard_2013}.

Resultants give an algebraic way to detect when two univariate polynomials have a common root.

\begin{definition} \label{def:sylvester matrix}
	Let $f(X) = \sum_{i=0}^{d_1} f_i X^i$ and $g(X) = \sum_{i=0}^{d_2} g_i X^i$ be polynomials in $\F[X]$ of degrees $d_1$ and $d_2$, respectively.
	The \emph{Sylvester matrix} of $f$ and $g$ is the $(d_1 + d_2) \times (d_1 + d_2)$ matrix given below, where blank entries are zero:
	\[
		\Syl(f,g) \coloneqq \begin{pmatrix}
			f_{d_1} & & & & g_{d_2} & & & & & \\
			f_{d_1-1} & f_{d_1} & & & g_{d_2-1} & g_{d_2} & & & & \\
			\vdots & \vdots & \ddots & & \vdots & \vdots & \ddots & & & \\
			\vdots & \vdots & & f_{d_1} & g_1 & \vdots & & \ddots & & \\
			\vdots & \vdots & & f_{d_1-1} & g_0 & \vdots & & & \ddots & \\
			\vdots & \vdots & & \vdots & & g_0 & & & & g_{d_2} \\
			f_0 & \vdots & & \vdots & & & \ddots & & & \vdots \\
			& f_0 & & \vdots & & & & \ddots & & \vdots \\
			& & \ddots & \vdots & & & & & \ddots & \vdots \\
			& & & f_0 & & & & & & g_0
		\end{pmatrix}.
	\]
	The \emph{resultant} of $f$ and $g$ is $\Res(f,g) \coloneqq \det \Syl(f,g)$.
\end{definition}

Discriminants give an algebraic way to detect when a univariate polynomial has a repeated root.

\begin{definition}\label{def: Discriminant}
	Let $f(X) \in \F[X]$ be a univariate polynomial.
	The \emph{discriminant} of $f$, denoted $\Disc(f)$, is defined as $\Disc(f) \coloneqq  \Res(f,f')$, where $f'$ is the formal derivative of $f$.
\end{definition}

For $f,g \in \F[\X,Y]$, we write $\Syl_Y(f,g)$, $\Res_Y(f,g)$, and $\Disc_Y(f)$ for the Sylvester matrix, resultant, and discriminant obtained by viewing $f$ and $g$ as polynomials in $Y$ over $\F(\X)$.

The following well-known lemma shows a strong connection between the kernel of the Sylvester matrix and the degree of the gcd. This was also used in other recent factoring algorithms by \cite{Volkovich17,chuyoon2026factorization}.
For completeness, we include a proof in \cref{Appendix:resultant-discriminant}.

\begin{lemma}\label{lem:resultant-prop}
    Let $\F$ be a field. Let $f,g \in \F[X]$ be two non-zero polynomials. Then $\dim\ker(\Syl(f,g)) = \deg(\gcd(f,g))$.
\end{lemma}

In particular, for univariate polynomials, the resultant is nonzero exactly when the polynomials are coprime, and the discriminant is nonzero exactly when the polynomial is square-free.
Applying the preceding lemma to a polynomial and its derivative relates the kernel of the corresponding Sylvester matrix to the repeated factors of the polynomial.

\begin{corollary}\label{cor:resultant-prop}
    Let $f \in \F[X]$ be a monic non-constant polynomial of degree at most $d$, where $\Char(\F)$ is either $0$ or greater than $d$. Suppose $f=\prod_i h_i^{e_i}$ for some monic non-constant polynomials $h_i\in\F[X]$. Then
    \begin{equation*}
        \dim\ker(\Syl(f,f')) \geq \sum_i (e_i-1)\deg(h_i).
    \end{equation*}
    Moreover, equality holds if the $h_i$ are distinct pairwise coprime irreducible polynomials.
\end{corollary}

A proof of the above corollary appears in~\cite{chuyoon2026factorization} and we include a proof in \cref{Appendix:resultant-discriminant}.

Now we discuss some useful properties of the resultant and discriminant of univariate polynomials in the variable $Y$ whose coefficients come from $\F[\X]$. The following lemma is well-known, and for details we refer to \cite{von_zur_Gathen_Gerhard_2013,AW24}.
\begin{lemma}\label{ResultantFact}
Let $\F$ be a field whose characteristic is either $0$ or greater than $d$. Let
$f,g\in\F[\X,Y]$ be monic in $Y$, with $Y$-degrees at most $d$. Then, for any ${\balpha}\in\F^n$, we have
    \begin{enumerate}
        \item $\Syl_Y(f,g)({\balpha})=\Syl(f({\balpha},Y),g({\balpha},Y))$ and $\Res_Y(f,g)({\balpha})=\Res(f({\balpha},Y),g({\balpha},Y))$.
        \item If $f({\balpha},Y)$ and $g({\balpha},Y)$ are coprime, then $\Res_Y(f,g)({\balpha})\neq0$.
        \item If $f({\balpha},Y)$ is square-free, then $\Disc_Y(f)({\balpha})\neq0$.
    \end{enumerate}
\end{lemma}

We next record how the Sylvester matrix changes under the substitution $Y\mapsto Y\cdot g(\X)$.

\begin{lemma}\label{lem:Sylvester-substitution-identity}
    Let $f(\X,Y) \in \F[\X,Y]$ satisfy $\deg_Y(f)=k$, and let $g(\X) \in \F(\X)$ be non-zero. Define $\bar f(\X,Y)\coloneqq f(\X,Y\cdot g(\X))$. Then
    \begin{equation*}
        \Syl_Y(\bar f, \bar f') =
        \Diag\left(1/g,1/g^2,\dots,1/g^{2k-1}\right)
        \cdot \Syl_Y(f,f')
        \cdot
        \Diag\left(g^{k+1},g^{k+2},\dots,g^{2k-1},g^{k+1},g^{k+2},\dots,g^{2k}\right).
    \end{equation*}
\end{lemma}
\begin{proof}
    Write
    \begin{equation*}
        f(\X,Y) = \sum_{i=0}^k f_i(\X)Y^i
    \end{equation*}
    with $f_k\neq 0$. Then
    \begin{equation*}
        \bar{f}(\X,Y) = \sum_{i=0}^k f_i(\X)g(\X)^i Y^i.
    \end{equation*}
    Since $\deg_Y(f')=k-1$, the matrices $\Syl_Y(f,f')$ and $\Syl_Y(\bar f,\bar f')$ have size $(2k-1)\times(2k-1)$. For $i,j\in[2k-1]$, we have
    \begin{align*}
        \Syl_Y(f,f')_{i,j} &= \begin{cases}
            f_{k-(i-j)} & j \in [k-1] \text{ and } 0\leq i-j \leq k,\\
            (j-i+1) f_{j-i+1} & k \leq j \leq 2k-1 \text{ and } 0 \leq j-i \leq k-1,\\
            0 & \text{otherwise,}
        \end{cases}\\
        \Syl_Y(\bar f,\bar f')_{i,j} &= \begin{cases}
            f_{k-(i-j)}g^{k-(i-j)} & j \in [k-1] \text{ and } 0 \leq i-j \leq k,\\
            (j-i+1)f_{j-i+1}g^{j-i+1}  & k \leq j \leq 2k-1 \text{ and } 0 \leq j-i \leq k-1,\\
            0 & \text{otherwise.}
        \end{cases}
    \end{align*}
    These entrywise identities can be rewritten as
    \begin{equation*}
        \Syl_Y(\bar f, \bar f') =
        \Diag\left(1/g,1/g^2,\dots,1/g^{2k-1}\right)
        \cdot \Syl_Y(f,f')
        \cdot
        \Diag\left(g^{k+1},g^{k+2},\dots,g^{2k-1},g^{k+1},g^{k+2},\dots,g^{2k}\right).
    \end{equation*}
    This proves the lemma.
\end{proof}

\begin{remark}
    Since the lemma allows $g(\X)$ to be an element of $\F(\X)$, it also applies to substitutions of the form $Y\mapsto Y/h(\X)$ for non-zero $h(\X)\in\F[\X]$, by taking $g(\X)=1/h(\X)$.
\end{remark}

\subsection{Factors of Normalized and Reverse-Normalized Polynomials}

We record two simple facts showing that factors behave well under the
normalization and reverse-normalization operations defined above. In both cases,
divisibility is immediate over the rational function field $\F(\X)$; the point
is that, after the chosen scaling, the relevant factor lies in $\F[\X,Y]$ and is
monic or reversed-monic in $Y$. We will use the following elementary form of
Gauss's lemma.

\begin{lemma}[Gauss's lemma]
\label{lem:gauss-lemma}
Let $\mathbb A$ be a unique factorization domain with fraction field $\K$. If
$p(Y),q(Y)\in \mathbb A[Y]$ and $p(Y)$ is primitive as a polynomial over $\mathbb A$, then
\[
    p(Y)\mid q(Y) \text{ in } \K[Y]
    \qquad\Longrightarrow\qquad
    p(Y)\mid q(Y) \text{ in } \mathbb A[Y].
\]
In particular, the conclusion holds if $p(Y)$ is monic or reversed-monic.
\end{lemma}

\begin{lemma}
\label{lem:factor-of-normalized-poly}
Let
\[
    f(\X,Y)=\sum_{i=0}^{d} f_i(\X)Y^i,
    \qquad
    g(\X,Y)=\sum_{i=0}^{t} g_i(\X)Y^i
\]
be polynomials in $\F[\X,Y]$ with $f_d\neq 0$ and $g_t\neq 0$. Suppose
$g\mid f$ in $\F[\X,Y]$, and define
\[
    \hat f(\X,Y)
    =
    f(\X,Y/f_d)\, f_d^{d-1},
    \qquad
    \hat g(\X,Y)
    =
    g(\X,Y/f_d)\,\frac{f_d^t}{g_t}.
\]
Then $\hat g\in \F[\X,Y]$ and $\hat g\mid \hat f$ in $\F[\X,Y]$.
\end{lemma}

\begin{proof}
Since $g\mid f$, the leading coefficient $g_t$ divides $f_d$ in $\F[\X]$. Hence
\[
    \hat g(\X,Y)
    =
    \sum_{i=0}^{t} g_i(\X)Y^i\frac{f_d^{t-i}}{g_t}
\]
lies in $\F[\X,Y]$, and it is monic in $Y$. Moreover, over $\F(\X)[Y]$, the
substitution $Y\mapsto Y/f_d$ preserves divisibility, and hence
$\hat g\mid \hat f$ in $\F(\X)[Y]$. Since $\hat g$ is monic, by \cref{lem:gauss-lemma},
$\hat g\mid \hat f$ in $\F[\X,Y]$.
\end{proof}

\begin{lemma}
\label{lem:factor-of-rev-normalized-poly}
Let
\[
    f(\X,Y)=\sum_{i=0}^{d} f_i(\X)Y^i,
    \qquad
    g(\X,Y)=\sum_{i=0}^{t} g_i(\X)Y^i
\]
be polynomials in $\F[\X,Y]$ with $f_0\neq 0$ and $g_0\neq 0$. Suppose
$g\mid f$ in $\F[\X,Y]$, and define
\[
    \tilde f(\X,Y)
    =
    \frac{f(\X,Y f_0)}{f_0},
    \qquad
    \tilde g(\X,Y)
    =
    \frac{g(\X,Y f_0)}{g_0}.
\]
Then $\tilde g\in \F[\X,Y]$ and $\tilde g\mid \tilde f$ in $\F[\X,Y]$.
\end{lemma}

\begin{proof}
Since $g\mid f$, the constant coefficient $g_0$ divides $f_0$ in $\F[\X]$. Thus
\[
    \tilde g(\X,Y)
    =
    \sum_{i=0}^{t} g_i(\X)Y^i\frac{f_0^i}{g_0}
\]
lies in $\F[\X,Y]$. Moreover, over $\F(\X)[Y]$, the substitution
$Y\mapsto Y f_0$ preserves divisibility, and hence
$\tilde g\mid \tilde f$ in $\F(\X)[Y]$. Since $\tilde g$ is reversed-monic, by
\cref{lem:gauss-lemma}, $\tilde g\mid \tilde f$ in $\F[\X,Y]$.
\end{proof}

\subsection{Power Series Roots of Monic Polynomials} \label{Subsec: Power-series root}

In this section, we discuss the factorization of monic polynomials into power series roots. To guarantee the existence of such power series roots, we require the polynomial to satisfy a suitable non-degeneracy condition, formalized in the following definition. Once this condition holds, we will describe how to compute the corresponding power series roots algorithmically.
\begin{definition}
    \label{def:hensel-ready}
    Let $f(\X_n,Y) \in \F[\X_n, Y]$ and suppose $f(\X_n,Y)=f_dY^d+f_{d-1}Y^{d-1}+\dots+f_1Y+f_0$ where $f_0, \dots, f_d \in \F[\X_n]$ and $f_d \neq 0$. Let the irreducible factorization of $f$ be $f=h_{1}^{e_{1}}h_{2}^{e_{2}}\dots h_{r}^{e_{r}}$. We call $f$ \textbf{Hensel-ready} at $\X_n=\balpha\in\F^n$ if $\Res_Y(h_i, h_j)(\balpha)\neq 0$ for all $i\neq j$ and $\Disc_Y(h_i)(\balpha)\neq0$ for all $i \in [r]$.
\end{definition}
Once we find such a point $\balpha$, we may apply the shift $\X_n \mapsto \X_n+\balpha$. The resulting polynomial is Hensel-ready at $\X_n=\mathbf{0}$. We can therefore factor the shifted polynomial around the origin, and finally undo the shift to recover the corresponding factorization of the original polynomial. Thus, without loss of generality, we will assume throughout this section that the polynomial under consideration is Hensel-ready at $\X_n=\mathbf{0}$.

Suppose $f$ is a monic polynomial that is Hensel-ready at $\0$,
we have $$f({\0},Y)=h_{1}({\0},Y)^{e_{1}}h_{2}({\0},Y)^{e_{2}}\dots h_{r}({\0},Y)^{e_{r}}$$
then for each $i \neq j \in [r]$, $\gcd(h_i(\0,Y),h_j(\0,Y))=1$ and $h_j(\0,Y)$ is a square-free polynomial. Let $\cB:=\{\beta_1,\dots,\beta_d\}$ be the multiset of roots of $f(\0,Y)$ over a field extension of $\F$. The following well-known lemma states that each such univariate root can be uniquely lifted to a power series root of $f(\X, Y)$ in $Y$. This lemma appeared in many forms in the literature of commutative algebra and polynomial factorization; for references, see \cite{eisenbud1995commutative, ruiz1993,BDSSurvey25}. For completeness, we include a proof of the lemma in the appendix (see \cref{Proof: Hensel lemma}).

\begin{lemma}[Hensel's lemma]\label{Hensel lemma}
    Let $f(\X,Y)$ be a monic polynomial in $Y$ with $\deg_Y(f)=d$, and let the irreducible factorization of $f$ be $f(\X,Y)=h_{1}^{e_{1}}(\X,Y)h_{2}^{e_{2}}(\X,Y)\dots h_{r}^{e_{r}}(\X,Y)$. Suppose $f$ is Hensel-ready at $\X=\0$ and $\cB=\{\beta_1,\dots,\beta_d\}$ is the multiset of roots of $f(\0,Y)$. Then
    \[
        f(\X,Y)=\prod_{i=1}^d(Y-\Phi_i(\X))
    \]
    where $\Phi_i(\X)\in\F(\beta_i)[[\X]]$, $\Phi_i(\0)=\beta_i$, and $\Phi_i=\Phi_j$ if and only if $\beta_i=\beta_j$.
\end{lemma}

The preceding lemma shows that each power series root of $f(\X,Y)$ is uniquely determined by its constant term, which is a root of $f(\0,Y)$. We will therefore write $\Phi_\beta$ for the power series root corresponding to the root $\beta$ of $f(\0,Y)$. Moreover, since $\overline{\F}[[\X]]$ is a UFD, the above factorization is the unique factorization of $f$ into linear factors in $\overline{\F}[[\X]][Y]$.

The following corollary shows that every monic factor of $f(\X,Y)$ arises from a factor of $f(\0,Y)$, by lifting the corresponding sub-multiset of roots of $f(\0,Y)$ to their power series roots.

\begin{corollary}\label{cor: potential factor}
    Let $f(\X,Y)$ be a monic polynomial with $\deg_Y(f)=d$ and such that $f$ is Hensel-ready at $\X=\0$. Let $\cB=\{\beta_1,\dots,\beta_d\}$ be the multiset of univariate roots of $f(\0,Y)$. Let $g(\X,Y)\in\F[\X,Y]$ be monic in $Y$, with $\deg_Y(g) = k$, and suppose that $g$ divides $f$. Then there exists a multiset $S\subseteq \cB$ with $|S|=k$ such that $g=\displaystyle\prod_{\beta\in S}(Y-\Phi_{\beta})$ and $\Esym_j(S)\in\F$ for all $j\in[k]$. %
\end{corollary}

We now explain how to compute the power series $\Phi_\beta$ algorithmically. We first reduce the task to computing power series roots of bivariate polynomials. Introduce a new variable $T$ and replace each $X_i$ by $TX_i$. Define $F(T,Y):=f(T\X,Y)$, which we view as a polynomial in $T$ and $Y$ over the coefficient ring $\F[\X]$. Since $F(0,Y)=f(\0,Y)$, Hensel's lemma applies to $F$ as well. Moreover, suppose the power series root of $f(\X,Y)$ corresponding to $\beta$ has the homogeneous decomposition $\Phi_\beta(\X)=\sum_{r\geq 0}\Phi_\beta^{(r)}(\X)$, where each $\Phi_\beta^{(r)}(\X)$ is homogeneous of degree $r$. Then the corresponding power series root of $F(T,Y)$ is $\Phi_\beta(T)=\sum_{r\geq 0}\Phi_\beta^{(r)}(\X)T^r$. Thus, the coefficient of $T^r$ in $\Phi_\beta(T)$ recovers precisely the degree-$r$ homogeneous component of $\Phi_\beta(\X)$. By a slight abuse of notation, we will also write $\Phi_\beta(T)$ for the power series root of $F(T,Y)$ corresponding to $\beta$.

At this point, any standard polynomial-time algorithm for computing power series roots (such as Newton iteration or Hensel lifting) would work.
Obtaining the additional properties we need requires more work: the coefficients of the roots, and hence eventually the factors, should have constant-depth circuits, and the output circuits computing the factors should use only constants from the base field $\F$. For this, we use some results from the recent works~\cite{BKRRSS25a, BKRRSS25b} on factoring constant-depth circuits, which we summarize below.

For the rest of the section, assume that $\Char(\F)=0$ or $\Char(\F)>d$. Let $F(T,Y)\in \F[\X][T,Y]$ be monic in $Y$, with $\deg_Y(F)=d$ and $\deg_T(F)=D$, and suppose that $F$ is Hensel-ready at $T=0$. Let $\cB$ be the set of roots of $F(0,Y)$, and let $\beta$ be a root of $F(0,Y)$ of multiplicity $e$. By \cref{Hensel lemma}, $F(T,Y)$ has a power series root $\Phi_{\beta}$ of multiplicity $e$. We now give an explicit formula for $\Phi_{\beta}(T)$. Such a formula was obtained in~\cite{BKRRSS25a} when $e=1$\footnote{The factoring algorithm of \cite{BKRRSS25a} first shows how to reduce to the case $e=1$ over characteristic zero fields and fields where the characteristic does not divide $e$, and hence does not have to consider expressions for power series roots for larger values of $e$.}, and over fields of characteristic dividing $e$, it used Hasse derivatives to obtain an expression. Here, we use standard derivatives to obtain a variation of the formula in~\cite{BKRRSS25a} for general $e$. The proof follows by a close adaptation of the arguments of~\cite{BKRRSS25a}; for completeness, we include the details in the appendix.

 Note that, since the multiplicity of $\beta$ in $F(0,Y)$ is $e$, $\partial_{Y^{e-1}}F(0,Y)$ is a non-zero polynomial in $Y$. Define the Laurent series $R_e\in \F[\X][T]((Z))$ in $Z$

$$R_{e}(Z):=Z+\sum_{m\geq 0}^{2e(D+1)}[Y^{em+e-2}]\Big\{\Big(\frac{e!}{\partial_{Y^{e}}F(0,Z)}\Big)^{m+1}\frac{\partial_{Y}F(T,Y+Z)}{e}\Big(\frac{\partial_{Y^{e}}F(0,Z)}{e!}Y^e-F(T,Y+Z)\Big)^m\Big\} $$
 Here $[Y^r]\{P(Y)\}$ denotes the coefficient of $Y^r$ in $P(Y)$.

One can observe that if we set $e=1$ in the above definition, we get the $R(Z)$ defined in \cite[Theorem 4.3]{BKRRSS25a}.
For more on the background and motivation of the origins of the above expression, the reader is encouraged to look at ~\cite{BKRRSS25a}.
The exact expression for $R_{e}(Z)$ is provided for completeness, but what we really need for the algorithm is the following two properties of $R_e(Z)$:
\begin{enumerate}
     \item If $F(T,Y)=f(T\X,Y)$, then $R_e(Z)$ can be written as a ratio of $A(\X,T,Z)$ and $B(Z)$, where both polynomials are computable by polynomial-size constant-depth circuits (since $f$ is sparse).
     \item $R_e(\beta)$ computes a truncated version of $\Phi_\beta$.
 \end{enumerate}

 So in particular, we are able to obtain constant-depth circuits (with division gates) that compute suitable truncations of the power series $\Phi_\beta$.
 We formally state this fact below (and prove it in the appendix, \cref{Proof: closed form root}). In the setting of $e=1$, such a result was also provided in~\cite{BKRRSS25a}.

\begin{theorem}
    [\cite{BKRRSS25a}]\label{Closed form root}
   Let $\F$ be a field of characteristic 0 or greater than $d$.  Let $F(T,Y) \in \F[\X][T,Y]$ be monic in $Y$ and Hensel-ready at $T=0$, furthermore assume $F(T,Y)$ can be computed by a circuit of size $s'$ and depth $\Delta$. Let $\beta$ be a root of $F(0,Y)$ with multiplicity $e$, let $\deg_Y(F)=d$ and $\deg_T(F)=D$. Let $R_{e}$ be as defined above. Then \begin{enumerate}
       \item $R_e(Z)$ can be expressed as a ratio of two polynomials $A(\X,T,Z)$ and $B(Z)$ where $A(\X,T,Z)\in\F[\X,T,Z]$ and $B(Z)\in\F[Z]$. Furthermore both $A(\X,T,Z)$ and $B(Z)$ can be computed by circuits of size $\poly(n,s',d,D)$ and depth $\Delta+\cO(1)$.
       \item $B(\beta) \neq 0$.
       \item $\Phi_{\beta}(T)\Trunc T^{D+1}=R_{e}(\beta)\Trunc T^{D+1}$
   \end{enumerate} 
\end{theorem}

\subsection{Sparse Polynomials}
\label{sec:sparse-polynomials}

We say that a polynomial $f(\X)\in\F[\X]$ is \emph{$s$-sparse} if it has at most $s$ monomials with non-zero coefficients. We denote the sparsity of $f$ by $\|f\|$. Thus, if $f$ has exactly $s$ non-zero monomials, then $\|f\|=s$.

Klivans and Spielman~\cite{KS01} gave a beautiful deterministic hitting-set construction for sparse polynomials. We will use the following form of their result.

\begin{theorem}[PIT for sparse polynomials {\cite{KS01}}]
\label{thm:sparse_pit}
Let $n,s,d \in \N$ and let $\F$ be a field with $\abs{\F}\geq \poly(n,s,d)$. There exists an explicit hitting set $S_{n,s,d}\subseteq \F^n$ of size $\poly(n,s,d)$ for the class of $s$-sparse polynomials in $\F[\X_n]$ of individual degree at most $d$. That is, every non-zero polynomial in this class evaluates to a non-zero value at some point of $S_{n,s,d}$. 

As a consequence, there exists a hitting set $S_{n,s,d,l}\subseteq\F^n$ of size $\poly(n,s,d,l)$ for the class of products of $l$ $s$-sparse polynomials, each with individual degree at most $d$. Moreover, such a hitting set can be constructed deterministically in time $\poly(n,s,d,l)$.%
\end{theorem}

The Klivans--Spielman construction can also be viewed as giving a low-variate
generator for sparse polynomials, and this yields a black-box reconstruction
algorithm. A generator is a tuple
\[
    \cG(\bZ) = \bigl(\cG_1(\bZ), \ldots, \cG_n(\bZ)\bigr),
\]
where \(\bZ = (Z_1,\ldots,Z_c)\) and each \(\cG_i(\bZ) \in \F[\bZ]\). For sparse
polynomial reconstruction, one wants \(\cG\) to have low degree and to preserve
enough information about every \(s\)-sparse polynomial \(P \in \F[\X_n]\) of
individual degree at most \(d\), so that \(P\) can be recovered from black-box
access to the low-variate polynomial \(P \circ \cG\). The construction in
\cite{KS01} provides such a generator. The following theorem, which is
essentially \cite[Theorem~3.2]{chuyoon2026factorization}, states the resulting
reconstruction guarantee.

\begin{theorem}[Sparse reconstruction via the Klivans--Spielman generator]
    \label{Thm: Sparse Reconstruction}
    Let \(n,s,d \in \N\), and let \(S \subseteq \F\) be a set\footnote{We canonically choose $S$ to be the first $m$ elements in the field, according to some fixed standard enumeration.} of size
    \(|S| \eqcolon m \geq \poly(n,s,d)\). There exists a generator
    \(\cG^{KS}_{(s,d,n)}(Z_1, Z_2, Z_3, Z_4) =
    \bigl(\cG_1(\bZ_4),\ldots,\cG_n(\bZ_4)\bigr)\), such that each
    \(\cG_i \in \F[\bZ_4]\) has degree at most \(\poly(n,s,d)\), with the following
    property. Given the evaluations of \(P \circ \cG^{KS}_{(s,d,n)}\)
    on all points of \(S^4\), any \(s\)-sparse polynomial
    \(P \in \F[\X_n]\) of individual degree at most \(d\) can be reconstructed
    in monomial representation in time \(\poly(n,s,d)\). Moreover,
    \(\cG^{KS}_{(s,d,n)}\) can be constructed in time
    \(\poly(n,s,d)\).
\end{theorem}

In \cref{sec:recovering-all-factors}, we will first give an algorithm that produces a list of candidate circuits (with division gates) that are monomial equivalent to factors of the input polynomial $f$. In the post-processing step, we will remove the suitable monomial divisors from these candidates and eliminate division gates. For this purpose, we record the following simple consequences of the Klivans--Spielman hitting-set and reconstruction results above. For brevity, we shift their proofs to \cref{app:sparse-poly} in the appendix.

\begin{corollary}
\label{cor:factor_sparse_monomial_division}
Let $\F$ be a field with $\abs{\F}\geq \poly(n,s,d)$, and let $i\in[n]$.
Suppose $f(\X_n)$ is a factor of a non-zero $s$-sparse polynomial of individual degree at most $d$, and suppose additionally that $f$ is not divisible by a monomial.
Let $C$ be an algebraic circuit of size $m$ and degree $D$ computing a non-zero polynomial $g(\X_n)\in\F[\X_n]$ such that $f \prec_d g$. Then there exists a deterministic algorithm that, given $C$, computes the largest power of $X_i$ dividing $g$ in time $\poly(n,s,d, D, m)$.
\end{corollary}

\begin{lemma}[Division elimination]
    \label{lem:division-elimination}
    Let $C$ be a circuit of size $m$ and depth $\Delta$ with division gates over
    $\F$, computing a polynomial $f \in \F[\X_n]$ of individual degree at most
    $d$. Suppose that every denominator appearing at a division gate of $C$ is a
    product of $s$-sparse polynomials of individual degree at most $d$. Then
    there is a deterministic algorithm that runs in $\poly(n,s,m,d)$ time and
    outputs a circuit without division gates, of size $\poly(m,d)$ and depth
    at most $\Delta+\cO(1)$, computing $f$.
\end{lemma}

Our algorithms will output a list of candidate factors of sparse polynomials. In general, this list may contain spurious candidates that do not correspond to actual factors of the input polynomial. In some settings, such as the problem of outputting all sparse factors considered in \cite{chuyoon2026factorization}, these spurious candidates can be removed using divisibility tests. We record below a result of Volkovich~\cite{Volkovich17}, which provides an efficient deterministic divisibility test in the bounded individual degree regime.

\begin{theorem}[\cite{Volkovich17}]
\label{thm:sparse-divisibility-test}
Let $f,g \in \F[\X_n]$ be $s$-sparse polynomials of individual degree at most $d$, and suppose that $\Char(\F)=0$ or $\Char(\F)>d$. Then there exists a deterministic algorithm that, given black-box access to $f$ and $g$, decides whether $g$ divides $f$ in time $\poly(n,s^d,d!)$.
\end{theorem}

We can also use the above theorem to prune out the list and output \emph{all} the factors in the bounded individual degree regime, in quasi-polynomial time. For this, we will need the following result of Bhargava, Saraf and Volkovich~\cite{BSV20}: bounded individual degree factors of sparse polynomials have at most quasi-polynomial sparsity.

\begin{theorem}[Sparsity Bound]{\cite[Theorem 4.1]{BSV20}}\label{thm:bsv_general_sparse}
    Let $f \in \F[\X_n]$ be a $s$-sparse polynomial. Suppose $g$ is a factor of $f$ where $g$ has individual degree at most $d$. Then, the sparsity of $g$ is at most $s^{4 d^2 \log n}$.
 \end{theorem}

 \begin{remark}
     \cite[Theorem 4.1]{BSV20} is stated when $f$ has individual degree $d$. But on inspecting their proof, it is clear that it extends to the setting of \cref{thm:bsv_general_sparse}.
 \end{remark}

Finally in \cref{sec:recovering-bounded-individual-degree-factors}, we use the divisibility test by Forbes~\cite{Forbes15} and prune out the list to obtain all constant degree factors of general sparse polynomials. This recovers the results of Kumar, Ramanathan and Saptharishi and Dutta, Sinhababu and Thierauf~\cite{KRS24, DST24} for sparse polynomials.

\begin{theorem}[{\cite{Forbes15}}]
    \label{thm:divisibility-testing-constant-deg}
    Let $\F$ be a field of characteristic $0$ or characteristic at least
    $\poly(D)$. Let $f \in \F[\X_n]$ be an $n$-variate $s$-sparse polynomial of
    degree at most $D$, and let $g \in \F[\X_n]$ be a polynomial of constant degree.
    Then there is a deterministic algorithm that computes the multiplicity of $g$
    as a factor of $f$ in time $\poly(s,n,D)^{\cO(\log s)}$.
\end{theorem}

\section{Factoring Normalized and Reverse Normalized Sparse Polynomials}
\label{sec:factor-normalized-and-reverse-normalized-polynomial}

In this section, we prove the factorization subroutine used by our main algorithm. The input is a sparse polynomial
\[
    f(\X,Y)=f_d(\X)Y^d+\cdots+f_0(\X)
\]
of bounded individual degree, given in its monomial representation. We first handle the normalization \(\hat f\), which is monic in \(Y\), and later reduce the reverse-normalized case to the same setting using polynomial reversals. We assume throughout this section that \(\Char(\F)=0\) or \(\Char(\F)>d\).

The goal is to produce a small list of constant-depth circuits that contains every non-constant factor of the normalized polynomial \(\hat f\), while allowing spurious candidates.

\begin{lemma}
    \label{lem:monic-factorization}
    Let $\F$ be a field of characteristic $0$ or greater than $d$. Let $f(\X, Y) \in \F[\X, Y]$ be an $n$-variate polynomial with sparsity at most $s$ and individual degree at most $d$, and let $\hat f(\X, Y)$ be its normalization with respect to $Y$. Then there is an algorithm running in time $\poly(s^d, d!, n, \mathcal{T}(\F,d))$ that outputs a list $\cL$ such that
    \begin{enumerate}
        \item $|\cL|\leq 2^d-1$
        \item every element of $\cL$ is a $\poly(n,s,d)$-size, $\cO(1)$-depth circuit over $\F$
        \item every non-constant factor $\hg(\X,Y)$ of $\hf(\X,Y)$ is computed by some circuit in $\cL$
    \end{enumerate}
    The list $\cL$ may contain spurious candidates. Moreover, the algorithm can compute the $Y$-degree of the polynomial computed by every circuit in $\cL$.
\end{lemma}

Before proving the lemma, we establish the main preprocessing step: finding a point $\balpha\in\F^n$ at which $\hf(\X,Y)$ is Hensel-ready, in the sense of \cref{def:hensel-ready}. After shifting $\X$ by this point, the polynomial can be factored through its power-series roots around the origin.

\paragraph{Finding a Hensel-ready shift.}

Let us first consider the easier case where $\hf$ is square-free. By \cref{lem:resultant-prop}, this is equivalent to $\Disc_Y(\hf)\neq 0$, or equivalently to $\Syl_Y(\hf,\partial_Y\hf)$ having full rank. Since the normalization step can increase sparsity from $s$ to $s^d$, every entry of this Sylvester matrix is $s^d$-sparse, and hence $\Disc_Y(\hf)$ is $(2d)!s^{d^2}$-sparse. Therefore, any point $\balpha\in\F^n$ for which $\Disc_Y(\hf)(\balpha)\neq0$ makes $\hf$ Hensel-ready at $\X=\balpha$.

In general, $\hf$ need not be square-free, so $\Disc_Y(\hf)$ may vanish identically. In this case, we choose $\balpha$ so that the rank of the Sylvester matrix $\Syl_Y(\hf,\partial_Y\hf)$ is preserved after substituting $\X=\balpha$. To be precise, as observed in \cite[Lemma 4.19]{chuyoon2026factorization}, it is enough to find a point $\balpha$ such that a nonzero maximum-rank minor of $\Syl_Y(\hf,\partial_Y\hf)$ remains nonzero after substituting $\X=\balpha$. Such a minor is again $(2d)!s^{d^2}$-sparse.

Naively, this would lead to a running time polynomial in $s^{d^2}$. The key point is that the relevant rank-preservation test can be carried out before normalization. Let $\bar f(\X,Y)=f(\X,Y/f_d)$. By \cref{lem:Sylvester-substitution-identity}, the minors of $\Syl_Y(\bar f,\partial_Y\bar f)$ are related, up to diagonal row and column scalings, to the corresponding minors of $\Syl_Y(f,\partial_Yf)$. Since $\hat f=f_d^{d-1}\bar f$ and $f_d$ is independent of $Y$, preserving rank for $\bar f$ is equivalent to preserving rank for $\hat f$ at points where $f_d$ is nonzero. Thus the point $\balpha$ can be found by working with the Sylvester matrix of the original sparse polynomial $f$, which is what keeps the running time polynomial in $s^d$.

We formalize this preprocessing step in the next lemma. It is a variant of \cite[Lemma 4.19]{chuyoon2026factorization}, adapted to our notation and parameter setting. We include the proof because the lemma is used here in a slightly different context.
\begin{lemma}\label{Algo:to hit disc}
    Let $\F$ be a field of characteristic $0$ or greater than $d$. Let $f(\X,Y)\in\F[\X,Y]$ be an $s$-sparse polynomial of individual degree $d$, and let $\hf(\X,Y)$ be the normalization of $f$ with respect to the variable $Y$. Then there is a deterministic algorithm that runs in $\poly(n,d!,s^d)$ time and finds $\balpha\in\F^n$ such that $\hat{f}$ is Hensel-ready at $\X=\balpha$.
\end{lemma}
\begin{proof}

    The matrix $\Syl_Y(f,\partial_Yf)$ has size at most $(2d-1)\times(2d-1)$. Hence it has at most $2^{O(d)}$ minors, and each minor expands to a polynomial with sparsity at most $(2d)!s^{O(d)}$. The algorithm will find $\balpha\in\F^n$ that keeps every nonzero minor of $\Syl_Y(f,\partial_Y f)$ nonzero after substituting $\X=\balpha$, and moreover keeps $f_d(\balpha)$ nonzero. By \cref{thm:sparse_pit}, such a point can be found in $\poly(d!s^d,n)$ time. 
    Let $\bar f(\X,Y)=f(\X,Y/f_d)$. By \cref{lem:Sylvester-substitution-identity}, the choice of $\balpha$ preserves the rank of $\Syl_Y(\bar f,\partial_Y\bar f)$. Since $\hat f=f_d^{d-1}\bar f$ and $f_d$ is independent of $Y$, the Sylvester matrices $\Syl_Y(\hat f,\partial_Y\hat f)$ and $\Syl_Y(\bar f,\partial_Y\bar f)$ differ only by multiplication by powers of $f_d$. Thus, over $\F(\X)$, they have the same rank, and after substituting $\X=\balpha$, where $f_d(\balpha)\neq0$, rank preservation for $\bar f$ implies rank preservation for $\hat f$. In particular, the rank of $\Syl_Y(\hat{f},\partial_Y\hat{f})$ over $\F(\X)$ is the same as the rank of $\Syl_Y(\hat{f},\partial_Y\hat{f})(\balpha)$.

    We now show that $\hat{f}$ is Hensel-ready at $\X=\balpha$.
    Let us assume the irreducible factorization of $\hat{f}$ is $\hat{f}=h_1(\X,Y)^{e_1}\dots h_k(\X,Y)^{e_k}$ and $\deg_Y(h_i)=r_i$. We choose each $h_i$ to be monic in $Y$. Then from \cref{cor:resultant-prop} we get that 
    $$\dim(\ker (\Syl_Y(\hat{f},\partial_Y\hat{f})))=\sum_{i=1}^k(e_i-1)r_i=:R.$$ 
    Since $\balpha$ preserves the rank of the Sylvester matrix, we have $\dim(\ker(\Syl_Y(\hat{f}({\balpha},Y),\partial_Y\hat{f}({\balpha},Y))))=R$. Assume the irreducible factorization of $\hat f({\balpha},Y)$ is $p_1(Y)^{l_1}\dots p_r(Y)^{l_r}$ and $\deg(p_i)=b_i$. Then again from \cref{cor:resultant-prop} we have $R=\sum_{j=1}^r(l_j-1)b_j$.

    We will now compare the above expression for the kernel dimension $R$ with the way the irreducible factors $p_j(Y)$ distribute among the substituted factors $h_i(\balpha,Y)$. This will show that no $p_j(Y)$ appears in two different substituted factors $h_i(\balpha,Y)$, and that it appears with multiplicity at most one inside whichever $h_i(\balpha,Y)$ contains it. In other words, each $h_i(\balpha,Y)$ is square-free, and for $i\neq j$, the polynomials $h_i(\balpha,Y)$ and $h_j(\balpha,Y)$ are coprime.

    Let the irreducible factorization of each $h_i(\balpha,Y)$ be $$h_i(\balpha,Y)=p_1(Y)^{l_{i,1}}\dots p_r(Y)^{l_{i,r}}.$$ Since the $h_i$ are monic in $Y$, substituting $\X=\balpha$ preserves their $Y$-degrees. Clearly $\displaystyle\sum_{i=1}^k e_il_{i,j}=l_j$ and $r_i=\sum_{j=1}^rl_{i,j}b_j$. So \begin{align*}
        R&=\sum_{j=1}^r(l_j-1)b_j\\
        &=\sum_{j=1}^r((\sum_{i=1}^ke_il_{i,j})-1)b_j\\
        &\geq \sum_{j=1}^r(\sum_{i=1}^k(e_i-1)l_{i,j})b_j\tag{as $\sum_il_{i,j}\geq1$}\\
        &=\sum_{i=1}^k(e_i-1)\sum_{j=1}^rl_{i,j}b_j\\
        &=\sum_{i=1}^k(e_i-1)r_i=R
    \end{align*}
    Hence the inequality must be an equality, and in particular $\sum_il_{i,j}=1$ for every $j$. This means that for all $j$, there is only one $i$ such that $l_{i,j}=1$, and all other $l_{i',j}=0$. In other words, $p_j(Y)$ divides only one $h_i(\balpha,Y)$, and it does so with multiplicity one. Thus the polynomials $h_i(\balpha,Y)$ are pairwise coprime, and each $h_i(\balpha,Y)$ is square-free. By the standard resultant and discriminant criteria from \cref{ResultantFact}, this implies that for all $i\neq j$, $\Res_Y(h_i,h_j)(\balpha)\neq0$, and for every $i$, $\Disc_Y(h_i)(\balpha)\neq0$. These are precisely the conditions for $\hat f$ to be Hensel-ready at $\X=\balpha$.
\end{proof}

To apply Hensel's lemma (\cref{Hensel lemma}), we need the polynomial to be Hensel-ready at $\X=\mathbf{0}$. We therefore shift $\X\mapsto \X+\balpha$ in $\hf$. This shift may destroy sparsity, but sparsity is only needed to find a Hensel-ready point. Once the shift has been found, the power-series factorization no longer uses sparsity. After computing the potential factors of the shifted polynomial, we will shift back by $\X\mapsto \X-\balpha$ to recover the corresponding factors of the original polynomial.

\paragraph{Lifting factors from power-series roots.}

Let $\hat{f}_{\balpha}(\X,Y)=\hat{f}(\X+\balpha,Y)$. Then $\hat{f}_{\balpha}$ is monic in $Y$, has $\deg_Y(\hat{f}_{\balpha})=d$, has total $\X$-degree at most $D\leq nd^2$, and is Hensel-ready at $\X=\mathbf{0}$. Moreover, we can construct a constant-depth circuit for $\hat{f}_{\balpha}$, since $f$ is sparse and the operations used to obtain $\hat{f}_{\balpha}$ from $f$ can be implemented in constant depth.

Let $\cB$ be the multiset of roots of $\hat{f}_{\balpha}(0,Y)$. By \cref{cor: potential factor}, every factor of $\hat{f}_{\balpha}(\X,Y)$ arises by lifting a factor of $\hat{f}_{\balpha}(0,Y)$. More concretely, if $S\subseteq\cB$ is a sub-multiset and $q(Y)=\prod_{\beta\in S}(Y-\beta)$ is a factor of $\hat{f}_{\balpha}(0,Y)$, then $\prod_{\beta\in S}(Y-\Phi_\beta)$, truncated up to $\X$-degree $D$, is a potential factor of $\hat{f}_{\balpha}$; here $\Phi_\beta$ is the unique power-series root of $\hat{f}_{\balpha}(\X,Y)$ with constant term $\beta$ given by \cref{Hensel lemma}. This potential factor is a true factor precisely when the corresponding product of power series, $\prod_{\beta\in S}(Y-\Phi_\beta)$, is in fact a polynomial. We now show how to construct small constant-depth circuits for all such potential factors.

\cite{BKRRSS25a} showed how to compute the homogeneous components of $\Phi_\beta$ in constant depth. Following their approach, we introduce a fresh variable $T$ and replace each $X_i$ by $TX_i$. We then view the resulting polynomial as a bivariate polynomial in $T$ and $Y$, with coefficients in $\F[\X]$. We refer to \cref{Subsec: Power-series root} for a more complete discussion. The key point is that the Laurent series $R_e(Z)$ from \cref{Closed form root}, viewed as a series in $Z$ whose coefficients depend on the variables $\X$ and $T$, can be represented as a ratio of two constant-depth circuits. Substituting $Z=\beta$ gives the truncated power-series root $\Phi_\beta(T)$, and then setting $T=1$ gives $\Phi_\beta(\X)$ truncated up to degree $D$.

Define $F(T,Y)=\hat{f}_{\balpha}(TX_1,TX_2,\dots,TX_n,Y) \in \F[\X][T,Y]$. Since $F(0,Y)=\hat{f}_{\balpha}(0,Y)$, the multiset of roots of $F(0,Y)$ is exactly $\cB$. For $\beta\in\cB$, let $\Phi_{\beta}(T)\in \F(\beta)[\X][[T]]$ denote the power-series root of $F(T,Y)$ corresponding to $\beta$. Then the set $\{\Phi_{\beta}(1):\beta\in\cB\}$ is exactly the set of power-series roots of $\hat f_{\balpha}(\X,Y)$.

Let $q(Y)$ be a factor of $F(0,Y)$. As discussed above, $q(Y)$ potentially lifts to a factor of $F(T,Y)$. We formalize this correspondence in the following definition.

\begin{definition}
    Let $q(Y)$ be a factor of $F(0,Y)$ and $S$ be the multiset of roots of $q(Y)$. We define the \textit{factor-lift of $q(Y)$} to be $\prod_{\beta\in S}(Y-\Phi_{\beta}(T))\Trunc T^{D+1}$.
\end{definition}
The next lemma records that true factors are recovered by this lifting operation.

\begin{lemma}\label{PotentialFactorCorrect}
    If $g(\X,Y)$ is a factor of $\hat{f}_{\balpha}(\X,Y)$, then $g(\X,Y)$ is obtained from the factor-lift of $g(0,Y)$ by setting $T=1$.
\end{lemma}
\begin{proof}
    By \cref{cor: potential factor}, if $S$ is the multiset of roots of $g(0,Y)$, then $g(T\X,Y)=\prod_{\beta\in S}(Y-\Phi_\beta(T))$ as a power series in $T$. Since $g(\X,Y)$ is a polynomial of total $\X$-degree at most $D$, the polynomial $g(T\X,Y)$ has $T$-degree at most $D$. Therefore \begin{align*}
        g(T\X,Y)&=\prod_{\beta\in S}(Y-\Phi_\beta(T))\\
        &=\prod_{\beta\in S}(Y-\Phi_\beta(T))\Trunc T^{D+1}
    \end{align*}
    Setting $T=1$ proves the lemma.
\end{proof}

For any $\beta \in S$, let $e(\beta)$ be the multiplicity of $\beta$ in $F(0,Y)$.
Using \cref{Closed form root}, we can rewrite the factor-lift as \begin{align*}
    \prod_{\beta\in S}(Y-\Phi_{\beta}(T))\Trunc T^{D+1}&=\prod_{\beta\in S}(Y-(\Phi_{\beta}(T)\Trunc T^{D+1}))\Trunc T^{D+1}\\
    &=\prod_{\beta\in S}(Y-(R_{e(\beta)}(\beta)\Trunc T^{D+1}))\Trunc T^{D+1}\\
    &=\prod_{\beta\in S}(Y-R_{e(\beta)}(\beta))\Trunc T^{D+1}
\end{align*}
The rational function $R_e(Z)$ has a representation as a ratio of constant-depth circuits in the variables $\X,T,Z$, where $Z$ is the main variable and $\X,T$ are retained as formal variables. Moreover, its denominator depends only on $Z$. Hence, after substituting $Z=\beta$, the denominator becomes a constant in the field extension containing $\beta$. This gives a constant-depth representation of $R_e(\beta)$ as a series in $\X$ and $T$. Since multiplication and truncation can also be carried out in constant depth, every factor-lift of a factor of $F(0,Y)$ has a constant-depth representation over a suitable field extension. The remaining issue is that these circuits may use constants such as $\beta$ outside the base field $\F$. To remove these extension-field constants, we use ideas from \cite{BKRRSS25a}, and in particular \cref{lem:esym-of-g/h-on-roots-of-f}, to obtain constant-depth representations of the same factor-lifts using only constants from $\F$.

\paragraph{Ensuring output circuits are over the base field.}

We now formalize the base-field conversion step. The input to this step is a factor $q(Y)$ of $F(0,Y)$, together with its factorization over $\F$. Although the factor-lift is naturally described using the roots of $q(Y)$ over an extension field, the following lemma shows that it can be computed by a constant-depth circuit over the original field $\F$.

\begin{lemma}\label{Algo: PotentialFactorComputing}
 Let $F(T,Y)$ be as above, with $\deg_Y(F)=d$. Let $q(Y) \in \F[Y]$ be a factor of $F(0,Y)$, and suppose its irreducible factorization over $\F$ is $q(Y)=q_1(Y)^{l_1}\dots q_r(Y)^{l_r}$. Then there is an algorithm that takes as input a size $s'$ constant-depth circuit for $F(T,Y)$ and the irreducible factorization of $q(Y)$, and outputs a constant-depth circuit over $\F$ for the factor-lift of $q(Y)$ in $\poly(n,s',d)$ time.
\end{lemma}
\begin{proof}

    We first express the factor-lift in terms of the power-series root formula from \cref{Closed form root}, and then explain why the resulting expression can be written over the base field.

    For each $i\in[r]$, let $d_i=\deg(q_i)$, let $S_i$ be the set of roots of $q_i$ over $\overline{\F}$, and let $e_i$ be the multiplicity of $q_i$ as a factor of $F(0,Y)$. Since $q_i$ is irreducible, every $\beta\in S_i$ has multiplicity $e_i$ as a root of $F(0,Y)$. As $q_i$ appears in $q$ with multiplicity $l_i$, the factor-lift of $q(Y)$ is
    \[
        \prod_{i=1}^r
        \left(
            \prod_{\beta\in S_i}(Y-R_{e_i}(\beta))
        \right)^{l_i}
        \Trunc T^{D+1}.
    \]

    Fix $i\in[r]$. By \cref{Closed form root}, $R_{e_i}(Z)$ can be written as a rational function $A_i(\X,T,Z)/B_i(Z)$, where $A_i$ and $B_i$ have $\poly(n,s',d)$-size, $\cO(1)$-depth circuits constructible in $\poly(n,s',d)$ time, and $B_i(Z)$ depends only on the main variable $Z$. Moreover, \cref{Closed form root} guarantees that $B_i(\beta)\neq0$ for every $\beta\in S_i$, since each such $\beta$ has multiplicity exactly $e_i$ in $F(0,Y)$. The product $\prod_{\beta\in S_i}(Y-R_{e_i}(\beta))$ is symmetric in the roots in $S_i$. Expanding it gives
    \[
        \prod_{\beta\in S_i}(Y-R_{e_i}(\beta))
        =
        Y^{d_i}+\sum_{a=1}^{d_i}(-1)^{a}
        \Esym_a\Big\{R_{e_i}(\beta) \mid \beta\in S_i\Big\}Y^{d_i-a}.
    \]

    Applying \cref{lem:esym-of-g/h-on-roots-of-f} to the irreducible polynomial $q_i$ and the rational function $A_i/B_i$, we can construct, for every $a\in[d_i]$, a constant-depth circuit (with division gate at the top) over $\F$ for the elementary symmetric expression $\Esym_a\{R_{e_i}(\beta)\mid \beta\in S_i\}$. To make the denominator clearing explicit, let $P_i=\prod_{\beta\in S_i}B_i(\beta)$. Multiplying $\Esym_a\{R_{e_i}(\beta)\mid \beta\in S_i\}$ by $P_i$ clears all denominators in each summand, and the resulting expression is symmetric in the roots of the monic polynomial $q_i$. Hence, by the fundamental theorem of symmetric polynomials, it lies in $\F[\X,T]$. Also $P_i=\Res(q_i,B_i)$ (by the Poisson formula for resultants) is a nonzero element of $\F$. We can therefore absorb $P_i^{-1}$ into the circuit constants. It follows that the coefficients of $\prod_{\beta\in S_i}(Y-R_{e_i}(\beta))$ have constant-depth circuits over $\F$, constructible in $\poly(n,s',d)$ time.

    We apply this construction for every irreducible factor $q_i$, raise the corresponding lifted factor to the power $l_i$, multiply these expressions, and truncate at $T^{D+1}$. These operations preserve constant depth and can be carried out in $\poly(n,s',d)$ time (using product gates for powers and products, and \cref{lem:circuit-truncation} for the final truncation), giving the desired circuit for the factor-lift of $q(Y)$.
\end{proof}
We remark that another way to prove the above lemma is to make use of the fact that the factor-lift of $q(Y)$ is a multi-symmetric polynomial in its roots. Although the factor-lift of $q(Y)$ need not be symmetric in all roots of $q(Y)$ at once, it is symmetric separately in the roots of each irreducible factor $q_i(Y)$. Thus Theorem 3.6 of~\cite{BKRRSS25b}, which generalizes \cref{lem:esym-of-g/h-on-roots-of-f} to multi-symmetric polynomials, gives an alternative route to the same constant-depth construction over $\F$.

Now we are ready to prove Lemma~\ref{lem:monic-factorization}.
\begin{proof}[Proof of \cref{lem:monic-factorization}]
    Given $f(\X,Y)$ in monomial representation, we first compute the normalization $\hat f(\X,Y)$. This has at most $s^d$ monomials and can be computed efficiently. We then apply \cref{Algo:to hit disc} to find $\balpha\in\F^n$ such that $\hat f_{\balpha}(\X,Y)=\hat f(\X+\balpha,Y)$ is Hensel-ready at $\X=\mathbf{0}$. As discussed above, we can construct a constant-depth circuit for $\hat f_{\balpha}$, and hence also for $F(T,Y)=\hat f_{\balpha}(T\X,Y)$.

    We factor the univariate polynomial $F(0,Y)$ using the deterministic univariate factorization algorithm summarized in \cref{UnivariateFactorization}. Starting with an empty list $\cL$, for every non-constant factor $q(Y)$ of $F(0,Y)$, we compute the factor-lift $\hat q(\X,T,Y)$ using \cref{Algo: PotentialFactorComputing}, and add $\hat q(\X-\balpha,1,Y)$ to $\cL$. Since $F(0,Y)$ has degree at most $d$, it has at most $2^d-1$ non-constant factors. Thus $|\cL|\leq 2^d-1$, and the whole procedure runs in time $\poly(n,d!,s^d,\cT(\F,d))$.

    Finally, let $g(\X,Y)$ be any non-constant factor of $\hat f(\X,Y)$. Then $g_{\balpha}(\X,Y)=g(\X+\balpha,Y)$ is a factor of $\hat f_{\balpha}(\X,Y)$. By \cref{PotentialFactorCorrect}, $g_{\balpha}$ is obtained by setting $T=1$ in the factor-lift of the univariate factor $g_{\balpha}(0,Y)$ of $F(0,Y)$. Therefore, when the algorithm processes $g_{\balpha}(0,Y)$, it adds $g_{\balpha}(\X-\balpha,Y)=g(\X,Y)$ to the list. This proves correctness.
\end{proof}

We now show how factoring reverse-normalized polynomials can be reduced to factoring normalized polynomials, and thus show how to recover an analog of Lemma~\ref{lem:monic-factorization} for the case of reverse-normalized polynomials.
\begin{corollary}
    \label{cor:reversed-monic-factorization}
    Let $\F$ be a field of characteristic $0$ or greater than $d$. Let $f(\X, Y) \in \F[\X, Y]$ be an $n$-variate polynomial with sparsity at most $s$ and individual degree at most $d$. Suppose $\tilde f(\X, Y) \in \F[\X, Y]$ is the reverse-normalized form of $f$ with respect to $Y$. Then there is an algorithm running in time $\poly(s^d, d!, n, \mathcal{T}(\F,d))$ that outputs a list $\cL$ such that \begin{enumerate}
        \item $|\cL|\leq 2^d-1$
        \item every element of $\cL$ is a $\poly(n,s,d)$-size, $\cO(1)$-depth circuit over $\F$
        \item every non-constant factor $\tg(\X,Y)$ of $\tf(\X,Y)$ is computed by some circuit in $\cL$
    \end{enumerate}
    The list $\cL$ may contain spurious candidates. Moreover, the algorithm can compute the $Y$-degree of the polynomial computed by every circuit in $\cL$.
\end{corollary}
\begin{proof}
    Let $g=\rev_Y(f)$. The normalization of $g$ is exactly $\rev_Y(\tilde f)$. Since $g=\rev_Y(f)$ is still $s$-sparse and has individual degree at most $d$, we can apply \cref{lem:monic-factorization} to $g$ and obtain a list of circuits containing all non-constant factors of $\rev_Y(\tilde f)$. By \cref{lem: reversefactor}, reversing these candidate circuits gives a list containing all non-constant factors of $\tilde f$. Reversal preserves the list size, constant depth, and the claimed size bound (using coefficient extraction in $Y$, as in \cref{lem:circuit-coefficient-interpolation}), and the $Y$-degree of each reversed candidate is computable from the corresponding candidate for $\rev_Y(\tilde f)$.
\end{proof}

\section{Recovering \emph{all} Factors of Sparse Polynomials of Bounded Individual Degree}
\label{sec:recovering-all-factors}

In this section, we prove \cref{thm:main1}. In particular, we obtain a deterministic polynomial-time algorithm that recovers \emph{all} factors $g$ (not divisible by a monomial) of an $n$-variate $s$-sparse polynomial $f \in \F[\X_n]$ of bounded individual degree $d$. The algorithm outputs an $s^d$-size list $\cL$ of constant-depth circuits such that each factor $g$ is computed by some circuit in $\cL$. $\cL$ might contain spurious circuits that do not compute any factor of $f$. Unfortunately, we are unable to prune out the spurious elements, since we do not know of efficient deterministic algorithms to check if a constant-depth circuit divides $f$ (see \cref{sec:future-directions-and-open-questions}). Note that the restriction of outputting factors not divisible by a monomial is unavoidable (see \cref{rem:factor-not-divisible-by-monomial}). Since we are given white-box access to the polynomial $f$, we compensate for this restriction by outputting the maximal degree monomial factor dividing $f$. %
Throughout this section, when a list consists of circuits, we often identify a circuit in the list with the polynomial it computes.

\begin{theorem}[Main Theorem 1]
\label{thm:compute-all-factors-of-bounded-degree}
Let $f(\X_n)\in \F[\X_n]$ be an $s$-sparse polynomial with individual degree at most $d$, where $\Char(\F)$ is $0$ or greater than $d$. Then, there is an algorithm that runs in time $\poly(s^d, d!, n)$ and outputs an $s^d$-size list of circuits $\cL$ such that every factor of $f$ that is not divisible by a monomial is computed by some circuit in $\cL$. The algorithm also outputs the monomial of maximal degree dividing $f$. Moreover, each circuit in $\cL$ is over $\F$ and has size $\poly(s^d, n)$ and depth $\cO(1)$.
\end{theorem}
\begin{remark}
    \label{rem:dependence-on-field}
    In the above theorem, the displayed running time suppresses the dependence on the time taken to perform univariate polynomial factorization.
    More precisely, over a field $\F$, the algorithm runs in time $\poly(s^d, d!, n) \cdot \cT(\F, d)$ (see \cref{def:time-taken-to-factor-univariates}). In particular, over $\Q$, if the input polynomial has bit complexity $B$, then the algorithm runs in time $\poly(s^d, d!, n, B)$.
    If $\F$ is a finite field with characteristic $p$ (greater than $d$), the algorithm runs in time $\poly(s^d, d!, n, p, \log{\abs{\F}})$.
\end{remark}

The algorithm runs in two stages, the main recursive \cref{algo:monomial-equivalent-factors} and the post-processing \cref{algo:non-monomial-factors}. \cref{algo:monomial-equivalent-factors} outputs a set $\cL_{mon}$ of \emph{monomial equivalent} factors of $f$ as constant-depth circuits (with division gates). More precisely, for every factor $g$ of $f$, there exists a polynomial $h$ computed by some circuit in $\cL_{mon}$ such that $g \prec_d h$ (see \cref{def:monomial-equivalent-factor}). So for a factor $g$ with no monomial divisor, $h=gM$ will be computed by some circuit $C$ in $\cL_{mon}$, where $M$ is a monomial with individual degree at most $d$. \cref{algo:non-monomial-factors} strips out this monomial $M$ from circuit $C$ by iteratively using Klivans-Spielman PIT on factors of sparse polynomials (see \cref{cor:factor_sparse_monomial_division}), thus obtaining the promised list $\cL$ containing factors\footnote{On members of the list that do not correspond to true factors, we do not guarantee what might come out of this process, but it does not affect our run time or the guarantee on true factors.}. Additionally, \cref{algo:non-monomial-factors} will eliminate all division gates.

We formally argue the correctness and the runtime of \cref{algo:monomial-equivalent-factors} below. We refer the reader to \cref{sec:proof-overview} for a high-level overview of the key ideas in the algorithm.

\begin{theorem}[Theorem for the Recursive Algorithm]
    \label{thm:recursive-algorithm}
    Let $f(\X_n) \in \F[\X_n]$ be an $s$-sparse polynomial with individual degree at most $d$, where $\Char(\F)$ is $0$ or greater than $d$. There is an algorithm that runs in time $\poly(s^d, d!, n)$ and outputs an $s^{d}$-size list $\cL_{mon}$ of circuits with division gates. For every factor $g$ of $f$, there exists a polynomial $h$ computed by some circuit in $\cL_{mon}$ such that $g \prec_d h$.

    Moreover, each circuit in $\cL_{mon}$ is over $\F$ and has size $\poly(s^d, n)$ and depth $\cO(1)$.
\end{theorem}

\begin{proof}
    We will show that \cref{algo:monomial-equivalent-factors} outputs the required list of circuits.
    We prove the theorem by induction on the sparsity of $f$. \cref{algo:monomial-equivalent-factors} initially strips off the maximal monomial divisor of $f$. Thus, from this point on, we work with an $f$ that has no nonconstant monomial divisor.

    \paragraph{Base case ($||f|| \leq 1$):}
    Then $f$ is a constant, since we have stripped out the maximal monomial divisor.
    \cref{algo:monomial-equivalent-factors} outputs $\{f \cdot (X_1 \cdot X_2 \cdots X_n)^d\}$, and thus gives a $d$-monomial equivalent polynomial to the constant factor $1$.

    \paragraph{Induction case ($||f|| \leq s$):}
    Without loss of generality, $f$ depends on the variable $X_n$ (\cref{algo:monomial-equivalent-factors} ensures this by renaming variables). Let $k = \deg_{X_n}(f)$, and
    \begin{equation*}
        f = \sum_{i=0}^k f_i(\X_{n-1})X_n^i
    \end{equation*}
    where $f_i(\X_{n-1}) \in \F[\X_{n-1}]$. Since we stripped off the maximal monomial divisor, $f_0(\X_{n-1})$ and $f_k(\X_{n-1})$ are both non-zero. By averaging, either $||f_0(\X_{n-1})||$ or $||f_k(\X_{n-1})||$ is at most $s/2$. We first deal with the case where $||f_k(\X_{n-1})||$ is at most $s/2$.

    \paragraph{Sparsity of $f_k \leq s/2$:}
    \cref{algo:monomial-equivalent-factors} normalizes $f$ to get $\hat{f}(\X_n) = f(\X_{n-1},X_n/f_k)\cdot f_k^{k-1}$, which is monic in $X_n$ and has sparsity at most $s^d$.
    The algorithm invokes \cref{lem:monic-factorization} to compute a $(2^d -1)$-size list $\cL_{\hat f}$ of circuits such that every factor of $\hat f$ is computed by some circuit in $\cL_{\hat f}$.
    \begin{equation}
        \abs{\cL_{\hat f}} \leq 2^d -1 \label{eqn:bounded-sparse-L-hat-f-size}
    \end{equation}
    Moreover, this computation takes $\poly(d!, s^d, n)$ time and each circuit in $\cL_{\hat f}$ is over $\F$ and is of size $\poly(s^d, n)$ and depth $\cO(1)$.

    Next it recursively computes the list $\cL_{f_k, mon}$ of circuits corresponding to monomial equivalent factors of $f_k$. For every factor $g$ of $f_k$, there exists a polynomial $h$ computed by some circuit in $\cL_{f_k, mon}$ such that $g \prec_d h$. By our induction hypothesis, since $||f_k|| \leq s/2$, we have $\abs{\cL_{f_k, mon}} \leq (s/2)^d$ and we can compute this list in time $\poly(d!, (s/2)^d, n)$.
    \begin{equation}
        \abs{\cL_{f_k, mon}} \leq (s/2)^d \label{eqn:bounded-sparse-L-fk-mon-size}
    \end{equation}
    Moreover, each circuit in $\cL_{f_k, mon}$ is over $\F$ and has size $\poly((s/2)^d, n)$ and depth $\cO(1)$.

    Finally, the algorithm uses the lists $\cL_{\hat f}$ and $\cL_{f_k, mon}$ to construct the list $\cL_{mon}$. It defines $\cL_{mon}$ as follows.
    \begin{equation*}
        \cL_{mon} =  \left\{\hat g(\X_{n-1}, X_n f_k)\cdot \frac{h}{f_k^{\deg_{X_n}(\hat g)}} \,\middle\vert\, \hat g \in \cL_{\hat f}, h \in \cL_{f_k,mon}\right\} \cup \cL_{f_k,mon}
    \end{equation*}
    where $\deg_{X_n}(\hat g)$ is obtained from \cref{lem:monic-factorization}. We will argue next that for any factor of $f$, there is a $d$-monomial equivalent polynomial in $\cL_{mon}$.

    Let $g$ be a factor of $f$. Suppose $g$ does not depend on the variable $X_n$. Then $g$ divides $f_k$. Hence, there exists a $g'$ in $\cL_{f_k, mon}\subset \cL_{mon}$ such that $g \prec_d g'$.

    Now suppose $g$ depends on $X_n$. Let $g$ be of the following form,
    \begin{equation*}
        g(\X_n) = \sum_{i=0}^t g_i(\X_{n-1}) X_n^i
    \end{equation*}
    where $g_t \neq 0$. Then by \cref{lem:factor-of-normalized-poly}, $\hat g | \hat f$  where $\hat g$ is defined by
    \begin{equation*}
       \hat{g}(\X_n) = g(\X_{n-1}, X_n/f_k) \frac{f_k^t}{g_t}.
    \end{equation*}
    In other words,
    \begin{equation*}
        g(\X_n) = \hat{g}(\X_{n-1}, X_n f_k) \frac{g_t}{f_k^t}.
    \end{equation*}
    Note that $g_t | f_k$, which implies there exists a $g'_t$ in $\cL_{f_k, mon}$ such that $g_t \prec_{d} g'_t$. Hence, the polynomial $g'(\X_n)$ defined as
    \begin{equation*}
        g'(\X_{n}) = \hat{g}(\X_{n-1},X_n f_k)\cdot \frac{g'_t}{f_k^t}
    \end{equation*}
    lies in $\cL_{mon}$ since $\hat g \in \cL_{\hat f}$. Moreover, $g \prec_d g'$, which finishes the proof that $\cL_{mon}$ contains the claimed circuits for this case.

    Next, we prove the bound on the size of $\cL_{mon}$, the runtime of the algorithm, and the size and depth of the circuits.

    Note that
    \begin{align*}
        \abs{\cL_{mon}} &= \abs{\cL_{\hat f}} \cdot \abs{\cL_{f_k,mon}}  + \abs{\cL_{f_k,mon}} \\
            &\le (2^d-1) \cdot \abs{\cL_{f_k,mon}} + \abs{\cL_{f_k,mon}} & \text{(Using \cref{eqn:bounded-sparse-L-hat-f-size})}\\
            &\le 2^{d} \cdot \left(\frac{s}{2}\right)^{d } &\text{(Using \cref{eqn:bounded-sparse-L-fk-mon-size})}\\
            &=s^{d}
    \end{align*}
    which proves the list bound.

    The total time taken is dominated by the time taken to compute the lists $\cL_{\hat f}$ and $\cL_{f_k, mon}$ and to combine them to compute $\cL_{mon}$, which takes $\poly(d!, s^d, n)$ time.

    We conclude by showing the size and depth bounds for circuits in $\cL_{mon}$. The circuits in $\cL_{\hat f}$ are of size $\poly(s^d, n)$ and depth $\Delta$. Here $\Delta$ is an absolute constant (by \cref{lem:monic-factorization}). Without loss of generality, the top gate of these circuits is a multiplication gate.

    We will inductively show that the circuits in $\cL_{mon}$ have $\poly(s^d, n)$ size and $\Delta+6$ depth with a multiplication gate on top.
    The recursively inherited circuits in $\cL_{f_k,mon}$ satisfy these bounds by the induction hypothesis. It remains to consider the newly constructed circuits, which are of the form
    \begin{equation*}
        \hat g(\X_{n-1}, X_n f_k)\cdot \frac{h}{f_k^{\deg_{X_n}(\hat g)}}
    \end{equation*}
    where $\hat g \in \cL_{\hat f}$ and $h \in \cL_{f_k, mon}$.
    By \cref{lem:circuit-depth-on-compose}, $\hat g(\X_{n-1}, X_n f_k)$ is computable by a circuit $C_1$ of size $\poly(s^d, n)$ and depth $\Delta + 6$ (with a multiplication gate on top). $1/f_k^{\deg_{X_n}(\hat g)}$ is computable by a circuit $C_2$ (recall that we are allowing division gates) of size $\poly(s, n)$ and depth 5 (with a multiplication gate on top). By induction, $h$ is computable by a circuit $C_3$ of size $\poly(s^d, n)$ and depth $\Delta + 6$ (with a multiplication gate on top). We multiply all these circuits together by putting a multiplication gate on top. The resulting circuit will have size $\poly(s^d, n)$ and depth $\Delta + 6$ (we merge the multiplication gates of $C_1, C_2, C_3$ into one).

    We now address the case where sparsity of $f_0 < s/2$. Many of the ingredients are common to the previous case, but the normalization and combining procedures to obtain $\cL_{mon}$ are different.

    \paragraph{Sparsity of $f_0 < s/2$:}

    \cref{algo:monomial-equivalent-factors} reverse-normalizes $f$ to get $\tilde{f}(\X_n) = f(\X_{n-1},X_n f_0)/f_0$, which is reversed-monic in $X_n$ and has sparsity at most $s^d$.
    The algorithm invokes \cref{cor:reversed-monic-factorization} to compute a $(2^d -1)$-size list $\cL_{\tilde f}$ of circuits such that every factor of $\tilde f$ is computed by some circuit in $\cL_{\tilde f}$. Moreover, this computation takes $\poly(d!, s^d, n)$ time and each circuit in $\cL_{\tilde f}$ is over $\F$ and is of size $\poly(s^d, n)$ and depth $\cO(1)$.

    Next it recursively computes the list $\cL_{f_0, mon}$ of circuits. For every factor $g$ of $f_0$, there exists a polynomial $h$ computed by some circuit in $\cL_{f_0, mon}$ such that $g \prec_d h$. By our induction hypothesis, since $||f_0|| \leq s/2$, we have $\abs{\cL_{f_0, mon}} \leq (s/2)^d$ and we can compute this list in time $\poly(d!, (s/2)^d, n)$. Moreover, each circuit in $\cL_{f_0, mon}$ is over $\F$ and has size $\poly((s/2)^d, n)$ and depth $\cO(1)$.

    Finally, the algorithm uses the lists $\cL_{\tilde f}$ and $\cL_{f_0, mon}$ to construct the list $\cL_{mon}$. It defines $\cL_{mon}$ as follows.
    \begin{equation*}
        \cL_{mon} =  \left\{\tilde g(\X_{n-1}, X_n/f_0)\cdot h \,\middle\vert\, \tilde g \in \cL_{\tilde f}, h \in \cL_{f_0,mon}\right\} \cup \cL_{f_0,mon}
    \end{equation*}
    We will argue next that for any factor of $f$, there is a $d$-monomial equivalent polynomial in $\cL_{mon}$.

    Let $g$ be a factor of $f$. Suppose $g$ does not depend on the variable $X_n$. Then $g$ divides $f_0$. Hence, there exists a $g'$ in $\cL_{f_0, mon}\subset \cL_{mon}$ such that $g \prec_d g'$.

    Now suppose $g$ depends on $X_n$. Let $g$ be of the following form,
    \begin{equation*}
        g(\X_n) = \sum_{i=0}^t g_i(\X_{n-1}) X_n^i
    \end{equation*}
    where $g_0 \neq 0$ (since we had stripped out the maximal monomial divisor initially). Then by \cref{lem:factor-of-rev-normalized-poly}, $\tilde  g | \tilde f$  where $\tilde g$ is defined by
    \begin{equation*}
        \tilde{g}(\X_{n}) = \frac{g(\X_{n-1}, X_n f_0)}{g_0}.
    \end{equation*}
    In other words,
    \begin{equation*}
        g(\X_n) = \tilde{g}(\X_{n-1}, X_n/f_0) \cdot g_0.
    \end{equation*}
    Note that $g_0 | f_0$, which implies there exists a $g'_0$ in $\cL_{f_0, mon}$ such that $g_0 \prec_{d} g'_0$. Hence, the polynomial $g'(\X_n)$ defined as
    \begin{equation*}
        g'(\X_{n}) = \tilde{g}(\X_{n-1},X_n/f_0)\cdot g'_0
    \end{equation*}
    lies in $\cL_{mon}$ since $\tilde g \in \cL_{\tilde f}$. Moreover, $g \prec_d g'$, which finishes the proof of correctness for this case.

    The arguments for the bounds on the size of $\cL_{mon}$ and the runtime of the algorithm are the same as in the previous case. Thus, the size of $\cL_{mon}$ is at most $s^d$ and the total runtime is $\poly(d!, s^d, n)$.

    We conclude by showing the size and depth bounds for circuits in $\cL_{mon}$. The circuits in $\cL_{\tilde f}$ are of size $\poly(s^d, n)$ and depth $\Delta$. Here $\Delta$ is an absolute constant (see \cref{cor:reversed-monic-factorization}). Without loss of generality, the top gate of these circuits is a multiplication gate.

    We will inductively show that the circuits in $\cL_{mon}$ have $\poly(s^d, n)$ size and $\Delta+6$ depth with a multiplication gate on top.
    The recursively inherited circuits in $\cL_{f_0,mon}$ satisfy these bounds by the induction hypothesis. It remains to consider the newly constructed circuits, which are of the form
    \begin{equation*}
        \tilde g(\X_{n-1}, X_n/f_0)\cdot h
    \end{equation*}
    where $\tilde g \in \cL_{\tilde f}$ and $h \in \cL_{f_0, mon}$.
    By \cref{lem:circuit-depth-on-compose}, $\tilde g(\X_{n-1}, X_n /f_0)$ is computable by a circuit $C_1$ (in this case $C_1$ also has division gates) of size $\poly(s^d, n)$ and depth $\Delta + 6$ (with a multiplication gate on top).
    By induction, $h$ is computable by a circuit $C_2$ of size $\poly(s^d, n)$ and depth $\Delta + 6$ (with a multiplication gate on top). We multiply all these circuits together by putting a multiplication gate on top. The resulting circuit will have size $\poly(s^d, n)$ and depth $\Delta + 6$ (we merge the multiplication gates of $C_1, C_2$ into one).

    Note that in our construction of the circuits (in both cases) in $\cL_{mon}$, the denominators of the division gates are products of $s$-sparse polynomials. We record this fact as a claim below.

    \begin{claim}
        \label{clm:sparse-denominators}
        The denominators of the division gates in the circuits in $\cL_{mon}$ are products of $s$-sparse polynomials.
    \end{claim}
    Indeed, the recursive calls only introduce denominators with this form by induction. In the case $\|f_k\|\leq s/2$, the newly constructed circuits introduce only powers of $f_k$, and in the case $\|f_0\|\leq s/2$, the newly constructed circuits introduce only powers of $f_0$ through the substitution $X_n/f_0$. Since both $f_k$ and $f_0$ are coefficients of the $s$-sparse polynomial $f$, they are $s$-sparse.

\end{proof}

At this point, \cref{algo:monomial-equivalent-factors} outputs a list $\cL_{mon}$ containing monomial equivalent factors of $f$. \cref{algo:non-monomial-factors} will perform post-processing on this list to obtain a new list $\cL$ that contains all factors (not divisible by a monomial) of $f$. We show this formally by proving \cref{thm:compute-all-factors-of-bounded-degree} below.

\begin{proof}[Proof of \cref{thm:compute-all-factors-of-bounded-degree}]
  We will show that \cref{algo:non-monomial-factors} outputs the required list of circuits. The algorithm first removes and records the maximal monomial divisor of the input polynomial $f$; this gives the monomial output promised in the theorem. It then computes the list of circuits $\cL_{mon}$ using \cref{algo:monomial-equivalent-factors}. By \cref{clm:sparse-denominators}, the denominators of all division gates in these circuits are products of $s$-sparse polynomials. Thus, before removing monomial factors, \cref{algo:non-monomial-factors} applies \cref{lem:division-elimination} to every circuit in $\cL_{mon}$, replacing $\cL_{mon}$ by an equivalent list of division-free circuits. It then removes the largest monomial factor from each resulting circuit using \cref{cor:factor_sparse_monomial_division}, adds the result to $\cL$, and finally eliminates the monomial divisions introduced in this last step.

    Let $g$ be a factor of $f$ such that no monomial divides $g$. By \cref{thm:recursive-algorithm}, there exists $h \in \cL_{mon}$ such that $g \prec_d h$. By removing the largest monomial factor from $h$, we ensure that $g \in \cL$. Thus, every factor of $f$ that is not divisible by a monomial is computed by some circuit in $\cL$.

    Note that $\abs{\cL} \leq \abs{\cL_{mon}} = s^d$, and the runtime is dominated by \cref{algo:monomial-equivalent-factors}, which takes time $\poly(s^d,d!,n)$. The first division-elimination step gives division-free circuits of size $\poly(s^d,n)$ and depth $\cO(1)$. Removing the largest monomial factor then only introduces monomial denominators, so a second application of \cref{lem:division-elimination} removes these divisions with polynomial blow-up and only a constant increase in depth.
    Thus, the circuits in $\cL$ are of size $\poly(s^d,n)$ and depth $\cO(1).$

    Moreover, for circuits in $\cL_{mon}$ that do not correspond to true monomial equivalent factors of $f$, the above post-processing will add circuits in $\cL$ that do not correspond to true factors of $f$ (in $\poly(s^d,d!,n)$-time). These spurious circuits are also of size $\poly(s^d,n)$ and depth $\cO(1)$.

\end{proof}

As a consequence of \cref{thm:main1}, we recover the following corollary. This recovers a result of Chuyoon and Shpilka~\cite{chuyoon2026factorization}, as well as the algorithmic result of Bhargava et al.~\cite{BSV20} and its improvement in~\cite{chuyoon2026factorization}.
\begin{corollary}
\label{cor:cs-and-bsv-generalization}
Let $f(X_1,\dots, X_n)\in \F[X_1, \dots, X_n]$ be an $s$-sparse polynomial with individual degree at most $d$, where $\Char(\F)$ is $0$ or greater than $d$. Then there is an algorithm that runs in time $\poly(s^d, d!, n, s')$ and outputs all $s'$-sparse factors of $f$ not divisible by a monomial.
\end{corollary}
\begin{proof}
    Applying \cref{thm:compute-all-factors-of-bounded-degree} to $f$, we obtain a list of circuits $\cL$ such that every factor of $f$ that is not divisible by a monomial is computed by some circuit in $\cL$. In particular, $\cL$ could contain spurious circuits that do not correspond to factors of $f$. We apply $s'$-sparse polynomial interpolation (\cref{Thm: Sparse Reconstruction}) to each circuit $C$ in $\cL$. For each circuit where interpolation succeeds, let $g_C$ be the corresponding polynomial. We use \cref{thm:sparse-divisibility-test} with sparsity parameter $s'$ to check whether $g_C$ divides $f$. We prune out the polynomials $g_C$ that do not divide $f$ and get a final list of all $s'$-sparse factors of $f$ that are not divisible by any monomial.
\end{proof}
\begin{remark}
\label{rem:cs-and-bsv-generalization}
    By setting $s' = s$ in \cref{cor:cs-and-bsv-generalization}, we obtain a $\poly(s^d, d!, n)$-time algorithm to output all $s$-sparse factors of $f$ not divisible by a monomial. Thus, we recover \cite[Theorem 1.9]{chuyoon2026factorization} for fields of characteristic zero or greater than $d$.

    By \cref{thm:divisor_bound_general_sparse}, all factors of $f$ have sparsity at most $s^{4 d^2 \log n}$. Thus, by setting $s'=s^{4 d^2 \log n}$ in \cref{cor:cs-and-bsv-generalization}, we obtain a $\poly(s^{d^2 \log n}, d!, n)$-time algorithm to output all factors of $f$ not divisible by a monomial. This generalizes \cite[Theorem 2]{BSV20} and matches the improved runtime obtained in \cite{chuyoon2026factorization}.
\end{remark}

\newpage

\begin{algorithm}[H]
    \caption{List of monomial equivalent factors of sparse polynomials with bounded individual degree (possibly with spurious entries)}
    \label{algo:monomial-equivalent-factors}
    \KwData{$f \in \F[\X_n]$, sparsity parameter $s$, individual degree parameter $d$}
    \KwResult{List $\cL_{mon}$ of circuits over $\F$ with division gates}

    \quad\\
    \texttt{// Pre-processing}\\
    Remove the maximal monomial divisor from $f$ and rename the polynomial $f$.\\
    \quad\\
    \texttt{// Base Case} \\
    \If{$f$ is a constant}
    {
    $\cL_{mon} = \{f \cdot (X_1 \cdot X_2 \cdots X_n)^d\}$ \\
    \Return{$\cL_{mon}$}
    }
    \quad\\
    Without loss of generality assume $f$ depends on variable $X_n$, otherwise rename the variables appropriately.\\
    Let $k = \deg_{X_n}(f)$, and find $f_0,...,f_k \in \F[\X_{n-1}]$ such that $f = \displaystyle\sum_{i=0}^k f_i(\X_{n-1})X_n^i$\\

    \texttt{// Normalize $f$}\\
    \If{$||f_k|| \leq s/2$}{
    Compute $\hat{f} = f(\X_{n-1},X_n/f_k)\cdot f_k^{k-1}$ \\
    Compute $\cL_{\hat f} = $ \{Factors of $\hat{f}\}$ using monic factorization (\cref{lem:monic-factorization})\\
    $\cL_{f_k,mon} = \cref{algo:monomial-equivalent-factors}(f_k,s/2,d)$ \qquad \texttt{// Monomial equivalent factors of $f_k$}\\
    $\cL_{mon} =  \left\{\hat g(\X_{n-1}, X_n f_k)\cdot \frac{h}{f_k^{\deg_{X_n}(\hat g)}} \,\mid\, \hat g \in \cL_{\hat f}, h \in \cL_{f_k,mon}\right\} \cup \cL_{f_k,mon}$\\
    \texttt{// We know $\deg_{X_n}(\hat g)$ by \cref{lem:monic-factorization}}
    }
    \qquad \\
    \texttt{// Reverse Normalize $f$}\\
    \Else{
    Compute $\tilde f = f(\X_{n-1},X_n f_0)/f_0$ \\
    Compute $\cL_{\tilde f} = $ \{Factors of $\tilde{f}\}$ using reversed-monic factorization (\cref{cor:reversed-monic-factorization}) \\
    $\cL_{f_0,mon} = \cref{algo:monomial-equivalent-factors}(f_0,s/2,d)$ \qquad \texttt{// Monomial equivalent factors of $f_0$}\\
    $\cL_{mon} = \left\{\tilde g(\X_{n-1},X_n/f_0)\cdot h \,\mid\, \tilde g \in \cL_{\tilde f}, h \in \cL_{f_0,mon} \right\} \cup \cL_{f_0,mon}$\\

    }

    \qquad \\

    \Return{$\cL_{mon}$}
\end{algorithm}

\begin{algorithm}[H]
    \caption{List of factors not divisible by a monomial of bounded individual degree sparse polynomials (possibly with spurious entries)}
    \label{algo:non-monomial-factors}
    \KwData{$f \in \F[\X_n]$, sparsity parameter $s$, individual degree parameter $d$}
    \KwResult{List $\cL$ of circuits over $\F$ and the maximal monomial divisor $M$ of $f$}

    \quad\\

    Let $M$ be the maximal monomial divisor of $f$, and replace $f$ by $f/M$.\\
    $\cL_{mon}$ =  \cref{algo:monomial-equivalent-factors}$(f, s, d)$ \\
    $\cL = \emptyset$

    \quad\\

    \texttt{// Eliminating division gates}\\
    Use \cref{lem:division-elimination} to eliminate division gates in circuits in $\cL_{mon}$\\

    \quad\\

    \texttt{// Removing some spurious elements from the list}\\
    \For{$g \in \cL_{mon}$}{
        \For{$i \in [n]$}{
        \If{$X_i^{d+1}|g$ $($Using \cref{cor:factor_sparse_monomial_division}$)$ }{
            Remove $g$ from $\cL_{mon}$\\
            \textbf{Break}
        }
        }
    }

    \quad\\
    \texttt{// Removing monomial divisors from true factors}\\

    \For{$g \in \cL_{mon}$}{
        \For{$i \in [n]$}{
        Find largest power $d_i$ such that $X_i^{d_i}| g$ (Using \cref{cor:factor_sparse_monomial_division})\\
        }
        $\cL$ = $\cL \cup \left\{\frac{g}{X_1^{d_1}\cdot\cdot X_n^{d_n}}\right\}$\\

    }

    \quad\\
    Use \cref{lem:division-elimination} to eliminate monomial division gates in circuits in $\cL$ \\

    \Return{$(\cL,M)$}
\end{algorithm}

\section{Recovering Bounded Individual Degree Factors of General Sparse Polynomials}
\label{sec:recovering-bounded-individual-degree-factors}

In this section, we prove \cref{thm:main2}. In particular, we obtain a deterministic quasipolynomial-time algorithm that recovers all factors $g$ that have bounded individual degree $d$ (and are not divisible by a monomial) of an $n$-variate $s$-sparse polynomial $f \in \F[\X_n]$ of (unbounded) individual degree $D$. The algorithm outputs a $\binom{D}{\leq d}^{\log s}$-size list $\cL$ such that each such factor $g$ is in $\cL$. $\cL$ might contain spurious polynomials that do not correspond to any factor of $f$. Unfortunately, we are unable to prune out the spurious elements since we do not know of efficient deterministic algorithms to check if a bounded individual degree polynomial $g$ divides $f$ (see \cref{sec:future-directions-and-open-questions}). Indeed, if we could do this divisibility testing efficiently, then we would be able to clean the list and output only true factors. Note that the restriction to outputting factors not divisible by a monomial is unavoidable (see \cref{rem:factor-not-divisible-by-monomial}). Since we are given white-box access to the polynomial $f$, we compensate for this restriction by outputting the maximal degree monomial factor dividing $f$.

We restate our \cref{thm:main2} (slightly more formally) below.

\begin{theorem}[Main Theorem 2]
    \label{thm:general-sparse-bounded_ind_deg_factor}
    Let $f(\X_n) \in \F[\X_n]$ be an $s$-sparse polynomial with individual degree at most $D$. Then, there is an algorithm that runs in time $\poly\left(s^{d^2\log n}, \binom{D}{\leq d}^{\log s}, n\right)$ and outputs a $\binom{D}{\leq d}^{\log s}$-size list of $s^{4d^2\log n}$-sparse polynomials $\cL$. For every factor $g$ of $f$, where the individual degree of $g$ is at most $d$ and $g$ is not divisible by any monomial, $g$ lies in $\cL$.
\end{theorem}
\begin{remark}
    \label{rem:dependence-on-field-all-factors}
    In the above theorem, the displayed running time suppresses the dependence on the time taken to perform univariate polynomial factorization.
    More precisely, over a field $\F$, the algorithm runs in time $\poly\left(s^{d^2\log n}, \binom{D}{\leq d}^{\log s}, n\right) \cdot \cT(\F, D)$ (see \cref{def:time-taken-to-factor-univariates}). In particular, over $\Q$, if the input polynomial has bit complexity $B$, then the algorithm runs in time $\poly\left(s^{d^2\log n}, \binom{D}{\leq d}^{\log s}, n, B\right)$.
    If $\F$ is a finite field with characteristic $p$ (greater than $d$), the algorithm runs in time $\poly\left(s^{d^2\log n}, \binom{D}{\leq d}^{\log s}, n, p, \log{\abs{\F}}\right)$.
\end{remark}

At a high level, this algorithm also runs in two stages as in \cref{sec:recovering-all-factors}: the main recursive \cref{algo:monomial-equivalent-factors_general sparse} and the post-processing \cref{algo:non-monomial-factors_general_sparse}. The main recursive \cref{algo:monomial-equivalent-factors_general sparse} uses a similar divide-and-conquer strategy to that in \cref{algo:monomial-equivalent-factors}. \cref{algo:monomial-equivalent-factors_general sparse} outputs a set $\cL_{mon}$ of \emph{monomial equivalent} factors of $f$. More precisely, for every factor $g$ (of individual degree at most $d$) of $f$, there exists a polynomial $h$ in $\cL_{mon}$ such that $g \prec_d h$ (see \cref{def:monomial-equivalent-factor}). Thus, for such a factor $g$ with no monomial divisor, $h=gM$ will be in $\cL_{mon}$, where $M$ is a monomial with individual degree at most $d$. \cref{algo:non-monomial-factors_general_sparse} strips out this monomial $M$ from circuit $C$, thus obtaining the promised list $\cL$ containing the factors\footnote{On members of the list that do not correspond to true factors, we do not guarantee what might come out of this process, but it does not affect our run time or the guarantee on true factors.}.

We formally argue the correctness and the runtime of \cref{algo:monomial-equivalent-factors_general sparse} below. We refer the reader to \cref{sec:proof-overview-bounded-individual-degree-factors} for a high-level overview of the key ideas in the algorithm.

\begin{theorem}[Recursive algorithm for factoring general sparse polynomials]
    \label{thm:general-sparse-recursive-algorithm}
    Let $f(\X_n) \in \F[\X_n]$ be an $s$-sparse polynomial with individual degree at most $D$. Then, there is an algorithm that runs in time $\poly\left(s^{d^2\log n}, \binom{D}{\leq d}^{\log s}, n\right)$ and outputs a $\binom{D}{\leq d}^{\log s}$-size list of $s^{4d^2\log n}$-sparse polynomials $\cL_{mon}$. For every factor $g$ of $f$, where the individual degree of $g$ is at most $d$, there exists a polynomial $h$ in $\cL_{mon}$ such that $g \prec_d h$. %
\end{theorem}

\begin{proof}
    We will show that \cref{algo:monomial-equivalent-factors_general sparse} outputs the required list of circuits.
    We prove the theorem by induction on the sparsity of $f$. \cref{algo:monomial-equivalent-factors_general sparse} initially strips off the maximal monomial divisor of $f$. Thus, from this point on, we work with an $f$ that has no nonconstant monomial divisor.

    \paragraph{Base case ($||f|| \leq 1$):}
    Then $f$ is a constant, since we have stripped off the maximal monomial divisor.
    \cref{algo:monomial-equivalent-factors_general sparse} outputs $\{f \cdot (X_1 \cdot X_2 \cdots X_n)^d\}$, and thus gives a $d$-monomial equivalent polynomial to $f$.

    \paragraph{Induction case ($||f|| \leq s$):}

    Without loss of generality, $f$ depends on the variable $X_n$ (\cref{algo:monomial-equivalent-factors_general sparse} ensures this by renaming variables). The algorithm instantiates a $(n-1)$-variate generator $\cG \coloneq \cG_{(s',2d,n-1)}^{KS}$ from \cref{Thm: Sparse Reconstruction}, where $s' = s^{4 d^2 \log{n}}$ (as we will see later, the value of $s'$ comes from \cref{thm:bsv_general_sparse}). By \cref{Thm: Sparse Reconstruction}, it takes $\poly(s^{4 d^2 \log{n}}, n, d)$ time to construct the generator $\cG$. Let $k = \deg_{X_n}(f)$, and
    \begin{equation*}
        f = \sum_{i=0}^k f_i(\X_{n-1})X_n^i
    \end{equation*}
    where $f_i(\X_{n-1}) \in \F[\X_{n-1}]$. Since we stripped off the maximal monomial divisor, $f_0(\X_{n-1})$ and $f_k(\X_{n-1})$ are both non-zero. By averaging, either $||f_0(\X_{n-1})||$ or $||f_k(\X_{n-1})||$ is at most $s/2$. We first deal with the case where $||f_k(\X_{n-1})||$ is at most $s/2$.

    \paragraph{Sparsity of $f_k \leq s/2$:}
    \cref{algo:monomial-equivalent-factors_general sparse} normalizes $f$ to get $\hat{f}(\X_n) = f(\X_{n-1},X_n/f_k)\cdot f_k^{k-1}$, which is monic in $X_n$. Note that since $\hat{f}$ is $\cO(s^D)$-sparse, we do not have the budget to explicitly store its monomial representation. Hence, we represent it by a circuit $C$. Using $\hat{f}$, the algorithm computes the five-variate polynomial
    \begin{equation*}
        \hat f_{\cG}(\bZ_4,X_n) = \hat f (\cG(\bZ_4),X_n)
    \end{equation*}
    in its monomial representation. We can compute the monomial representation (using interpolation) in time $\poly(n,s',d)$ since $\deg(\hat f_{\cG}(\bZ_4,X_n)) = \poly(n,s',d)$.
    The algorithm invokes \cref{lem: Constantvariate factorization} to compute the list $\cL_{\hat f_{\cG}}$ such that every nonconstant factor of $\hat f_\cG$ %
    with $X_n$-degree at most $d$ is in $\cL_{\hat f_\cG}$. Since $\hat f_{\cG}$ is monic in $X_n$ and degree at most $D$, it follows that
    \begin{equation}
        \abs{\cL_{\hat f_\cG}} \leq \binom{D}{\leq d} - 1 \label{eqn:general-sparse-size-of-f_G}
    \end{equation}
    Moreover, this computation takes $\poly(s', \binom{D}{\leq d}, n)$ time and each polynomial in $\cL_{\hat f_\cG}$ is $\poly(s', n)$-sparse.

    Next it recursively computes the list $\cL_{f_k, mon}$. For every factor $g$ of $f_k$, where the individual degree of $g$ is at most $d$, there exists a polynomial $h$ in $\cL_{f_k, mon}$ such that $g \prec_d h$. By our induction hypothesis, since $||f_k|| \leq s/2$, we have $\abs{\cL_{f_k, mon}} \leq \binom{D}{\leq d}^{\log s - 1}$ and we can compute this list in time $\poly((s/2)^{d^2 \log n}, \binom{D}{\leq d}^{\log s - 1}, n)$.
    \begin{equation}
        \abs{\cL_{f_k, mon}} \leq \binom{D}{\leq d}^{\log s - 1} \label{eqn:general-sparse-size-of-f_k-mon}
    \end{equation}
    Moreover, each polynomial in $\cL_{f_k, mon}$ is $(s/2)^{4d^2\log n}$-sparse.

    Next, the algorithm uses the lists $\cL_{\hat f_\cG}$ and $\cL_{f_k, mon}$ to construct the list $\cL_{f_\cG, mon}$. It defines $\cL_{f_\cG, mon}$ as follows.
    \begin{equation*}
        \cL_{f_\cG, mon} =  \left\{\hat{g}(\bZ_4, X_n \cdot \left(f_k \circ \cG(\bZ_4)\right))\cdot \frac{h \circ \cG(\bZ_4)}{\left(f_k \circ \cG(\bZ_4)\right)^{\deg_{X_n}(\hat{g})}} \,\middle\vert\, \hat{g} \in \cL_{\hat f_\cG}, h \in \cL_{f_k,mon}\right\}.
    \end{equation*}

    The algorithm then uses $\cL_{f_{\cG},mon}$ and $\cL_{f_k,mon}$ to construct the list $\cL_{mon}$. For any $g'
    (\bZ_4,X_n)\in \cL_{f_{\cG},mon}$, the algorithm applies the reconstruction algorithm from \cref{Thm: Sparse Reconstruction} to find a polynomial $g(\X_n)$ of sparsity\footnote{If the reconstruction algorithm outputs a polynomial of sparsity greater than $s'$, we simply remove it from $\cL_{mon}$. Since the reconstruction algorithm runs in time $\poly(s',D,n)$, the sparsity of the polynomial output by it is at most $\poly(s',D,n)$, which makes it feasible to prune the list this way.} at most $s'$, such that $g(\cG(\bZ_4),X_n) = g'(\bZ_4,X_n)$\footnote{If the reconstruction algorithm outputs a polynomial $g$ such that $g(\cG(\bZ_4),X_n) \neq g'(\bZ_4,X_n)$, then we remove it from $\cL_{mon}$. Checking whether $g(\cG(\bZ_4),X_n) \neq g'(\bZ_4,X_n)$ takes time at most $\poly(s',D,n)$, so this pruning is feasible.}. We further include $\cL_{f_k,mon}$, and thus $\cL_{mon}$ is defined as
    \begin{equation*}
        \cL_{mon} = \left\{ \text{Reconstruct } g(\X_{n}) \text{ from } g'(\bZ_4, X_n) \,\middle\vert\, g' \in \cL_{f_\cG, mon} \right\} \cup \cL_{f_k,mon}.
    \end{equation*}
    Constructing $\cL_{f_\cG,mon}$ takes time polynomial in $\abs{\cL_{\hat f_\cG}}\cdot \abs{\cL_{f_k,mon}}$, $s'$, $D$, and $n$. Since each reconstruction and pruning step takes $\poly(s',D,n)$ time, the list $\cL_{mon}$ is constructed in time $\poly\left(s',\binom{D}{\leq d}^{\log s}, n\right)$.

    We will argue next that for any factor $g$ of $f$ where the individual degree of $g$ is at most $d$, there is a $d$-monomial equivalent polynomial $h$ in $\cL_{mon}$.
    Let $g$ be such a factor. Suppose $g$ does not depend on the variable $X_n$. Then $g$ divides $f_k$. Hence, there exists a $g'$ in $\cL_{f_k, mon}\subset \cL_{mon}$ such that $g \prec_d g'$.

    Now suppose $g$ depends on $X_n$. Let $g$ be of the following form,
    \begin{equation*}
        g(\X_n) = \sum_{i=0}^t g_i(\X_{n-1}) X_n^i
    \end{equation*}
    where $g_t \neq 0$. Then by \cref{lem:factor-of-normalized-poly}, $\hat g | \hat f$  where $\hat g$ is defined by
    \begin{equation*}
        \hat{g}(\X_n) = g(\X_{n-1}, X_n/f_k) \frac{f_k^t}{g_t}.
    \end{equation*}
    In other words,
    \begin{equation*}
        g(\X_n) = \hat{g}(\X_{n-1}, X_n f_k) \frac{g_t}{f_k^t}.
    \end{equation*}
    Note that $g_t | f_k$, and $g_t$ has individual degree at most $d$, which implies that there exists $g'_t$ in $\cL_{f_k, mon}$ such that $g_t \prec_{d} g'_t$. We will show that the polynomial $g'(\X_n)$ defined by
    \begin{equation*}
        g'(\X_n) = \hat{g}(\X_{n-1},X_n f_k)\cdot \frac{g'_t}{f_k^t}
    \end{equation*}
    lies in $\cL_{mon}$. Since $\hat g(\X_n) \mid \hat f(\X_n)$, we have $\hat g(\cG(\bZ_4),X_n) \mid \hat f(\cG(\bZ_4),X_n)$. Furthermore, $\deg_{X_n}(\hat g(\cG(\bZ_4),X_n))$ is at most $d$, and hence $\hat g(\cG(\bZ_4),X_n) \in \cL_{\hat f_\cG}$. By definition of $\cL_{f_{\cG},mon}$, since $\hat g(\cG(\bZ_4),X_n) \in \cL_{\hat f_\cG}$ and $g'_t \in \cL_{f_k,mon}$, we have $g'(\cG(\bZ_4),X_n) \in \cL_{f_{\cG},mon}$. Let $g'(\X_n)$ be of the form
    \begin{equation*}
        g'(\X_n) =\sum_{i=0}^{d} g'_i(\X_{n-1}) X_n^i
    \end{equation*}
    then
    \begin{equation*}
        g'(\cG(\bZ_4),X_n) = \sum_{i=0}^{d} g'_i(\cG(\bZ_4)) X_n^i
    \end{equation*}

    Since $g$ has individual degree at most $d$ and $g \mid f$, by \cref{thm:bsv_general_sparse}, we know that $g$ has sparsity at most $s^{4d^2 \log n}$. Note that the sparsity of $g'$ is equal to the sparsity of $g$, which is $s^{4d^2 \log n}$, and its individual degree is at most $2d$. Thus, for all $i$, $g'_i(\X_{n-1})$ is $s^{4d^2 \log n}$-sparse and has individual degree at most $2d$. Since we have access to the monomial representation of $g'(\cG(\bZ_4),X_n)$, we also have access to the monomial representation of $g'_i(\cG(\bZ_4))$. Thus, by using \cref{Thm: Sparse Reconstruction}, we reconstruct $g'_i(\X_{n-1})$ from $g'_i(\cG(\bZ_4))$ for $0 \leq i \leq d$, thereby reconstructing $g'$.
    This finishes the proof that $\cL_{mon}$ contains the claimed polynomials for this case.

    Next, we prove the bound on the size of $\cL_{mon}$, the runtime of the algorithm, and the sparsity of the polynomials in it.

    Observe that
    \begin{align*}
         \abs{\cL_{mon}} &\leq \abs{\cL_{f_{\cG},mon}} + \abs{\cL_{f_k,mon}} &\text{(Since } \cL_{mon} = \cL_{f_{\cG},mon} \cup \cL_{f_k,mon}\text{)} \\
        &\leq \abs{\cL_{\hat f_{\cG}}} \cdot \abs{\cL_{f_k,mon}}  + \abs{\cL_{f_k,mon}} &\text{(Since } \abs{\cL_{f_{\cG},mon}} \leq \abs{\cL_{\hat f_{\cG}}} \cdot \abs{\cL_{f_k,mon}}\text{)} \\
            &\le \left( \binom{D}{\leq d}-1 \right) \cdot \abs{\cL_{f_k,mon}} + \abs{\cL_{f_k,mon}} & \text{(Using \cref{eqn:general-sparse-size-of-f_G})}\\
            &\le \binom{D}{\leq d}^{\log s} &\text{(Using \cref{eqn:general-sparse-size-of-f_k-mon})}
    \end{align*}
    which proves the required list size bound.

    The total time taken is dominated by the time taken to compute the lists $\cL_{\hat f_{\cG}}$, $\cL_{f_k, mon}$, $\cL_{f_{\cG},mon}$, and $\cL_{mon}$, which takes $\poly\left( \binom{D}{\leq d}^{\log s}, s^{d^2 \log n},n \right)$ time.

    Observe that every element in $\cL_{mon}$ is obtained from an application of $s^{4d^2\log n}$-sparse polynomial reconstruction (by \cref{Thm: Sparse Reconstruction}). Thus, every polynomial in $\cL_{mon}$, including spurious elements not corresponding to factors, is $s^{4d^2\log n}$-sparse.

    We now address the other case where the sparsity of $f_0$ is less than $s/2$. Many of the ingredients are common to the previous case.

    \paragraph{Sparsity of $f_0 < s/2$:}

       \cref{algo:monomial-equivalent-factors_general sparse} reverse-normalizes $f$ to get $\tilde{f}(\X_n) = f(\X_{n-1},X_n\cdot f_0)/f_0$, which is reversed-monic in $X_n$. Note that since $\tilde{f}$ is $\cO(s^D)$-sparse, we do not have the budget to explicitly store its monomial representation. Hence, we represent it by a circuit $C$. Using $\tilde{f}$, the algorithm computes the five-variate polynomial
    \begin{equation*}
        \tilde f_{\cG}(\bZ_4,X_n) = \tilde f (\cG(\bZ_4),X_n)
    \end{equation*}

    The algorithm invokes \cref{lem: Constantvariate factorization} to compute a $\binom{D}{\leq d} - 1$-size list $\cL_{\tilde f_{\cG}}$ such that every nonconstant factor of $\tilde f_\cG$ with $X_n$-degree at most $d$ is in $\cL_{\tilde f_\cG}$.
    Moreover, this computation takes $\poly(s^{d^2 \log n}, \binom{D}{\leq d}^{\log s}, n)$ time and each polynomial in $\cL_{\tilde f_\cG}$ is $\poly(s^{d^2\log n}, n)$-sparse.

    Next it recursively computes the list $\cL_{f_0, mon}$. For every factor $g$ of $f_0$, where the individual degree of $g$ is at most $d$, there exists a polynomial $h$ in $\cL_{f_0, mon}$ such that $g \prec_d h$. By our induction hypothesis, since $||f_0|| \leq s/2$, we have $\abs{\cL_{f_0, mon}} \leq \binom{D}{\leq d}^{\log s - 1}$ and we can compute this list in time $\poly((s/2)^{d^2 \log n}, \binom{D}{\leq d}^{\log s - 1}, n)$.
    Moreover, each polynomial in $\cL_{f_0, mon}$ is $(s/2)^{4d^2\log n}$-sparse.

    Next, the algorithm uses the lists $\cL_{\tilde f_\cG}$ and $\cL_{f_0, mon}$ to construct the list $\cL_{f_\cG, mon}$. It defines $\cL_{f_\cG, mon}$ as follows.
    \begin{equation*}
        \cL_{f_\cG, mon} = \left\{\tilde g(\bZ_4,X_n/\left(f_0\circ \cG(\bZ_4)\right))\cdot h \circ \cG(\bZ_4) \,\middle\vert\, \tilde g \in \cL_{\tilde f_\cG}, h \in \cL_{f_0,mon} \right\}.
    \end{equation*}

    The algorithm then uses $\cL_{f_{\cG},mon}$ and $\cL_{f_0,mon}$ to construct the list $\cL_{mon}$. The list $\cL_{mon}$ is defined as
    \begin{equation*}
        \cL_{mon} = \left\{ \text{Reconstruct } g(\X_{n-1}, X_n) \text{ from } g'(\bZ_4, X_n) \,\middle\vert\, g' \in \cL_{f_\cG, mon} \right\} \cup \cL_{f_0,mon}.
    \end{equation*}
    where the reconstruction is similar to the previous case.

    We will argue next that for any factor $g$ of $f$ where the individual degree of $g$ is at most $d$, there is a $d$-monomial equivalent polynomial $h$ in $\cL_{mon}$.
    Let $g$ be such a factor. Suppose $g$ does not depend on the variable $X_n$. Then $g$ divides $f_0$. Hence, there exists a $g'$ in $\cL_{f_0, mon}\subset \cL_{mon}$ such that $g \prec_d g'$. %

    Now suppose $g$ depends on $X_n$. Let $g$ be of the following form,
    \begin{equation*}
        g(\X_n) = \sum_{i=0}^t g_i(\X_{n-1}) X_n^i
    \end{equation*}
    where $g_0 \neq 0$ (since we had stripped out the maximal monomial divisor initially). Then by \cref{lem:factor-of-rev-normalized-poly}, $\tilde  g | \tilde f$  where $\tilde g$ is defined by
    \begin{equation*}
        \tilde{g}(\X_n) = \frac{g(\X_{n-1}, X_n f_0)}{g_0}.
    \end{equation*}
    In other words,
    \begin{equation*}
        g(\X_n) = \tilde{g}(\X_{n-1}, X_n/f_0) \cdot g_0.
    \end{equation*}
    Note that $g_0 | f_0$, which implies that there exists $g'_0$ in $\cL_{f_0, mon}$ such that $g_0 \prec_{d} g'_0$. We will now show that the polynomial $g'(\X_n)$ defined by
    \begin{equation*}
        g'(\X_n) = \tilde{g}(\X_{n-1},X_n /f_0)\cdot g'_0
    \end{equation*}
    lies in $\cL_{mon}$. Since $\tilde g(\X_n) \mid \tilde f(\X_n)$, we have $\tilde g(\cG(\bZ_4),X_n) \mid \tilde f(\cG(\bZ_4),X_n)$. Furthermore, $\deg_{X_n}(\tilde g(\cG(\bZ_4),X_n))$ is at most $d$, and hence $\tilde g(\cG(\bZ_4),X_n) \in \cL_{\tilde f_\cG}$. By definition of $\cL_{f_{\cG},mon}$, since $\tilde g(\cG(\bZ_4),X_n) \in \cL_{\tilde f_\cG}$ and $g'_0 \in \cL_{f_0,mon}$, we have $g'(\cG(\bZ_4),X_n) \in \cL_{f_{\cG},mon}$. By \cref{thm:bsv_general_sparse}, we know that $g$ has sparsity at most $s^{4d^2 \log n}$, and the sparsity of $g'$ is equal to the sparsity of $g$. Since $g'$ is $s^{4d^2 \log n}$-sparse and has individual degree at most $2d$, by \cref{Thm: Sparse Reconstruction}, $g' \in \cL_{mon}.$ This finishes the proof that $\cL_{mon}$ contains the claimed polynomials for this case.

    The arguments for the bounds on the size of $\cL_{mon}$, the runtime of the algorithm, and the sparsity of the polynomials in $\cL_{mon}$ are the same as in the previous case. Thus, the size of $\cL_{mon}$ is at most $\binom{D}{\leq d}^{\log s}$, the total runtime is $\poly\left( \binom{D}{\leq d}^{\log s}, s^{d^2 \log n},n \right)$, and every polynomial in $\cL_{mon}$ is $s^{4d^2\log n}$-sparse.

\end{proof}

At this point, \cref{algo:monomial-equivalent-factors_general sparse} outputs a list $\cL_{mon}$ containing monomial equivalent factors of $f$. \cref{algo:non-monomial-factors_general_sparse} performs post-processing on this list to obtain a new list $\cL$ that contains all factors (not divisible by a monomial) of $f$. We show this formally by proving \cref{thm:general-sparse-bounded_ind_deg_factor} below.

\begin{proof}[Proof of \cref{thm:general-sparse-bounded_ind_deg_factor}]
    We will show that \cref{algo:non-monomial-factors_general_sparse} outputs the required list of polynomials.
    \cref{algo:non-monomial-factors_general_sparse} first computes the list $\cL_{mon}$ of $s^{4d^2\log n}$-sparse polynomials using \cref{algo:monomial-equivalent-factors_general sparse}. For any $h \in \cL_{mon}$, \cref{algo:non-monomial-factors_general_sparse} then removes the largest monomial factor of $h$ and adds the result to $\cL$. Note that after removing the largest monomial factor of $h$, the sparsity does not change. We prove the correctness and runtime of \cref{algo:non-monomial-factors_general_sparse} below.

    Let $g \mid f$, where $g$ has individual degree at most $d$ and no monomial divides $g$. We will show that $g \in \cL$. By \cref{thm:general-sparse-recursive-algorithm}, there exists $h \in \cL_{mon}$ such that $g \prec_d h$. By removing the largest monomial factor of $h$, we ensure $g \in \cL$.
    Note that $\abs{\cL} \leq \abs{\cL_{mon}} \leq \binom{D}{\leq d}^{\log s}$, which proves the required bound on list size. Furthermore, the time complexity of the algorithm is dominated by the computation of $\cL_{mon}$, which takes time $\poly(s^{d^2 \log n}, \binom{D}{\leq d}^{\log s},n)$ by \cref{thm:general-sparse-recursive-algorithm}.
\end{proof}

Along with an algorithm for factoring general sparse polynomials, \cref{thm:general-sparse-bounded_ind_deg_factor} also gives a bound on the number of factors of individual degree at most $d$ (not divisible by a monomial) of general sparse polynomials. In fact, the example from \cite[Remark 8.2]{chuyoon2026factorization} shows that this bound is tight.

\begin{theorem}[Divisor bound for general sparse polynomials]
    \label{thm:divisor_bound_general_sparse}
Let $f$ be any $n$-variate $s$-sparse polynomial of individual degree $D$. Then $f$ can have at most $\binom{D}{\leq d}^{\log s}$ factors of individual degree at most $d$ that are not divisible by any monomial. Moreover, this bound is tight.
\end{theorem}

\begin{proof}
    \Cref{thm:general-sparse-bounded_ind_deg_factor} gives a $\binom{D}{\leq d}^{\log s}$-size list $\cL$ that contains all factors of individual degree at most $d$ (not divisible by a monomial) of $f$. This shows the required divisor bound.

    To prove that the bound is tight, consider the polynomial $f(\X) = \prod_{i=1}^{\log s}(X_i^D -1)$ over $\C$. Observe that
    \begin{equation*}
        f(\X) = \prod_{i=1}^{\log s}(X_i^D - 1) = \prod_{i=1}^{\log s} \prod_{j=0}^{D-1} (X_i - \omega^j),
    \end{equation*}
    where $\omega \in \C$ is the $D$\emph{th} primitive root of unity.
    Thus, the number of factors of $f$ of individual degree at most $d$ is exactly $\binom{D}{\leq d}^{\log s}$.
\end{proof}

As an immediate consequence of our algorithm, we can find all constant-degree factors of an $s$-sparse polynomial in quasipolynomial time by pruning the list using the divisibility testing algorithm of \cite{Forbes15}. In doing so, we recover some results from \cite{KRS24,DST24}. We state the result formally in the following corollary.

\begin{corollary}\label{cor:Constant-Deg-Factors-of-Sparse}
    Let $f \in \F[\X_n]$ be an $s$-sparse polynomial of individual degree at most $D$. Let $\Char(\F) \geq \poly(D)$. We can find all factors $g$ of $f$ such that $\deg(g) = \cO(1)$ in time $\poly(s,n,D)^{\cO(\log s)}$.
\end{corollary}

\begin{proof}
    Using \cref{thm:general-sparse-bounded_ind_deg_factor}, we can find all factors $g$ of $f$ such that the individual degree of $g$ is $\cO(1)$ and no monomial divides it. We can prune the list of factors to only contain polynomials of degree $\cO(1)$. Every constant-degree factor of $f$ can be written as a monomial times such a polynomial, and there are only $\poly(n,D)$ monomials of degree $\cO(1)$ to try. Then, by trying all these monomial multiples (which keeps degree $\cO(1)$) and using \cref{thm:divisibility-testing-constant-deg}, we can find the list of all factors of $f$ of constant degree in time $\poly(s,n,D)^{\cO(\log s)}$.
\end{proof}

\begin{algorithm}[H]
    \caption{List of monomial equivalent bounded individual degree factors of general sparse polynomials (possibly with spurious entries)}
    \label{algo:monomial-equivalent-factors_general sparse}
    \KwData{$f \in \F[\X_n]$, sparsity parameter $s$, individual degree parameter $d$}
    \KwResult{$\cL_{mon} \subset \F[\X_n]$, list of circuits over $\F$ with division gates}

    \quad\\
    \texttt{// Pre-processing}\\
    \quad\\
    \texttt{// Base Case} \\
    \If{$f$ is a constant polynomial}
    {
    $\cL_{mon} = \{f \cdot (X_1 \cdot X_2 \cdots X_n)^d\}$ \\
    \Return{$\cL_{mon}$}
    }
    \quad\\
    Without loss of generality, assume $f$ depends on the variable $X_n$; otherwise, rename the variables appropriately.\\
    Let $s' = s^{4 d^2 \log{n}}$ \\
    Instantiate $(n-1)$-variate generator $\cG(\bZ_4) \coloneq \cG^{KS}_{(s',2d, n-1)}(Z_1,Z_2,Z_3,Z_4)$ from \cref{Thm: Sparse Reconstruction}.\\

    Let $k = \deg_{X_n}(f)$, and find $f_0,\ldots,f_k \in \F[\X_{n-1}]$ such that $f = \displaystyle\sum_{i=0}^k f_i(\X_{n-1})X_n^i$\\

    \texttt{// Normalize $f$}\\
    \If{$||f_k|| \leq s/2$}{
    Compute $\hat{f} = f(\X_{n-1},X_n/f_k)\cdot f_k^{k-1}$ as a circuit \\
    Compute $\hat{f}_\cG = \hat{f}(\cG(\bZ_4), X_n)$ as a sparse polynomial\\
    Compute $\cL_{\hat f_\cG} =$  \{All nonconstant factors of $\hat{f}_\cG$ with $X_n$-degree at most $d$\} using \cref{lem: Constantvariate factorization}\\
    \qquad\\
    \texttt{// Monomial equivalent individual degree $d$ factors of $f_k$}\\
    $\cL_{f_k,mon} = \cref{algo:monomial-equivalent-factors_general sparse}(f_k,s/2,d)$ \\

    \qquad\\
    $\cL_{f_\cG, mon} =  \left\{\hat{g}(\bZ_4, X_n \cdot \left(f_k \circ \cG(\bZ_4)\right))\cdot \frac{h \circ \cG(\bZ_4)}{\left(f_k \circ \cG(\bZ_4)\right)^{\deg_{X_n}(\hat{g})}} \,\middle\vert\, \hat{g} \in \cL_{\hat f_\cG}, h \in \cL_{f_k,mon}\right\} $\\

    \qquad\\
    \texttt{// Use \cref{Thm: Sparse Reconstruction} to reconstruct sparse polynomials below}\\
    $\cL_{mon} = \left\{ \text{Reconstruct } g(\X_{n-1}, X_n) \text{ from } g'(\bZ_4, X_n) \,\middle\vert\, g' \in \cL_{f_\cG, mon} \right\} \cup \cL_{f_k,mon}$

    }
    \qquad \\
    \label{alg:last-1}%
    \end{algorithm}

 \begin{algorithm}[H]
    \addtocounter{algocf}{-1}
    \caption{List of monomial equivalent bounded individual degree factors of general sparse polynomials (possibly with spurious entries) (Continued)}

    \texttt{// Reverse-normalize $f$}\\
    \Else{
    Compute $\tilde f = f(\X_{n-1},X_n f_0)/f_0$ as a circuit \\ %
    
    Compute $\tilde{f}_\cG = \tilde{f}(\cG(\bZ_4), X_n)$ as a sparse polynomial\\
    Compute $\cL_{\tilde f_\cG} =$  \{All nonconstant factors of $\tilde{f}_\cG$ with $X_n$-degree at most $d$\} using \cref{lem: Constantvariate factorization}\\

    \qquad\\
    \texttt{// Monomial equivalent individual degree $d$ factors of $f_0$}\\
    $\cL_{f_0,mon} = \cref{algo:monomial-equivalent-factors_general sparse}(f_0,s/2,d)$ \\

    $\cL_{f_\cG, mon} = \left\{\tilde g(\bZ_4,X_n/\left(f_0\circ \cG(\bZ_4)\right))\cdot h \circ \cG(\bZ_4) \,\middle\vert\, \tilde g \in \cL_{\tilde f_\cG}, h \in \cL_{f_0,mon} \right\} $\\

    \qquad\\
    \texttt{// Use \cref{Thm: Sparse Reconstruction} to reconstruct sparse polynomials below}\\
    $\cL_{mon} = \left\{ \text{Reconstruct } g(\X_{n-1}, X_n) \text{ from } g'(\bZ_4, X_n) \,\middle\vert\, g' \in \cL_{f_\cG, mon} \right\} \cup \cL_{f_0,mon}$

    }

    \qquad \\

    \Return{$\cL_{mon}$}

\end{algorithm}

\begin{algorithm}[H]

    \caption{List of non-monomial bounded individual degree factors of general sparse polynomials (possibly with spurious entries)}
    \label{algo:non-monomial-factors_general_sparse}
    \KwData{$f \in \F[\X_n]$, sparsity parameter $s$, individual degree parameter $d$}
    \KwResult{$\cL \subset \F[\X_n]$, list of circuits over $\F$}

    \quad\\

    $\cL_{mon}$ =  \cref{algo:monomial-equivalent-factors_general sparse}$(f, s, d)$ \\
    $\cL = \emptyset$

    \quad\\

    \texttt{// Removing monomial factors from true factors}\\

    \For{$g \in \cL_{mon}$}{
        Find the largest monomial factor $M$ that divides $g$\\
        $\cL$ = $\cL \cup \left\{ g/M\right\}$
    }

    \Return{$\cL$}
\end{algorithm}

\section{Future Directions and Open Questions}
\label{sec:future-directions-and-open-questions}

Some interesting and natural open problems are the following.

\begin{enumerate}
    \item The most interesting open question is to prune our output lists and return only the true factors. This question arises for both our factorization algorithms. Consider the following problem: given an $s$-sparse polynomial $f$, find all its multilinear factors. In the setting of \cref{thm:main2}, we show that we can output a list of quasipolynomially sparse multilinear polynomials containing all such factors. To prune the list, we would need a divisibility test for checking if a given multilinear quasipolynomially sparse polynomial $g$ divides another sparse polynomial $f$. It can be shown that this can be reduced to PIT for constant-depth circuits, and hence has subexponential-time PIT algorithms. However, it would be interesting to get something more efficient in this case, for instance, a quasipolynomial-time algorithm.
    \item As of now, \cref{thm:main1} only works over fields of characteristic $0$ or characteristic greater than $d$. It would be nice to have the result over all fields.
\end{enumerate}

\textbf{Acknowledgements:} Somnath is grateful to Pranjal Dutta (NTU), Saswata Mukherjee (NUS) and Amit Sinhababu (CMI) for the helpful discussions on polynomial factorization. Shanthanu is grateful to Mrinal Kumar (TIFR) and Pratik Shastri (IMSc) for insightful discussions about sparse polynomials and, in particular, sparse polynomial factorization algorithms.

\bibliographystyle{alpha}
\bibliography{ref.bib}
\appendix
\section{Omitted Proofs from \cref{subsec:resultant_discriminant}}

\label{Appendix:resultant-discriminant}

In this section, we give the omitted proofs of some lemmas from \cref{subsec:resultant_discriminant}.

\begin{proof}[Proof of \cref{lem:resultant-prop}]
    Let $\deg(f)= d_1$ and $\deg(g) = d_2$. Define $h \coloneqq \gcd(f,g)$ and $D = \deg(h)$. By the definition of the Sylvester matrix,
    \begin{equation*}
        \ker(\Syl(f,g)) \cong
        \{(w_1,w_2)\in\F[X]\times\F[X]\mid
        \deg(w_1)<d_2,\ \deg(w_2)<d_1,\ \text{and } w_1f+w_2g=0\}.
    \end{equation*}
    Let $\hat f = f/h$ and $\hat g = g/h$. Then for any $(w_1,w_2) \in \ker(\Syl(f,g))$,
    \begin{equation*}
        w_1 \hat f + w_2 \hat g = 0.
    \end{equation*}
    Since $\gcd(\hat f, \hat g)=1$, it follows that $\hat g \mid w_1$ and $\hat f \mid w_2$. Thus, there is a polynomial $e(X)\in\F[X]$ such that
    \begin{equation*}
        w_1 = \hat g \cdot e \qquad w_2 = -\hat f \cdot e
    \end{equation*}
    The degree constraints imply $\deg(e)<D$, because $\deg(\hat g)=d_2-D$ and $\deg(\hat f)=d_1-D$.

    Conversely, every polynomial $e\in\F[X]$ with $\deg(e)<D$ gives a kernel element $(\hat g e,-\hat f e)$, since $\deg(\hat g e)<d_2$, $\deg(\hat f e)<d_1$, and $(\hat g e)f-(\hat f e)g=0$. Hence
    \[
        \ker(\Syl(f,g)) \cong
        \{(\hat g e,-\hat f e)\in\F[X]\times\F[X]\mid e\in\F[X]\text{ and }\deg(e)<D\}.
    \]
    Therefore $\ker(\Syl(f,g))$ is isomorphic to the vector space of polynomials of degree less than $D$, and hence
    \[
        \dim\ker(\Syl(f,g))=D=\deg(\gcd(f,g)).
    \]
\end{proof}

\begin{proof}[Proof of \cref{cor:resultant-prop}]
    We have
    \[
        \prod_i h_i^{e_i-1} \mid f
        \qquad\text{and}\qquad
        \prod_i h_i^{e_i-1} \mid f',
    \]
    and hence $\prod_i h_i^{e_i-1}\mid \gcd(f,f')$. Therefore
    \begin{equation*}
        \deg(\gcd(f,f')) \geq \sum_i (e_i-1)\deg(h_i).
    \end{equation*}
    By \cref{lem:resultant-prop},
    \begin{equation*}
        \dim\ker(\Syl(f,f')) = \deg(\gcd(f,f')) \geq \sum_i (e_i-1)\deg(h_i).
    \end{equation*}

    Now suppose the $h_i$ are distinct pairwise coprime irreducible polynomials. Since $\Char(\F)=0$ or $\Char(\F)>d$, each $h_i$ is coprime to its derivative. Thus the only common factors of $f$ and $f'$ come from repeated powers of the $h_i$, and
    \begin{equation*}
        \gcd(f,f')=\prod_i h_i^{e_i-1}.
    \end{equation*}
    Applying \cref{lem:resultant-prop} again gives equality.
\end{proof}

\section{Omitted Proofs from \cref{Subsec: Power-series root}}

Here, we add the proof of the version of Hensel's lemma (\cref{Hensel lemma}) we use. We also provide the proof of \cref{Closed form root}.
 \subsection{Proof of Hensel's Lemma}  \label{Proof: Hensel lemma}

\begin{proof}[Proof of \cref{Hensel lemma}]
    Note that it is enough to prove the above theorem for each $h_i(\X,Y)$, as the polynomials $h_i(0,Y)$ are pairwise coprime since $f(\X,Y)$ is Hensel-ready at $\X=0$. Furthermore, each $h_i(0,Y)$ is square-free. We first prove this for $n=1$, i.e., $\X=\{X\}$.

     Let $\deg_Y(h_i)=d$ and $\cB=\{\beta_1,\dots,\beta_d\}$ be the set of $d$ distinct roots of $h_i(0,Y)$ in $\overline{\F}$. Pick any $\beta\in\cB$. We apply Newton iteration (aka Newton-Raphson Method) starting at $\beta$. Suppose, for some $r\geq0$, we have $a_0=\beta,a_1,\dots, a_r\in\F(\beta)$ such that for $a^{(r)}(X):=a_0+a_1X+\dots+a_rX^r$, $h_i(X,a^{(r)}(X))=0\bmod X^{r+1}$. Then we will recursively compute $a_{r+1}$.
     \begin{claim}
         There is a unique $a_{r+1}\in\F(\beta)$ such that, for $a^{(r+1)}(X):=a^{(r)}(X)+a_{r+1}X^{r+1}$, $h_i(X,a^{(r+1)}(X))=0\bmod X^{r+2}$.
     \end{claim}
     \begin{proof}
         Since $h_i(0,Y)$ is square-free, $\partial_Yh_i(0,\beta)\neq0$. Now we have \begin{align*}
             &h_i(X,a^{(r+1)}(X))=0\bmod X^{r+2}\\
             \iff &h_i(X,a^{(r)}(X)+a_{r+1}X^{r+1})=0\bmod X^{r+2}\\
             \iff& h_i(X,a^{(r)}(X))+{a_{r+1}}X^{r+1}{\partial_Yh_i(X,a^{(r)}(X)})=0\bmod X^{r+2}\tag{Taylor expansion around $(X,a^{(r)}(X))$}
         \end{align*}
         Note $h_i(X,a^{(r)}(X))=0\bmod X^{r+1}$, hence $h_i(X,a^{(r)}(X))=cX^{r+1}\bmod X^{r+2}$ for some $c\in\F(\beta)$. Also note that, modulo $X^{r+2}$,
         \[
         {a_{r+1}}X^{r+1}{\partial_Yh_i(X,a^{(r)}(X)})=a_{r+1}\partial_Yh_i(0,\beta)X^{r+1}.
         \]
         Hence we can conclude $h_i(X,a^{(r+1)}(X))=0\bmod X^{r+2}\iff a_{r+1}=-\frac{c}{\partial_Yh_i(0,\beta)}.$
     \end{proof}
     Hence we can set $\Phi_\beta(X)=\sum_ra_rX^r$ and we get the unique power series root with constant term $\beta$. Since $\overline{\F}[[X]][Y]$ is a UFD, we have the unique factorization $h_i(X,Y)=\prod_\beta(Y-\Phi_\beta(X))$. We also remark that there is another nice proof using the fixed-point theorem in \cite{ruiz1993}.

     For the case when $n\geq 2$, introduce a new variable $T$ and consider
     \[
        H_i(T,Y):=h_i(TX_1,\dots,TX_n,Y)\in \F[\X][T,Y].
     \]
     Applying the one-variable argument to $H_i$ as a polynomial in $T$ gives a unique root
     \[
        \Psi_\beta(T)=\sum_{r\geq0} a_r(\X)T^r\in \F(\beta)[\X][[T]]
     \]
     with $\Psi_\beta(0)=\beta$ and $H_i(T,\Psi_\beta(T))=0$. Since the coefficient of $T^r$ in $H_i$ is homogeneous of $\X$-degree $r$, uniqueness implies that each $a_r(\X)$ is homogeneous of degree $r$. Setting $T=1$ gives
     \[
        \Phi_\beta(\X):=\Psi_\beta(1)=\sum_{r\geq0}a_r(\X)\in \F(\beta)[[\X]],
     \]
     and this is the unique power series root of $h_i(\X,Y)$ with constant term $\beta$.

\end{proof}

 \subsection{Proof of \cref{Closed form root}}\label{Proof: closed form root}

Let $F(T,Y)$ be a monic polynomial in $Y$ and Hensel-ready at $T=0$. For a root $\beta$ of $F(0,Y)$ with multiplicity $e$, define $$R_{e}(Z):=Z+\sum_{m\geq 0}^{2e(D+1)}[Y^{em+e-2}]\left\{\left(\frac{e!}{\partial_{Y^{e}}F(0,Z)}\right)^{m+1}\frac{\partial_{Y}F(T,Y+Z)}{e}\left(\frac{\partial_{Y^{e}}F(0,Z)}{e!}Y^e-F(T,Y+Z)\right)^m\right\} $$
 Here $[Y^r]\{P(Y)\}$ denotes the coefficient of $Y^r$ in $P(Y)$.

 For a bivariate power series $\varphi(X,Y)$, define the diagonal operator $$\cD(\varphi)(Z):=\sum_{r\geq0}([X^rY^r]\{\varphi(X,Y)\})Z^r.$$

 We will use the following consequence of Furstenberg's theorem\footnote{\cite{Furstenberg67} gave the exact form when $e=1$}.
\begin{lemma}[Theorem 3.1 in \cite{BKRRSS25a}]
    Assume that $0$ is a root of $F(0,Y)$ of multiplicity $e$, and let $\varphi(T)$ be the unique power series root of $F(T,Y)$ corresponding to $0$. Then
    $$\varphi(T)=\cD\Big(\frac{Y^2 \cdot \partial_YF(TY,Y)}{e \cdot F(TY,Y)}\Big).$$

\end{lemma}

Before we prove~\cref{Closed form root}, we first show how to express the above diagonal expression as a power series and do some initial calculations which will be helpful in the proof.

Now, since $0$ is a root of $F(0,Y)$ with multiplicity $e$, we have
\[
c \;:=\; \frac{\partial_{Y^e}F(0,0)}{e!}\neq 0.
\]
Moreover, $1-F(TY,Y)/(cY^e)$ has zero constant coefficient as a Laurent series in $Y$ whose coefficients are power series in $T$, so the following geometric expansion is valid coefficient-wise:
\[
F(TY,Y)
=
c\,Y^{e}\!\left[1-\left(1-\frac{F(TY,Y)}{c\,Y^{e}}\right)\right]
\]
we use $\displaystyle\frac{1}{1-z} = \sum_{m \geq 0} z^{m}$ with $z = 1 - F(TY,Y)/(c\,Y^{e})$. This gives
\begin{align*}
\varphi(T)
&\;=\; \mathcal{D}\!\left[\frac{Y^{2}\,\partial_{Y}F(TY,Y)}{e\,\cdot F(TY,Y)}\right] \\[6pt]
&\;=\; \mathcal{D}\!\left[\frac{Y\,\partial_{Y}F(TY,Y)\,/\,Y^{e-1}}{e\,\cdot c\bigl[1 - (1-F(TY,Y)/(c\,\cdot Y^{e}))\bigr]}\right] \\[6pt]
&\;=\; \mathcal{D}\!\left[\frac{Y\,\partial_{Y}F(TY,Y)}{Y^{e-1}\,\cdot e\,\cdot c}\;\sum_{m\geq 0}\!\Bigl(1 - \tfrac{F(TY,Y)}{c\,\cdot Y^{e}}\Bigr)^{\!m}\right] \\[6pt]
&\;=\; \mathcal{D}\!\left[\sum_{m\geq 0}\!\Bigl(1 - \tfrac{F(TY,Y)}{c\,\cdot Y^{e}}\Bigr)^{\!m}\,\frac{Y\,\partial_{Y}F(TY,Y)}{Y^{e-1}\,\cdot e\, \cdot c}\right].
\end{align*}

For $m\geq 0$, expand
\[
\Bigl(1 - \tfrac{F(TY,Y)}{c\,\cdot Y^{e}}\Bigr)^{\!m}
\;=\; \frac{\bigl(c\,\cdot Y^{e}-F(TY,Y)\bigr)^{m}}{c^{m}\,\cdot Y^{em}},
\]
so the $m$-th summand becomes
\[
\frac{\bigl(c\,\cdot Y^{e}-F(TY,Y)\bigr)^{m}\,\partial_{Y}F(TY,Y)}{e\,\cdot c^{\,m+1}\,\cdot Y^{em+e-2}}.
\]

Therefore,
\begin{align*}
[T^{n}](\varphi)
&\;=\; [T^{n}Y^{n}]\!\left\{\sum_{m\geq 0}\!\Bigl(1 - \tfrac{F(TY,Y)}{c \,\cdot Y^{e}}\Bigr)^{\!m}\,\frac{Y\,\partial_{Y}F(TY,Y)}{Y^{e-1}\cdot\,e\cdot\,c}\right\} \\[6pt]
&\;=\; \frac{1}{e}\sum_{m\geq 0}\frac{1}{c^{m+1}}\,[T^{n}Y^{\,n+em+e-2}]\Bigl\{\bigl(c\,\cdot Y^{e}-F(TY,Y)\bigr)^{m}\,\partial_{Y}F(TY,Y)\Bigr\}.
\end{align*}

Finally, for any Laurent series $L(T,Y)$, we have $[{T^i Y^j}]\{L(T,Y)\} = [{T^i Y^{i+j}}]\{L(TY, Y)\}$. Hence, we have 
\begin{align*}
    \varphi(T)\;&=\; \frac{1}{e}\,\sum_{n \geq 0}\sum_{m\geq 0}\frac{1}{c^{m+1}}\,[T^nY^{\,em+e-2}]\Bigl\{\bigl(c\,\cdot Y^{e}-F(T,Y)\bigr)^{m}\,\partial_{Y}F(T,Y)\Bigr\} \cdot T^n\\
    &= \frac{1}{e}\,\sum_{m\geq 0}\frac{1}{c^{m+1}}\,[Y^{\,em+e-2}]\Bigl\{\bigl(c\,\cdot Y^{e}-F(T,Y)\bigr)^{m}\,\partial_{Y}F(T,Y)\Bigr\}
\end{align*}

    Finally setting $c=\displaystyle\frac{\partial_{Y^e}F(0,0)}{e!}$ we get $$\varphi(T)=\sum_{m\geq 0}[Y^{em+e-2}]\left\{\left(\frac{e!}{\partial_{Y^{e}}F(0,0)}\right)^{m+1}\frac{\partial_{Y}F(T,Y)}{e} \cdot \left(\frac{\partial_{Y^{e}}F(0,0)}{e!}Y^e-F(T,Y)\right)^m\right\} $$

We now justify the finite cutoff used in the definition of $R_e$. Let
$Q(T,Y):=cY^e-F(T,Y)$. Since $0$ is a root of $F(0,Y)$ of multiplicity $e$,
every $T^0$-term of $Q$ has $Y$-degree at least $e+1$, while every remaining
term of $Q$ has positive $T$-degree. Consider a monomial contributing to
\[
[Y^{em+e-2}]\{Q(T,Y)^m\partial_YF(T,Y)\}.
\]
If $s$ of the $m$ factors of $Q$ contribute positive $T$-degree terms,
then the other $m-s$ factors contribute $Y$-degree at least $e+1$. Hence, for
the total $Y$-degree to be at most $em+e-2$, we must have
\[
(m-s)(e+1)\leq em+e-2,
\]
and therefore $s\geq (m-e+2)/(e+1)$. Thus, if the corresponding term has
$T$-degree at most $D$, then certainly $s\leq D$ since each of those $s$ factors
already contributes at least one to the $T$-degree. This implies
\[
m\leq (e+1)D+e-2 \leq 2e(D+1).
\]
Consequently, terms with $m>2e(D+1)$ do not contribute to the truncation
modulo $T^{D+1}$.

Note for the validity of the entire calculation above we need $e!\neq0$, for this purpose we will use the fact that $\Char(\F)=0$ or greater than $d$ and clearly $d\geq e$. Now we are ready to prove \cref{Closed form root}
\begin{proof}[Proof of \cref{Closed form root}]
Item (1) follows from the definition of $R_e(Z)$ and the fact that truncation and coefficient extraction can be carried out in constant depth (which follows from \cref{lem:circuit-truncation} and \cref{lem:circuit-coefficient-interpolation}). Moreover, the denominator $B(Z)$ can be chosen to be an appropriate power of $\partial_{Y^e}F(0,Z)$, and hence $B(\beta)\neq0$ because $\beta$ has multiplicity exactly $e$ in $F(0,Y)$. This proves item (2). We now prove item (3).

    Let $\hat{F}(T,Y)=F(T,Y+\beta)$. Then $\hat{F}(0,Y)$ has a root $0$ with multiplicity $e$, and the unique power series root corresponding to $0$ is $\hat{\varphi}(T):=\Phi_\beta(T)-\beta$.

    Now from the above calculation we have \begin{align*}
        &\Phi_\beta(T)\Trunc T^{D+1}\\
        =&\hat{\varphi}(T)\Trunc T^{D+1}+\beta\\
        =&\sum_{m\geq 0}[Y^{em+e-2}]\Big\{\Big(\frac{e!}{\partial_{Y^{e}}\hat{F}(0,0)}\Big)^{m+1}\frac{\partial_{Y}\hat{F}(T,Y)}{e}\Big(\frac{\partial_{Y^{e}}\hat{F}(0,0)}{e!}Y^e-\hat{F}(T,Y)\Big)^m\Big\}\Trunc T^{D+1}+\beta\\
        =&\beta+\sum_{m\geq 0}[Y^{em+e-2}]\Big\{\Big(\frac{e!}{\partial_{Y^{e}}{F}(0,\beta)}\Big)^{m+1}\frac{\partial_{Y}{F}(T,Y+\beta)}{e}\Big(\frac{\partial_{Y^{e}}{F}(0,\beta)}{e!}Y^e-{F}(T,Y+\beta)\Big)^m\Big\}\Trunc T^{D+1}\\
        =&\beta+\sum_{m\geq 0}^{2e(D+1)}[Y^{em+e-2}]\Big\{\Big(\frac{e!}{\partial_{Y^{e}}{F}(0,\beta)}\Big)^{m+1}\frac{\partial_{Y}{F}(T,Y+\beta)}{e}\Big(\frac{\partial_{Y^{e}}{F}(0,\beta)}{e!}Y^e-{F}(T,Y+\beta)\Big)^m\Big\}\Trunc T^{D+1}\\
        =&R_e(\beta)\Trunc T^{D+1}
    \end{align*}
\end{proof}

\section{Omitted Proofs from \cref{sec:sparse-polynomials}}
\label{app:sparse-poly}

We give the omitted proofs from \cref{sec:sparse-polynomials} below.

\begin{proof}[Proof of \cref{cor:factor_sparse_monomial_division}]
    It will be convenient to introduce the following notation. Let
    $P(n,s,d)$ denote the class of all polynomials $f \in \F[\X_n]$ such that
    $f$ divides a non-zero $s$-sparse polynomial in $\F[\X_n]$ of individual
    degree at most $d$.

    Since $f \in P(n,s,d)$ and $f \prec_d g$, the individual degree of $g$ is at
    most $2d$. Write
    $g(\X_n)=\sum_{j=0}^{2d}
    g_j(X_1,\dots,X_{i-1},X_{i+1},\dots,X_n)X_i^j$.
    By \cref{lem:circuit-coefficient-interpolation}, we can construct, in time
    $\poly(m,d)$, circuits $C_0,\dots,C_{2d}$ computing
    $g_0,\dots,g_{2d}$ respectively.

    Let $e$ be the largest power of $X_i$ dividing $g$. Then
    $g_0=g_1=\cdots=g_{e-1}=0$ and $g_e\neq 0$. Moreover, since
    $f \prec_d g$ and $f$ is not divisible by a monomial, $g_e$ is a monomial
    multiple of $f|_{X_i=0}$. We claim that $g_e\in P(n-1,s,2d)$.
    Indeed, since $f \in P(n,s,d)$, there is a non-zero $s$-sparse polynomial
    $F \in \F[\X_n]$ of individual degree at most $d$ such that $f \mid F$.
    Write $F=fH$, and let $H_r$ be the first non-zero coefficient of $H$ when
    viewed as a polynomial in $X_i$. Since $f$ is not divisible by $X_i$,
    $f|_{X_i=0}$ is non-zero, and the coefficient of $X_i^r$ in $F$ is
    $f|_{X_i=0}\cdot H_r$. This coefficient is a non-zero $(n-1)$-variate
    $s$-sparse polynomial of individual degree at most $d$ divisible by
    $f|_{X_i=0}$. Since $g_e$ differs from $f|_{X_i=0}$ by a monomial of
    individual degree at most $d$, it divides a non-zero $s$-sparse polynomial
    of individual degree at most $2d$.

    We now test the coefficient polynomials in order. Let
    $S_{n-1,s,2d}$ be the Klivans--Spielman hitting set from
    \cref{thm:sparse_pit}. For each $j=0,\dots,2d$, we evaluate $C_j$ on all
    points of $S_{n-1,s,2d}$. If all these evaluations are zero, we continue to
    the next value of $j$; otherwise, we output $j$.

    This procedure outputs $e$. Indeed, for every $j<e$, we have $g_j=0$, so
    all evaluations of $g_j$ vanish. On the other hand, $g_e$ is a
    non-zero polynomial in $P(n-1,s,2d)$, and therefore
    \cref{thm:sparse_pit} implies that $g_e$ is non-zero
    on some point of $S_{n-1,s,2d}$. Thus $e$ is the first index for which the test
    finds a non-zero evaluation.

    The running time is $\poly(n, s, d, m)$: there are at most
    $2d+1$ coefficient polynomials to test, the circuits for them are obtained
    in time $\poly(m,d)$, and each hitting-set construction and evaluation
    uses only $\poly(n,s,d,m)$ time by \cref{thm:sparse_pit}. Hence
    the largest power of $X_i$ dividing $g$ can be computed deterministically in
    $\poly(n,s,d,D,m)$ time.
\end{proof}

\begin{proof}[Proof of \cref{lem:division-elimination}]
    We use Strassen's division elimination~\cite{Strassen73}, together with the
    Klivans--Spielman hitting set from \cref{thm:sparse_pit}, to first
    reduce to the case where the denominator has a non-zero constant term.

    By the standard first step in Strassen's argument, we may move all divisions
    to the top. Thus, in deterministic $\poly(m,d)$ time, we obtain
    division-free circuits $C_1$ and $C_2$ of size $\poly(m,d)$ and depth at most
    $\Delta+\cO(1)$ such that $C$ computes $C_1/C_2$. Moreover, $C_2$ is a
    product of the denominator polynomials appearing in the division gates of
    $C$. In particular, by assumption, $C_2$ is a product of $s$-sparse
    polynomials of individual degree at most $d$.

    We now find a shift $\balpha \in \F^n$ such that $C_2(\balpha)\neq 0$.
    If $C_2=Q_1\cdots Q_t$ is the product decomposition of $C_2$ into
    sparse denominator factors, then $t\leq m$, and $C_2$ belongs to the class
    of products of at most $m$ many $s$-sparse polynomials of individual degree
    at most $d$. By \cref{thm:sparse_pit}, the hitting set
    $S_{n,s,d,m}$ contains a point $\balpha$ such that $C_2(\balpha)\neq0$.

    Consider the shifted quotient
    $C'(\X_n)=C(\X_n+\balpha)=C_1(\X_n+\balpha)/C_2(\X_n+\balpha)$. Its
    denominator has non-zero constant term, since
    $C_2(\balpha)\neq 0$. Therefore, after normalizing the denominator by its
    constant term, we may write it as $1+h(\X_n)$, where $h$ has zero constant
    term. Since $C'$ computes the polynomial $f(\X_n+\balpha)$, whose total
    degree is at most $nd$, Strassen's truncation argument applies: the inverse
    of $1+h$ is needed only up to total degree $nd$, and hence can be replaced by
    the truncated geometric series
    $1-h+h^2-\cdots+(-h)^{nd}$. Using \cref{lem:circuit-truncation}, this
    truncated product, namely $C_1(\X_n+\balpha) \cdot (1-h+h^2-\cdots+(-h)^{nd})$, can be computed by a division-free circuit of size
    $\poly(m,d)$ and depth increased by only $\cO(1)$. Thus we obtain a
    division-free circuit for $f(\X_n+\balpha)$.

    Finally, substituting $\X_n-\balpha$ for $\X_n$ in this division-free
    circuit gives a division-free circuit computing $f(\X_n)$. This substitution
    increases the size only polynomially and does not increase the depth by more
    than a constant. Therefore the resulting circuit has size $\poly(m,d)$ and
    depth at most $\Delta+\cO(1)$, and the whole construction runs in
    $\poly(n,s,m,d)$ time.
\end{proof}

\end{document}